\documentclass[10pt, reqno]{amsart}
\usepackage{amsthm,amsfonts,amssymb,euscript,cite,enumerate, mathrsfs}
\usepackage{graphics}
\usepackage{hyperref}

\usepackage{seqsplit}

\usepackage{xcolor} 

\usepackage{hyperref}

\usepackage{graphicx}
\usepackage{color}
\usepackage{transparent}

\usepackage{fullpage} 

\definecolor{green}{rgb}{0,0.8,0} 

\newtheorem{theorem}{Theorem}[section]
\newtheorem{corollary}[theorem]{Corollary}
\newtheorem{lemma}[theorem]{Lemma}
\newtheorem{proposition}[theorem]{Proposition}
\theoremstyle{definition}
\newtheorem{definition}[theorem]{Definition}
\newtheorem{example}[theorem]{Example}

\theoremstyle{remark}
\newtheorem{remark}[theorem]{Remark}
\newtheorem{conjecture}[theorem]{Conjecture}
\numberwithin{equation}{section}
\newcommand{\nrm}[1]{\Vert#1\Vert}
\newcommand{\abs}[1]{\vert#1\vert}
\newcommand{\brk}[1]{\langle#1\rangle}
\newcommand{\set}[1]{\{#1\}}

\newcommand{\ep}{\epsilon}
\def\beaa{\begin{eqnarray*}}
\def\eeaa{\end{eqnarray*}}
\def\bea{\begin{eqnarray}}
\def\eea{\end{eqnarray}}
\def\be{\begin{equation}}
\def\ee{\end{equation}}

\newcommand{\ud}{d}
\newcommand{\rd}{\partial}

\newcommand{\bb}{\Big}

\newcommand{\0}{\emptyset}

\newcommand{\alp}{\alpha}
\newcommand{\bt}{\beta}
\newcommand{\gmm}{\gamma}
\newcommand{\Gmm}{\Gamma}

\newcommand{\eps}{\epsilon}
\newcommand{\veps}{\varepsilon}
\newcommand{\kpp}{\kappa}
\newcommand{\lmb}{\lambda}

\newcommand{\sgm}{\sigma}
\newcommand{\Sgm}{\Sigma}

\newcommand{\Tht}{\Theta}

\newcommand{\omg}{\omega}
\newcommand{\om}{\omega}
\newcommand{\Omg}{\Omega}

\newcommand{\bfe}{{\bf e}}

\newcommand{\bbN}{\mathbb N}

\newcommand{\bbR}{\mathbb R}
\newcommand{\bbS}{\mathbb S}

\newcommand{\calA}{\mathcal A}

\newcommand{\calC}{\mathcal C}
\newcommand{\calD}{\mathcal D}

\newcommand{\calG}{\mathcal G}
\newcommand{\calH}{\mathcal H}
\newcommand{\calI}{\mathcal I}

\newcommand{\calL}{\mathcal L}
\newcommand{\calM}{\mathcal M}

\newcommand{\calQ}{\mathcal Q}
\newcommand{\calR}{\mathcal R}
\newcommand{\calS}{\mathcal S}

\newcommand{\calU}{\mathcal U}

\newcommand{\uC}{\underline{C}}

\newcommand{\dur}{\nu}
\newcommand{\dvr}{\lmb}

\newcommand{\f}{\frac}
\newcommand{\de}{\delta}
\newcommand{\nab}{\nabla}

\newcommand{\CH}{\calC \calH^{+}}		
\newcommand{\EH}{\calH^{+}}
\newcommand{\pfstep}[1]{\vspace{.5em} {\it \noindent #1.}}

\newcommand{\e}{\bfe}							

\newcommand{\NI}{\calI^{+}}						

\newcommand{\phibg}{\overline{\phi}}				
\newcommand{\rbg}{\overline{r}}					
\newcommand{\ebg}{\overline{\e}}					

\newcommand{\fbg}{\overline{f}}
\newcommand{\hbg}{\overline{h}}
\newcommand{\ellbg}{\overline{\ell}}
\newcommand{\dphibg}{\overline{\dot{\phi}}}

\newcommand{\ub}{\underline{u}}

\newcommand{\Int}{\mathscr{B}}
\newcommand{\Ext}{\mathscr{E}}

\newcommand{\Abs}[1]{\left\vert #1 \right\vert}

\begin{document}

\title[]{Strong cosmic censorship in spherical symmetry for two-ended asymptotically flat initial data I. The interior of the black hole region}
\author{Jonathan Luk}
\address{Department of Mathematics, Stanford University, Palo Alto, CA, USA}
\email{jluk@stanford.edu}

\author{Sung-Jin Oh}
\address{Korea Institute for Advanced Study, Seoul, Korea}
\email{sjoh@kias.re.kr}


\begin{abstract}
This is the first and main paper of a two-part series, in which we prove the $C^{2}$-formulation of the strong cosmic censorship conjecture for the Einstein--Maxwell--(real)--scalar--field system in spherical symmetry for two-ended asymptotically flat data. For this model, it is known through the works of Dafermos and Dafermos--Rodnianski that the maximal globally hyperbolic future development of any admissible two-ended asymptotically flat Cauchy initial data set possesses a non-empty Cauchy horizon, across which the spacetime is $C^{0}$-future-extendible (in particular, the $C^{0}$-formulation of the strong cosmic censorship conjecture is false). Nevertheless, the main conclusion of the present series of papers is that for a generic (in the sense of being open and dense relative to appropriate topologies) class of such data, the spacetime is future-inextendible with a Lorentzian metric of higher regularity (specifically, $C^{2}$).

In this paper, we prove that the solution is $C^{2}$-future-inextendible under the condition that the scalar field obeys an $L^{2}$-averaged polynomial lower bound along each of the event horizons.
This, in particular, improves upon a previous result of Dafermos, which required instead a pointwise lower bound. 
Key to the proof are appropriate stability and instability results in the interior of the black hole region, whose proofs are in turn based on ideas from the work of Dafermos--Luk on the stability of the Kerr Cauchy horizon (without symmetry) and from our previous paper on linear instability of Reissner--Nordstr\"om Cauchy horizon. 
In the second paper of the series \cite{LO.exterior}, which concerns analysis in the exterior of the black hole region, we show that the $L^2$-averaged polynomial lower bound needed for the instability result indeed holds for a generic class of admissible two-ended asymptotically flat Cauchy initial data. 
\end{abstract}

\maketitle

\tableofcontents

\section{Introduction}
This is the first of a series of two papers in which we prove the $C^2$-formulation of the strong cosmic censorship conjecture for the Einstein--Maxwell--(real)--scalar--field system in spherical symmetry for two-ended asymptotically flat initial data on $\mathbb R\times \mathbb S^2$. A solution to the Einstein--Maxwell--(real)--scalar--field system consists of $(\mathcal M,g,\phi,F)$, where $\mathcal M$ is a $4$-dimensional manifold, $g$ is a Lorentzian metric on $\mathcal M$, $\phi$ is a real-valued function on $\mathcal M$ and $F$ is a $2$-form on $\mathcal M$. The system of equations is given by
\begin{equation}\label{EMSFS}
\begin{cases}
Ric_{\mu\nu}-\f12 g_{\mu\nu} R=2(T^{(sf)}_{\mu\nu}+T^{(em)}_{\mu\nu}),\\
T^{(sf)}_{\mu\nu}=\rd_\mu\phi\rd_\nu\phi-\f 12 g_{\mu\nu} (g^{-1})^{\alp\beta}\rd_\alp\phi\rd_{\beta}\phi,\\
T^{(em)}_{\mu\nu}=(g^{-1})^{\alp\bt}F_{\mu\alp}F_{\nu\bt}-\f 14 g_{\mu\nu}(g^{-1})^{\alp\bt}(g^{-1})^{\gamma\sigma}F_{\alp\gamma}F_{\bt\sigma},
\end{cases}
\end{equation}
where $\phi$ and $F$ satisfy
$$\Box_g\phi=0,\,\quad dF=0,\,\quad (g^{-1})^{\alpha\mu}\nab_\alpha F_{\mu\nu}=0.$$
Here, $\Box_g$ and $\nab$ respectively denote the Laplace--Beltrami operator and the Levi--Civita connection associated to the metric $g$. 

An explicit spherically symmetric solution to this system with a vanishing scalar field is the Reissner--Nordstr\"om spacetime, whose metric is given in local coordinates by\footnote{We will use boldface ${\bf e}$ for the charge in this paper and reserve the notation $e$ for the Euler number.}
\begin{equation}\label{RN.metric}
g=-\left(1-\f {2M} r+\f{{\bf e}^2}{r^2}\right)dt^2+\left(1-\f {2M} r+\f{{\bf e}^2}{r^2}\right)^{-1} dr^2+r^2 d\sigma_{\mathbb S^2},
\end{equation}
where $d\sigma_{\mathbb S^2}$ denotes the standard round metric on the $2$-sphere with radius $1$ and the Maxwell field $F$ is given by\footnote{We remark on the well-known fact that there are different choices of (spherically symmetric) Maxwell field which together with the metric \eqref{RN.metric} solve the Einstein--Maxwell equations. Here, we are only giving one such example.} $F=\f{2{\bf e}}{r^2} \,dt\wedge dr$. Here, $M$ and ${\bf e}$ are real-valued constants representing the mass and the charge of the spacetime respectively.

The Reissner--Nordstr\"om spacetime is said to be subextremal with non-vanishing charge if $0<|{\bf e}|<M$. In this case, the maximal globally hyperbolic future development of Reissner--Nordstr\"om initial data has a Penrose diagram\footnote{In this work, we make an extensive use of Penrose diagrams to represent global causal properties of the spacetime. For an introduction, see \cite[Appendix~C]{DRPL}.} given by Figure \ref{fig:RN}. In particular, the spacetime possesses a smooth Cauchy horizon $\mathcal C\mathcal H^+$ and the maximal globally hyperbolic future development can be future-extended smoothly but non-uniquely as solutions to the Einstein--Maxwell system. 

\begin{figure}[h]
\begin{center}
\def\svgwidth{200px}
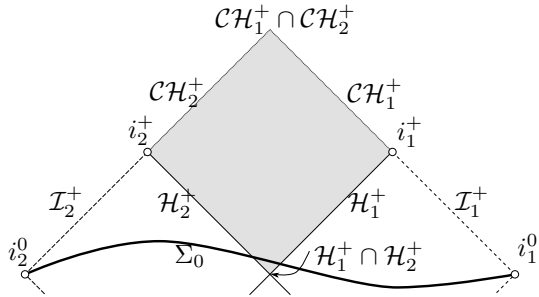 
\caption{Penrose diagram of Reissner--Nordstr\"om spacetime. For the notation, we refer to Theorem~\ref{thm:kommemi}.} \label{fig:RN}
\end{center}
\end{figure}

A priori, the non-uniqueness of possible future-extensions of the Reissner--Nordstr\"om spacetime challenges the deterministic nature of Einstein's theory. Nevertheless, it is widely expected that such a feature of the Reissner--Nordstr\"om spacetime is \underline{non-generic}. This non-genericity, understood from the point of view of the initial value problem, would follow from the celebrated strong cosmic censorship of Penrose \cite{Penrose.SCC}. Since we will be dealing with the Einstein--Maxwell--scalar field system in this paper, we state the strong cosmic censorship conjecture in the following form\footnote{One can of course also entertain the strong cosmic censorship conjecture for the Einstein--Maxwell system or the Einstein vacuum system. The conjecture can also be formulated in situations where the initial data are not asymptotically flat, for instance in cosmological settings; see for instance \cite{Ringstrom.book}. We will not discuss the conjecture in such generality in this paper.}:
\begin{conjecture}[Strong cosmic censorship conjecture]
Maximal globally hyperbolic future developments for the Einstein--Maxwell--(real)--scalar field system to generic asymptotically flat initial data are future-inextendible as suitably regular Lorentzian manifolds.
\end{conjecture}
The strong cosmic censorship conjecture as stated above is not precise regarding the notion of regularity for the extension. One may entertain the following slightly more precise formulation of the conjecture:
\begin{conjecture}[$C^k$-formulation of the strong cosmic censorship conjecture]\label{conj.Ck}
Maximal globally hyperbolic future developments for the Einstein--Maxwell--(real)--scalar field system corresponding to generic asymptotically flat initial data are future-inextendible as time-oriented Lorentzian manifolds with $C^k$ metrics.
\end{conjecture}
Our purpose in this paper and \cite{LO.exterior} is to understand Conjecture \ref{conj.Ck} for generic \underline{spherically symmetric} solutions. However, when discussing the strong cosmic censorship conjecture under the assumption of spherical symmetry, one should be reminded that the set of spherically symmetric solutions is only a very small subset of solutions and their behavior may not be representative of the ``generic'' phenomena in general without symmetry assumptions. Nevertheless, in view of \cite{DL}, we hope that some of the methods we develop in this paper may be relevant for Conjecture \ref{conj.Ck} in some settings without symmetry assumptions.

In view of the fact that the explicit Schwarzschild solution is inextendible to a larger Lorentzian manifold with a continuous metric \cite{Sbie.C0}, one may conjecture a very strong form of the strong cosmic censorship conjecture, namely, the $C^0$-formulation of Conjecture~\ref{conj.Ck} \cite{Chr:CQG}. This would be consistent with the expectation, which is common in the physics literature, that the ``tidal deformation becomes infinite'' in the interior of black holes. Indeed, for spherically symmetric solutions such that in addition the Maxwell field is assumed to vanish, i.e., for solutions to the Einstein--(real)--scalar field system in spherical symmetry, the $C^0$-formulation of the strong cosmic censorship conjecture was proven by Christodoulou:
\begin{theorem}[Christodoulou \cite{Christodoulou:1991yfa, Chr.instab}]
The \underline{$C^0$-formulation} of the strong cosmic censorship conjecture for the Einstein--(real)--scalar field system in \underline{spherical symmetry} with either $1$-ended asymptotically flat initial data on $\mathbb R^3$ or $2$-ended asymptotically flat initial data on $\mathbb R\times \mathbb S^2$ is \underline{true}.
\end{theorem}
Nevertheless, Dafermos--Rodnianski showed that as long as the charge is non-zero\footnote{As we will see in Section \ref{sec.SS}, the charge is a constant for solutions to the Einstein--Maxwell--(real)--scalar field system in spherical symmetry. It therefore makes sense to discuss solutions with non-vanishing charge.} \footnote{Notice that the Schwarzschild solution is \underline{not} a solution to the Einstein--Maxwell--(real)--scalar field system in the case where the charge is required to be non-vanishing.}, all solutions arising from a suitable class of initial data are extendible with a $C^0$ metric. Hence, if one views the non-vanishing of charge as a ``generic'' condition, this implies that the $C^0$-formulation of strong cosmic censorship is false for the Einstein--\underline{Maxwell}--(real)--scalar field system in spherical symmetry:
\begin{theorem}[Dafermos \cite{D2}, Dafermos--Rodnianski \cite{DRPL}]\label{DDR}
The \underline{$C^0$-formulation} of the strong cosmic censorship conjecture for the Einstein--\underline{Maxwell}--(real)--scalar field system in \underline{spherical symmetry} with $2$-ended asymptotically flat initial data on $\mathbb R\times \mathbb S^2$ is \underline{false}.
\end{theorem}
Our main result in this series of papers is that the $C^2$-formulation of the strong cosmic censorship conjecture remains true in this setting:
\begin{theorem}[Main theorem, rough version]\label{main.theorem.both}
The \underline{$C^2$-formulation} of the strong cosmic censorship conjecture for the Einstein--\underline{Maxwell}--(real)--scalar field system in \underline{spherical symmetry} with $2$-ended asymptotically flat initial data on $\mathbb R\times \mathbb S^2$ is \underline{true}.
\end{theorem}
In view of the above discussion, one only needs to understand the case where the charge is non-vanishing. We will make this assumption from now on.

\begin{remark}[$W^{1,2}_{loc}$ formulation of the strong cosmic censorship]
Theorems \ref{DDR} and \ref{main.theorem.both} leave open the question of the validity of some ``intermediate'' formulation of the strong cosmic censorship conjecture (for instance, the $C^1$ formulation of the strong cosmic censorship). At the same time, it is also of interest to consider a formulation of the strong cosmic censorship not in the class of $C^k$ metrics, but in terms of $W^{1,2}_{loc}$ metrics. This is particularly relevant to the problem of determinism since $C^0\cap W^{1,2}_{loc}$ is the minimal known requirement to define weak solutions to the Einstein--Maxwell--(real)--scalar--field system; see the discussions in the introduction of \cite{Chr}. 

In support of the $W^{1,2}_{loc}$ formulation of the conjecture, we show in Appendix~\ref{sec:christoffel} that in a particular $C^{0}$ extension of a generic solution whose existence is asserted by Theorem~\ref{DDR} (more precisely, see Theorems~\ref{main.theorem.C0.stability} and \ref{thm.nonpert} below), the Christoffel symbols, as well as the gradient of the scalar field, fail to be locally square-integrable\footnote{In fact, we prove a stronger blow up result for the Christoffel symbols in this extension, namely that they are not in $L^{p}_{loc}$ for all $p>1$, while for the gradient of the scalar field we only show the failure of the $L^{2}_{loc}$ condition.}. The $W^{1,2}_{loc}$ formulation of the strong cosmic censorship conjecture would follow if such blow up statements may be generalized to arbitrary nontrivial $C^{0}$ extensions. We do not pursue this issue in this paper, which remains an open problem.
\end{remark}

We will give a precise statement of Theorem~\ref{main.theorem.both} in Section \ref{sec.SCC}. In particular, we will define the notion of genericity, which very roughly is to be understood as being open in a (weighted) $C^1$ topology and being dense\footnote{In fact, a stronger statement is proven: for any element in the complement of the generic set, there exists a continuous (with respect to a weighted-$C^\infty$-topology) $1$-parameter family of initial data sets passing through it such that all other elements of the $1$-parameter family belong to the generic set.} in a (weighted) $C^\infty$ topology. The theorem will then be proven in Section~\ref{sec.pf.SCC}, using some results that are proven in the later parts of the paper, as well as some results which are proven in the companion paper \cite{LO.exterior}. We refer the reader to Section~\ref{sec.main.structure} below for the main elements of the proof. Theorem \ref{main.theorem.both} in particular implies that the smooth Cauchy horizon of Reissner--Nordstr\"om is unstable. In the special case of small perturbations of Reissner--Nordstr\"om, we in fact have more precise information regarding the maximal globally hyperbolic development, see Section~\ref{sec.instab.RN}.

Prior to the present paper, the best known result regarding the validity of the $C^2$-formulation of the strong cosmic censorship for this model in spherical symmetry was achieved in the seminal work of Dafermos \cite{D2} (see also \cite{D1, D3}), who proved a \emph{conditional} $C^2$-future-inextendibility result\footnote{In \cite{D2}, $C^1$-future-inextendibility within spherical symmetry was proven, but $C^2$-future-inextendibility without symmetry was not explicitly shown. This however follows easily from the mass inflation result.}. The required condition, however, remains difficult to verify\footnote{In fact, it is not known whether there exists a single regular solution such that this condition is satisfied.}. Part of our proof is to obtain a new and stronger conditional inextendibility result, so that we show moreover that the condition is satisfied for a generic set of data. See Section~\ref{sec.comparison} for further discussions.

We remark at this point that the class of $2$-ended initial data that we consider in this paper and \cite{LO.exterior} can be easily shown to have a complete future null infinity (with two connected components) \cite{DafTrapped}. Indeed, this is the main simplification that arises from studying the $2$-ended case so that we in particular do not need to handle potential singularities at the center of symmetry. For a more ``realistic'' model\footnote{Notice that there are \underline{no} regular $1$-ended solutions to the Einstein--Maxwell--(real)--scalar field system with data on $\mathbb R^3$ in the presence of charge.}, one may for instance study the Einstein--Maxwell--(charged)--scalar--field system with spherically symmetric initial data posed on $\mathbb R^3$, for which the strong cosmic censorship conjecture remains an open problem. In that case, the full resolution of the strong cosmic censorship conjecture seems to at least require an understanding of singularities arising at the centers of symmetry as well as Cauchy horizons emanating from them \cite{Kommemi}.

The remainder of the introduction will be structured as follows:
\begin{itemize}
\item In Section~\ref{sec.main.structure}, we explain the overall structure of the proof of Theorem~\ref{main.theorem.both}. The reader is encouraged to take this as a guide to our series of papers.
\item In Section~\ref{sec.previous.works}, we give a brief overview of some of the relevant previous results. 
In particular, we provide comparison of our proof with that of linear instability of Reissner--Nordstr\"om Cauchy horizon in \cite{LO.instab}, and with the possible alternative approach using the conditional instability theorem of Dafermos \cite{D2}.
\item Finally, we end the introduction with an outline of the remainder of the paper in Section~\ref{sec.outline}.
\end{itemize}

\subsection{Guide to the series: Structure of the proof of Theorem~\ref{main.theorem.both}}\label{sec.main.structure}

In very rough terms, the proof of Theorem \ref{main.theorem.both} proceeds by first showing that the maximal globally hyperbolic future developments of any ``admissible'' data approach Reissner--Nordstr\"om in a certain sense, and then using the ideas in \cite{LO.instab} (which were originally for linear instability on \emph{fixed} Reissner--Nordstr\"om; see Section~\ref{sec.previous.works} below) to prove nonlinear instability of the Cauchy horizon in the near-Reissner--Nordstr\"om region in the ``generic'' case. Finally, we use nonlinear methods specific to the Einstein--Maxwell--(real)--scalar--field system in spherical symmetry to derive the desired global $C^{2}$-future-inextendibility property from the aforementioned nonlinear instability. 

To discuss the main result and its proof in more detail, we begin with a brief description of the notions of ``admissible'' and ``generic'' initial data, which are necessary for a precise formulation of the strong cosmic censorship conjecture.

\pfstep{Definition of admissible initial data and genericity (Definitions~\ref{def:adm-data} and \ref{def:adm-top})}
Roughly speaking, we will consider \emph{admissible Cauchy initial data sets} which consists of data for the geometric quantities and the matter fields, which are spherically symmetric, two-ended asymptotically flat, future admissible, and have non-vanishing charge. Some remarks on these aspects are in order.
\begin{itemize}
\item By \underline{asymptotically flat}, we mean that the data for the metric, the Maxwell field and the scalar field approach that of the trivial solution (i.e., Minkowski metric with zero Maxwell and scalar fields) near each end at an appropriate inverse polynomial rate in $r$ (by definition, $r \to \infty$ near an asymptotically flat end).
\item It is in fact necessary for the initial hypersurface to have \underline{two} asymptotically flat ends in order to support a non-vanishing charge (or equivalently, a nontrivial spherically symmetric Maxwell field). 
\item The \underline{future admissibility} condition (see Definition~\ref{def:adm-data}.(5)), which was introduced in \cite{D3}, is a natural generalization of the ``no anti-trapped surfaces'' condition of Christodoulou to the 2-ended case. See Steps~1 and 2(a) below for further discussion of its significance.
\end{itemize}
It turns out that the precise nature of the strong cosmic censorship that holds depends on the rates for which the scalar field decays near the asymptotically ends. In the introduction, in order to simplify the exposition, we will only consider\footnote{As we will see in Theorem~\ref{cor:scc.small.omg}, our methods also apply to the case $2 < \omg_{0} < 3$, although the result will be qualitatively different. The significance of the number $3$ is that it is the sharp Price's law decay rate; assuming faster decay rates of the initial $\phi$ and its derivatives does not improve the decay rate of the future development in general.} the case where $\phi = O(r^{-\omg_{0}})$ near each asymptotically flat end, with $\omg_0\geq 3$.

On the space of admissible initial data sets, we define a scale of weighted-$C^{k}$-type distances (see Definition~\ref{def:adm-top}). The notion of genericity used in the precise version of Theorem~\ref{main.theorem.both} (see Theorem~\ref{cor:scc} and the ensuing statements) is defined (roughly) as being open relative to a weighted-$C^{1}$-type distance, and dense relative to a weighted-$C^{\infty}$-type distance. 

\vskip.5em
The proof of Theorem~\ref{main.theorem.both} can be roughly divided into the following five steps: 

\pfstep{Step 1: An a priori boundary characterization \cite{Kommemi} (Theorem~\ref{thm:kommemi})} Thanks to the future admissibility and asymptotic flatness conditions, we may apply a result of \cite{Kommemi} (or more precisely, by an adaptation of \cite{Kommemi} to the $2$-ended case in \cite{D3}) to show that the maximal globally hyperbolic future development $(\calM, g, \phi, F)$ of every admissible initial data set must have a black hole (interior) region $\Int = \Int_{i^{+}_{1}} \cup \Int_{i^{+}_{2}} \cup \Int_{nonpert}$\footnote{For the definition of $\Int_{i^{+}_{1}}$, $\Int_{i^{+}_{2}}$, see Step~2(b). Then we define $\Int_{nonpert} = \Int \setminus (\Int_{i^{+}_{1}} \cup \Int_{i^{+}_{2}})$.} and an exterior region $\Ext = \Ext_{1} \cup \Ext_{2}$ with $2$ connected components $\Ext_{1}$, $\Ext_{2}$ corresponding to the $2$ asymptotically flat ends of the initial hypersurface. Each connected component of the exterior region has a complete future null infinity (denoted by $\NI_{1}$ and $\NI_{2}$, respectively) and approaches a connected component of the exterior region of a subextremal Reissner--Nordstr\"om spacetime. Moreover, the future boundary of the black hole region can be characterized. In fact, using also the results of \cite{D2} and the fact that the charge is non-vanishing, the quotient $\calQ = \calM / SO(3)$ of the maximal globally hyperbolic future development must be given by one of the two Penrose diagrams in Figure~\ref{fig:main.structure}, where we refer the reader to Theorem~\ref{thm:kommemi} for the notation. 
\begin{figure}[h] 
\begin{center}
\def\svgwidth{450px}
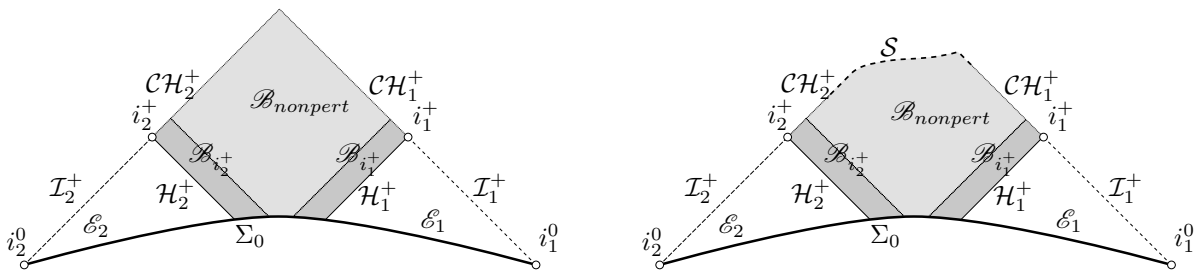 
\caption{Penrose diagram of the maximal globally hyperbolic future development of admissible initial data} \label{fig:main.structure}
\end{center}
\end{figure}

Combining this theorem with results in \cite{D2, DRPL}, it can be shown that the only way that the spacetime can be $C^2$-future-extendible is if there is an extension ``through the Cauchy horizons $\CH_{1} \cup \CH_{2}$''. This is what we will have to rule out in the generic scenario.

\pfstep{Step 2: Convergence to Reissner--Nordstr\"om} In order for the linear analysis in \cite{LO.instab} to be relevant to the nonlinear problem, it needs to be shown that the maximal globally hyperbolic development of every admissible initial data set in fact converges to Reissner--Nordstr\"om\footnote{Note that potentially, the solution approaches two \emph{different} Reissner--Nordstr\"om solutions (i.e., with different parameters of the final masses) ``towards two timelike infinities $i^{+}_{1}$, $i^{+}_{2}$.''} with sufficiently strong estimates. This step is further divided into two substeps. 
\begin{itemize}
\item \pfstep{Step~2(a): Price's law decay in the exterior region \cite{DRPL} (Theorem~\ref{thm:price-law-intro}, Theorem~\ref{thm:DR-decay} or \cite[Section~5]{LO.exterior})} In the exterior region $\Ext_{1} \cup \Ext_{2}$, since the event horizons are subextremal, the seminal work of Dafermos--Rodnianski shows that the spacetime approaches Reissner--Nordstr\"om and the scalar field (and its derivatives) decays with an inverse polynomial rate. In our setting, we will also need some refinements of the theorem of Dafermos--Rodnianski, for which we refer the reader to \cite[Section~5]{LO.exterior}.

\item \pfstep{Step~2(b): $C^0$-stability in the interior region (Theorem~\ref{main.theorem.C0.stability})} In view of the instability result we ultimately prove, in the interior we can only to hope\footnote{We will in fact show that the scalar field is not in $W^{1,2}_{loc}$ in a $C^0$ extension of the spacetime (cf.~Theorem~\ref{final.blow.up.step}). In view of the result for the linear wave equation in \cite{Gleeson}, one expects in general that the scalar field is not even in $W^{1,p}_{loc}$ for any $p>1$. The estimates we prove have to be consistent with this expectation.} to obtain estimates that degenerate at the Cauchy horizon, consistent with $C^0$-stability of the interior region. This is what we achieve in Theorem~\ref{main.theorem.C0.stability}. More precisely, given the scalar-field decay in Step~2(a) in the exterior region (specifically, on the event horizon), we show that the spacetime approaches Reissner--Nordstr\"om (with the same parameters as in the exterior) in some (non-empty!) interior regions $\Int_{i^{+}_{1}}, \, \Int_{i^{+}_{2}} \subset \Int$ sufficiently close to timelike infinity. 

Some of the estimates can be inferred from \cite{D2}, where the $C^0$-extendibility of the interior region was shown (cf.~Theorem~\ref{Dafermos.thm}), but for the later steps we will need a slightly stronger and more quantitative version. For the proof, we combine some ideas from Dafermos--Luk \cite{DL} on the stability of the Kerr Cauchy horizon (without symmetry) with weighted $L^\infty$ estimates that hold in the spherically symmetric setting. 
\end{itemize}

\pfstep{Step~3: The non-vanishing of $\mathfrak L_{(\omg_{0})\infty}$ and $\mathfrak L_{(\omg_{0})\infty}'$ implies $W^{1,2}_{loc}$ blow up of the scalar field near the Cauchy horizon} This can be viewed as a nonlinear version of the result in \cite{LO.instab}. Here we identify a real-valued quantity $\mathfrak L_{(\omg_{0})\infty}$ for an asymptotically flat end, and a corresponding $\mathfrak L_{(\omg_{0})\infty}'$ for the other asymptotically flat end, such that the non-vanishing of these quantities implies that the $W^{1,2}_{loc}$-norm of the scalar field blow up near each of the Cauchy horizons. 

For simplicity, we restrict our attention to one asymptotically flat end (specifically, the one on the right in Figure~\ref{fig:main.structure}); the case of the other asymptotically flat end is analogous. Roughly speaking, $\calL_{(\omg_{0}) \infty}$ measures the leading order coefficient of the expansion of the ``incoming'' part of $\rd \phi$ into powers of $r$ near $i^{+}_{1}$ along null infinity $\NI_{1}$. More precisely, introducing a double null coordinate system\footnote{Throughout this paper, our convention is that $\ud u$ and $\ud v$ are increasing to the future and null. See Section~\ref{sec.SS}.} $(u, v)$ on the region $\Ext_{1} \subset \calQ$ oriented so that the constant-$u$ curves $C_{u}$ are outgoing, we define
\begin{equation*}
	\mathfrak L_{(\omg_{0}) \infty} = \lim_{u \to u_{\EH_{1}}} (\lim_{r \to \infty} r^{3} (\rd_{v} r)^{-1} \rd_{v} (r \phi) \restriction_{C_{u}}),
\end{equation*}
where $u_{\EH_{1}}$ denotes the final $u$-value of null infinity $\NI_{1}$.\footnote{This definition requires $\omg_{0} \geq 3$. The power $r^{3}$ has the same root as the sharp Price's law rate. See Section~\ref{subsec:L-cauchy} for the definition in the case $2 < \omg_{0} < 3$.}

We divide the rest of this step into three substeps.
\begin{itemize}
\item \pfstep{Step~3(a): The non-vanishing of $\mathfrak L_{(\omg_{0}) \infty}$ implies an $L^2$-averaged lower bound of the scalar field on the event horizon (Theorem~\ref{thm:blowup} or \cite[Theorem~4.1]{LO.exterior})} 
In our previous paper \cite{LO.instab} on the linear wave equation on an exact Reissner--Nordstr\"om spacetime, it was shown that if (the suitable linear version of) $\mathfrak{L}_{(\omg_{0}) \infty}$ is nonzero, then the following $L^{2}$-averaged lower bound on the event horizon holds:
\begin{equation} \label{eq:L2-avg-lower}
	\int_{\EH_{1}} v^{\alp} (\rd_{v} \phi)^{2} \, \ud v = \infty \quad \hbox{ for any } \alp > 7,
\end{equation}
where $v$ is the Eddington--Finkelstein advanced null coordinate.

In this substep, we prove the analogous statement in our nonlinear context (see Theorem~\ref{thm:blowup}). Thanks to the Price's law decay estimates of Dafermos--Rodnianski in Step~2(a), we are able to use essentially the same strategy as in the linear case in \cite{LO.instab}.

\item \pfstep{Step~3(b): The lower bound in Step~3(a) implies $W^{1,2}_{loc}$ blow up of the scalar field on the Cauchy horizon near timelike infinity (Theorem~\ref{final.blow.up.step})} 
In \cite{LO.instab}, it was also proven that if the solution $\phi$ of the linear wave equation on exact Reissner--Nordstr\"om obeys \eqref{eq:L2-avg-lower}, then its $W^{1,2}$-norm on a neighborhood of any point on the Cauchy horizon (defined with respect to, say, the analytic extension of Reissner--Nordstr\"om) blows up.
By the decay estimates in our interior $C^{0}$-stability theorem in Step~2(b), we are again able to justify the strategy of proof from \cite{LO.instab} in our context, and establish the analogous $W^{1,2}_{loc}$ blow up\footnote{Here, the $W^{1,2}$ norm is defined with respect to the $C^{0}$ extension given by Step~2(b).} of the scalar field on $\CH_{1} \cap \Int_{i^{+}_{1}}$, provided that the $L^{2}$-averaged lower bound from Step~2(a) holds on $\EH_{1}$. Note that in view of the instability, the estimates we obtain in Step~2(b) are necessarily degenerate near the Cauchy horizon, but nonetheless turns out to be sufficient for our purposes.

\item \pfstep{Step~3(c): The lower bound in Step~3(a) implies $W^{1,2}_{loc}$ blow up of the scalar field on the entire Cauchy horizon (Theorem~\ref{thm.nonpert})} 
Finally, we propagate the $W^{1,2}_{loc}$ blow up of the scalar field on $\CH_{1} \cap \Int_{i^{+}_{1}}$ shown in Step~3(b) to the entire Cauchy horizon $\CH_{1}$. This requires analysis in the region $\Int_{nonpert}$, which is non-perturbative in the sense that the solution is not necessarily close to Reissner--Nordstr\"om. Nevertheless, an important a priori estimate still holds for the model under consideration, namely, any point on $\CH_{1}$ (with the exception of the endpoint) has a neighborhood with \emph{finite} spacetime volume (see Lemma~\ref{lem.nonpert-vol}). This a priori estimate, in turn, allows us to propagate both the $C^{0}$-extendibility and the $W^{1,2}_{loc}$ blow up statements on $\CH_{1} \cap \Int_{i^{+}_{1}}$ (cf. Steps~2(b) and 3(b), respectively) to the entire Cauchy horizon $\CH_{1}$.
\end{itemize}

\pfstep{Step~4: Generic non-vanishing of $\mathfrak L_{(\omg_{0}) \infty}$ and $\mathfrak L_{(\omg_{0}) \infty}'$} 
Let $\calG$ consist of admissible initial data sets whose maximal globally hyperbolic future developments obey $\mathfrak L_{(\omg_{0}) \infty} \neq 0$ and $\mathfrak L'_{(\omg_{0}) \infty} \neq 0$, so that the conclusion of Step~3 (namely, $W^{1,2}$ blow up of the scalar field on the Cauchy horizon) applies. We show that $\calG$ is \emph{generic} in the sense described earlier (see also the substeps below). This step is naturally divided into two substeps.
\begin{itemize}
\item \pfstep{Step~4(a): Nonlinear stability of $\mathfrak L_{(\omg_{0}) \infty}$ (Theorem~\ref{thm:L-stability} or \cite[Theorem~4.2]{LO.exterior})}
To establish openness of $\calG$, it suffices to show that $\mathfrak L_{(\omg_{0}) \infty}$ and $\mathfrak L'_{(\omg_{0}) \infty}$ are nonlinearly stable (or continuous) with respect to initial data perturbations. We focus only on the asymptotically flat end corresponding to $\mathfrak{L}_{(\omg_{0}) \infty}$, as the other case is similar. 

By performing an asymptotic analysis of the wave equation $\Box_{g} \phi = 0$ near null infinity (where ``$r = \infty$''), the quantity $\mathfrak{L}_{(\omg_{0}) \infty}$ can be decomposed into
\begin{equation*}
	\mathfrak L_{(\omg_{0}) \infty} = \mathfrak L + \mathfrak L_{(\omg_{0}) 0},
\end{equation*}
where $\mathfrak L_{(\omg_{0}) 0} = \lim_{r \to \infty} r^{3} (\rd_{v} r)^{-1} \rd_{v} (r \phi) \restriction_{\Sgm_{0}}$ is determined directly by the Cauchy initial data  and $\mathfrak L$ is an integral along null infinity:
\begin{equation} \label{eq:guide.L}
	\mathfrak L = \int_{\NI_{1}} 2 M(u) \Phi(u) \Gmm(u) \, d u.
\end{equation}
Here, $M(u)$, $\Phi(u)$ and $\Gmm(u)$ are limits of the Hawking mass (see Remark~\ref{rem.hawking.mass}), $r \phi$ and $-\frac{1}{4} \frac{\Omg^{2}}{\rd_{v} r}$ along the constant-$u$ curve towards null infinity. Since $\mathfrak L_{(\omg_{0}) 0}$ is clearly continuous with respect to initial data perturbations, it only remains to understand nonlinear stability of $\mathfrak L$. The key point, evident from \eqref{eq:guide.L}, is to show that under small perturbations, the \underline{integral} of the \underline{difference} of $M \Phi$ in an appropriate gauge (say, in which $\Gmm(u) \equiv -1$) is small.

In order to achieve this goal, we establish asymptotic stability of the maximal globally hyperbolic development of any admissible initial data set in the exterior region for initial data perturbations that are small in a suitably weighted $C^{1}$ topology (see Theorem~\ref{thm:L-stability} for the precise definition of the topology). This is, in a sense, the most technically involved part of the entire series. The ingredients of its proof include the Price's law decay theorem in Step~2(a) (to obtain quantitative information about the background solution), choice of suitable future-normalized double null coordinate systems (since decay is expected only in a well-chosen coordinate system), an interaction Morawetz estimate (to control the nonlinearity), $r^{p}$-weighted energy method of Dafermos--Rodnianski \cite{DRNM}, integration along characteristics method from \cite{LO1} (both for proving decay of the nonlinear perturbation) etc. We refer the reader to \cite[Section~8.1]{LO.exterior} for further discussions.

\item \pfstep{Step~4(b): Existence of a continuous one-parameter family of perturbations away from $\mathfrak{L}_{(\omg_{0}) \infty} = 0$ (Theorem~\ref{thm:instability} or \cite[Theorem~4.3]{LO.exterior})}
As before, we only discuss the proof of density of $\set{\mathfrak{L}_{(\omg_{0}) \infty} \neq 0}$, since the case of $\mathfrak{L}'_{(\omg_{0}) \infty}$ is analogous.

Suppose that we are given an admissible initial data set whose maximal globally hyperbolic future development satisfies $\mathfrak{L}_{(\omg_{0}) \infty} = 0$. The idea is to place a smooth compactly supported\footnote{Of course, the perturbation of the metric is \underline{not} compactly supported in general due to the constraint equation, but it is only the perturbation of the scalar field we arrange to be compactly supported.} \underline{outgoing} perturbation of the initial data for $\phi$ of size $\eps > 0$ in the region $\set{r \approx R_{\ast}}$ (near the end corresponding to $\mathfrak{L}_{(\omg_{0}) \infty}$), where $R_{\ast}$ is sufficiently large. On the one hand, by asymptotic flatness, we can perform an explicit calculation (essentially as in exact Reissner--Nordstr\"om) in the domain of dependence of the $\phi$-perturbation to ensure that the contribution to $\mathfrak{L}$ of the perturbation in this region is $\approx \eps$. 
On the other hand, since the $\phi$-perturbation is outgoing, the data on a fixed outgoing null hypersurface to the future of the domain of dependence of the $\phi$-perturbation become small as $R_{\ast} \to \infty$. Indeed, such data are of size $o(\eps)$ as $R_{\ast} \to \infty$, and therefore give negligible contribution to $\mathfrak L$ by the asymptotic stability theorem in Step~4(a).

Since $R_{\ast}$ can be chosen \underline{independent} of $\eps$, the above idea leads to construction of a one-parameter family of perturbations away from $\mathfrak{L}_{(\omg_{0}) \infty} = 0$, which is continuous in a weighted $C^{\infty}$ topology\footnote{More precisely, as regular as the admissible initial data set we started with.} (see Theorem~\ref{thm:instability} for the precise definition of the topology). This implies the desired density of the set $\set{\mathfrak{L}_{(\omg_{0}) \infty} \neq 0}$.
\end{itemize}

\pfstep{Step~5: $C^2$-future-inextendibility of the maximal globally hyperbolic future development (Theorem~\ref{main.theorem.C2})} In the final step, we prove that the generic blow up shown in Steps~3 and 4 in fact implies a \emph{geometric} statement that there does not exist any future extension of the maximal globally hyperbolic future development which has a $C^2$ Lorentzian metric. 
Here, the significance of the regularity $C^{2}$ is that it allows us to pointwisely make sense of the curvature tensor (which is a geometric invariant), whose possible blow up is directly connected with that of the scalar field through the Einstein--Maxwell--(real)--scalar--field system \eqref{EMSFS}. 

More precisely, recall from Step~1 that the goal is to rule out any $C^{2}$ future extensions through the Cauchy horizon. 
The generic $W^{1,2}_{loc}$ blow up of the scalar field on the Cauchy horizon implies, through \eqref{EMSFS}, that a certain component of the Ricci curvature in a frame parallely transported along a geodesic blows up on the Cauchy horizon, which is inconsistent with the $C^{2}$ future extension.

\vskip.5em

As we see from the above steps, the proof of Theorem \ref{main.theorem.both} relies on the analysis both in the interior and the exterior regions of the black hole. The analysis in the interior region $\Int = \Int_{i^{+}_{1}} \cup \Int_{i^{+}_{2}} \cup \Int_{nonpert}$, i.e., Steps 2(b), 3(b), 3(c) and\footnote{Strictly speaking, the proof of inextendibility in Step~5 requires information for both the interior and the exterior regions. Nevertheless, the most difficult step is to rule out the possibility of extending the spacetime through the boundary of the interior region.} 5 are carried out in this paper. The remaining steps, i.e., Steps 3(a), 4(a) and 4(b), which constitute analysis in the exterior region $\Ext = \Ext_{1} \cup \Ext_{2}$, are carried out in \cite{LO.exterior}. We refer the reader to Sections~\ref{sec.SCC}-\ref{sec.pf.SCC} for precise statements that are proven for each of these steps and how they fit together.

\subsection{Previous works}\label{sec.previous.works}
In this subsection, we provide a brief survey of some previous works to place our main result (Theorem~\ref{main.theorem.both}) and the ideas of its proof in context.

\subsubsection{Linear wave equation on Reissner--Nordstr\"om spacetime}

The simplest setting to study the stability and instability properties of the Cauchy horizon in the interior of Reissner--Nordstr\"om spacetime is to consider the \underline{linear} scalar wave equation $\Box_{g_{RN}}\phi=0$ where the background Reissner--Nordstr\"om metric $g_{RN}$ is \underline{fixed}. This problem has a long tradition in the physics literature and has attracted much renewed recent interest from the mathematical community. We refer the reader to \cite{CH, DafShl, Fra, Gleeson, GSNS, Hintz, LS, McN, SP, Sbi.2} for a sample of results and to the introduction of \cite{LO.instab} for further discussions. In the interior of the black hole, solutions to the linear wave equation exhibit both stability and instabilty properties. While $\phi$ itself is uniformly bounded and in fact decays along the Cauchy horizon \cite{Fra, Hintz}, the derivative of $\phi$ in a direction transversal to the Cauchy horizon blows up. Indeed, we have the following instability result:

\begin{theorem}[Linear instability on fixed Reissner--Nordstr\"om, Luk--Oh \cite{LO.instab}]\label{linear.thm}
Generic smooth and compactly supported initial data to the linear wave equation $\Box_{g_{RN}}\phi=0$ on a $2$-ended asymptotically flat complete Cauchy hypersurface of a fixed subextremal Reissner--Nordstr\"om spacetime with nonvanishing charge give rise to solutions that are not in $W^{1,2}_{loc}$ in a neighborhood of any point on the future Cauchy horizon $\CH$.
\end{theorem}

As we already discussed in Section~\ref{sec.main.structure}, the methods introduced in \cite{LO.instab} also play a crucial role in establishing instability for the nonlinear problem in spherical symmetry considered in this paper and \cite{LO.exterior}. Theorem \ref{linear.thm} is proven via showing that the spherically symmetric part of a solution arising from generic smooth and compactly supported initial data is not in $W^{1,2}_{loc}$ in a neighborhood of any point on the future Cauchy horizon. Its proof has the following three parts\footnote{These concern only one of the two asymptotically flat ends of Reissner--Nordstr\"om. An analogue of each part holds in the other end. In particular, the solution we construct in (3) can be made to vanish on the other component of the exterior region, and (by linearity) one can therefore easily add to it a solution that blows up on the ``outgoing'' part of the Cauchy horizon, constructed in essentially the same manner, so as to ensure that the solution blows up on the whole Cauchy horizon.}:
\begin{enumerate}
\item It is shown that in the interior of the black hole region if an $L^2$-averaged polynomial lower bound holds for the spherically symmetric part of the solution on the event horizon, then the solution is not in $W^{1,2}_{loc}$ in a neighborhood of any point on the future Cauchy horizon.
\item Moreover, a quantity\footnote{Notice that in the present paper, we introduce in addition the notation $\mathfrak L_{(\omg_{0}) \infty}$, which coincides with $\mathfrak L$ (in both this paper and \cite{LO.instab}) for solutions arising from compactly supported initial data.} $\mathfrak L$ at future null infinity associated to the spherically symmetric part of the solution is identified in \cite{LO.instab}. The quantity $\mathfrak L$, which is real-valued and depends linearly on $\phi$, moreover has the property that whenever $\mathfrak L\neq 0$, the $L^2$-averaged polynomial lower bound for the solution on the event horizon required in part (1) holds.
\item Finally, it is shown that $\mathfrak L$ is generically non-vanishing, by exhibiting a spherically symmetric solution to the linear wave equation with smooth compactly supported initial data on $\Sgm_{0}$ for which $\mathfrak{L} \neq 0$.
\end{enumerate}
In the language of Section~\ref{sec.main.structure}, Steps~3(b), 3(a) and 4 can be viewed as analogues of (1), (2) and (3) above, respectively, but in a \underline{nonlinear} setting.

\subsubsection{Nonlinear stability and instability of the Cauchy horizon in spherical symmetry} \label{sec.prev-works-nonlinear}

Going beyond the linear wave equation, the next simplest problem is to consider a nonlinear model but restricted to spherical symmetry\footnote{Recall that Reissner--Nordstr\"om is spherically symmetric.}. In part due to the fact that both stable and unstable features can be seen in the linear theory, the nature of the ``singularity'' that arises from perturbing Reissner--Nordstr\"om has been widely debated. In particular, it was speculated that nonlinear perturbations of Reissner--Nordstr\"om initial data may lead to a spacelike singularity. The study of the stability and instability of the Reissner--Nordstr\"om Cauchy horizon was initiated in the pioneering works of Hiscock \cite{Hiscock}, Poisson--Israel \cite{PI1, PI2} and Ori \cite{Ori}, who studied the Einstein equation coupled with null dusts in spherical symmetry. These works suggest that under nonlinear perturbations, the spacetime is regular up to the Cauchy horizon, which is in particular null, and the spacetime metric extends continuously to the Cauchy horizon. Nevertheless, generically, the metric ``blows up'' in the sense that the Hawking mass (described in Remark~\ref{rem.hawking.mass} below) is identically infinite on the Cauchy horizon.

The stability and instability of the Reissner--Nordstr\"om Cauchy horizon was finally settled mathematically in the seminal work\footnote{See also \cite{CGNS2, CGNS3 ,KomThe} for recent extensions of the results of Dafermos.} of Dafermos \cite{D2} (see also \cite{D1}) in the context of the characteristic initial value problem for the Einstein--Maxwell--(real)--scalar field system posed in the black hole interior. He showed that for all initial data on the event horizon approaching Reissner--Nordstr\"om sufficiently fast, in a neighborhood of timelike infinity, the spacetime has a null boundary such that the metric remains continuous. In particular, when sufficiently close to timelike infinity, there are no ``first singularities'' arising before the null boundary. The instability of the smooth Cauchy horizon as suggested by the strong cosmic censorship conjecture only manifests itself in that for a ``large subset'' of admissible initial data, the null boundary is also a null singularity such that the metric cannot be extended beyond in $C^1$ in spherical symmetry (see further discussions in \cite{D2}). We summarize the stability result in \cite{D2} as follows\footnote{In \cite{D2}, the decay rate for $\phi$ is not needed. Also, there is a version of the result requiring only $s>\f 12$. We state here only a version that is easy to compare with Theorem \ref{main.theorem.C0.stability} in which we do not optimize the necessary decay rate on the event horizon. This is in particular because the results in \cite{DRPL} show that the bound holds for some $s>1$ for solutions arising from asymptotically flat Cauchy data.}:
\begin{theorem}[Stability theorem, Dafermos \cite{D2}]\label{Dafermos.thm}
Fix $M$, ${\bf e}$ and $s$ such that $0<|{\bf e}|<M$ and $s>1$. Consider the characteristic initial value problem with smooth data given on $C_{-\infty}$ and $\underline{C}_1$ such that $C_{-\infty}$ approaches the event horizon of Reissner--Nordstr\"om with ${\bf e}$ and $M$ and such that in an ``Eddington--Finkelstein type'' coordinate system, we have
$$(|\phi|+|\rd_v\phi|)\restriction_{C_{-\infty}}(v)\leq E v^{-s}.$$
Then, by restricting to some nonempty, connected subset $\underline{C}'_1\subset \underline{C}_1$, the globally hyperbolic future development of the data on $\underline{C}'_1\cup C_{-\infty}$ has a Penrose diagram given by Figure \ref{fig:maintheorem}.
Moreover, the area-radius function $r$ and the scalar field $\phi$ extend continuously to the Cauchy horizon $\mathcal C\mathcal H^+$.
\end{theorem}

\begin{figure}[htbp]
\begin{center}
\def\svgwidth{150px}
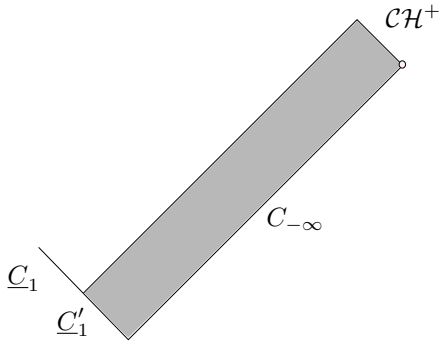 
\caption{Penrose diagram of the development of data on $\underline{C}'_1\cup C_{-\infty}$} \label{fig:maintheorem}
\end{center}
\end{figure}

The stability theorem in \cite{D2} is complemented by the following conditional instability result:
\begin{theorem}[Conditional instability theorem, Dafermos \cite{D2}]\label{Dafermos.instability.thm}
If, in addition to the assumptions in Theorem \ref{Dafermos.thm}, there exists $\epsilon>0$, $c>0$ and $v_*\geq 1$ such that the following \underline{pointwise} lower bound holds  :
$$ |\rd_v\phi|\restriction_{C_{-\infty}}(v)\geq c v^{-3s+\epsilon} $$
for all $v\geq v_*$ (with $s$ as in Theorem \ref{Dafermos.thm}),
then the Hawking mass is identically infinite along the Cauchy horizon $\mathcal C\mathcal H^+$.
\end{theorem}

\begin{remark} \label{rem.hawking.mass}
The Hawking mass described above is a geometric invariant for spherically symmetric spacetimes on every orbit of the $SO(3)$ action (see \eqref{Hawking.def}). The infinitude of the Hawking mass is an obstruction to extending the spacetime in spherical symmetry with a $C^1$ metric. Moreover, it can be shown \cite{Kommemi} in this setting that the blow up of the Hawking mass implies the blow up of the Kretschmann scalar, which is an obstruction to extending the metric in $C^2$ without assuming spherical symmetry for the extension. Note that in contrast to Theorem~\ref{Dafermos.instability.thm}, our proof of Theorem~\ref{main.theorem.both} does not show the blow up of the Hawking mass, but instead requires a different geometric argument, cf. Step~5 in Section~\ref{sec.main.structure}.
\end{remark}

\subsubsection{Nonlinear Price's law and the disproof of the $C^{0}$-formulation of strong cosmic censorship in spherical symmetry}
The preceding two results of Dafermos in the black hole interior assume inverse polynomial decay rates of the scalar field along the event horizon as upper and lower bounds (in Theorems~\ref{Dafermos.thm} and \ref{Dafermos.instability.thm}, respectively). These assumptions are consistent with a well-established heuristics called \emph{Price's law}, which predicts that any asymptotically flat perturbation of Reissner--Nordstr\"om leads to a scalar field with a specific inverse-polynomial upper bound on the event horizon (namely, $\abs{\phi} \leq v^{-3}$ and $\abs{\rd_{v} \phi} \leq v^{-4}$ for an ``Eddington--Finkelstein type'' coordinate $v$), and that the same inverse-polynomial lower bound holds in the ``generic'' case.

Price's law was derived via a heuristic linear analysis, and recently there have been numerous works on its rigorous proof in the context of the linear wave equation on the exterior of a black hole spacetime. We refer the reader to \cite{DS, DSS1, DSS2, AAG, AAG2, Ta,MTT} for a sample of such results. 
Remarkably, in the nonlinear context of the Einstein--Maxwell--(real)--scalar--field system in spherical symmetry, Dafermos--Rodnianski \cite{DRPL} were able to establish the upper bound assertion of Price's law for the maximal development of \underline{any} admissible initial data, which may in principle be very far from a small perturbation of Reissner--Nordstr\"om. We summarize a part of their theorem which is relevant for the current discussion as follows: 
\begin{theorem} [Price's law, Dafermos--Rodnianski \cite{DRPL}]\label{thm:price-law-intro}
Consider smooth spherically symmetric 2-ended Cauchy initial data for \eqref{EMSFS}, which are future-admissible and asymptotically flat (cf. Definition~\ref{def:adm-data}). Then the maximal globally hyperbolic future development obeys
\begin{equation*}
	\abs{\phi} + \abs{\rd_{v} \phi} \restriction_{\EH_{1}} (v)\leq E_{s} v^{-s},
\end{equation*}
for any $s < 3$ and for some $E_{s} > 0$, where $v$ is an ``Eddington--Finkelstein type'' advanced null coordinate on the event horizon $\EH_{1}$ (cf. Definition~\ref{def.EH}). An analogous statement holds on the event horizon $\EH_{2}$ corresponding to the other asymptotically flat end. 
\end{theorem}
For more a precise statement and further discussions, see Step~2(a) in Section~\ref{sec.main.structure}, Theorem~\ref{thm:DR-decay} and \cite[Section~5]{LO.exterior}.

Combining Theorem~\ref{thm:price-law-intro} with the interior stability result (Theorem~\ref{Dafermos.thm}), one reaches the striking conclusion that the $C^{0}$-formulation of the strong cosmic censorship conjecture for the Einstein--Maxwell--(real)--scalar--field system in spherical symmetry is \underline{false} (Theorem~\ref{DDR}). 

\subsubsection{Comparisons with an approach based on Theorem~\ref{Dafermos.instability.thm} in \cite{D2}} \label{sec.comparison}
Given the interior instability result (Theorem~\ref{Dafermos.instability.thm}) and Price's law, one obvious path to proving the $C^{2}$-formulation of the strong cosmic censorship conjecture (Theorem~\ref{main.theorem.both}) would be to establish the assumed pointwise lower bound in Theorem~\ref{Dafermos.instability.thm} for maximal globally hyperbolic future developments of generic data. In a recent work of Angelopoulos--Aretakis--Gajic \cite{AAG}, the authors proved an analogous lower bound\footnote{In fact, they obtained a much stronger result giving the precise asymptotics of the solution. In particular, they proved upper and lower bounds for $\rd_v\phi$ along the event horizon with the same rate predicted by Price's law.} for generic solutions to the \emph{linear} wave equation on a \emph{fixed} subextremal Reissner--Nordstr\"om spacetime with non-vanishing charge. 
In view of this, one may conjecture that the same lower bound holds for generic data in our nonlinear setting in spherical symmetry:

\begin{conjecture}
There exists a generic set of spherically symmetric admissible Cauchy initial data with non-vanishing charge (cf. Definition~\ref{def:adm-data}) and compactly supported initial scalar field such that the lower bound in Theorem~\ref{Dafermos.instability.thm} holds with $s=3$ on each of the event horizons in the maximal globally hyperbolic future developments.
\end{conjecture}

As discussed in Section~\ref{sec.main.structure}, however, we take a different route and instead prove a stronger conditional instability result which requires an $L^2$-averaged, instead of a pointwise, lower bound. The upshot is that such an $L^2$-averaged lower bound is considerably easier (even though it is still highly technical) to prove. In particular, even to prove the lower bound in the linear setting in \cite{AAG}, it first requires sharp upper bound estimates. Such estimates would be even harder to obtain in the nonlinear setting in view of the fact that we are considering large data solutions. In contrast, with our approach, in terms of the decay rates of a fixed solution, it suffices to use the Dafermos--Rodnianski decay theorem in \cite{DRPL}, in which the decay rates are conjecturally\footnote{At the very least, it is known by \cite{AAG} that the upper bounds in \cite{DRPL} are not sharp for solutions to the linear wave equation.} not sharp.

While the main motivation for our improved conditional instability theorem is that it is easier to verify than the condition in \cite{D2}, it should be noted that our approach based on the new conditional instability theorem also provides a road map for proving instability for other matter models and for the vacuum equations without symmetry. This is because in the problem of, say, the instability of the Reissner--Nordstr\"om Cauchy horizon in spherical symmetry for the Einstein--Maxwell--charged--(complex)--scalar field model, or in the problem of the instability of the Kerr Cauchy horizon without symmetry condition for the Einstein vacuum equations, the generic solutions along the event horizon are expected to be oscillatory \cite{HodPiran1, HodPiran2, HodPiran3, Barack, OriSpin}. This is in contrast to the type of behavior that is required by \cite{D2}, but on the other hand is consistent with the condition that we require in our approach.

\subsection{Outline of the paper}\label{sec.outline} 
We end the introduction with an outline of the remainder of the paper. 
\begin{itemize}
\item In Section~\ref{sec.SS}, we introduce the setup for the Einstein--Maxwell--(real)--scalar--field system in spherical symmetry.
\item In Section~\ref{sec.SCC}, we give the precise formulation of strong cosmic censorship (Theorem~\ref{main.theorem.both}, cf. Theorems~\ref{cor:scc}, \ref{cor:scc.small.omg} and Corollary~\ref{cor:RN.instab}).
\end{itemize}
The next few sections are dedicated to the main steps of the proof of Theorem~\ref{main.theorem.both}:
\begin{itemize}
\item In Section~\ref{sec:mghd}, we discuss previous results of Kommemi \cite{Kommemi} and Dafermos--Rodnianski \cite{DRPL} regarding the maximal globally hyperbolic development. 
\item In Section~\ref{sec.precise}, we give the precise statements of the main theorems (Theorems~\ref{main.theorem.C0.stability}, \ref{final.blow.up.step}, \ref{thm.nonpert} and \ref{main.theorem.C2}) proved in this paper regarding the interior region. 
\item In Section~\ref{sec:main-thm}, we discuss the results in \cite{LO.exterior} regarding the exterior region. 
\item In Section~\ref{sec.pf.SCC}, combining the results of Sections~\ref{sec.precise} and \ref{sec:main-thm}, we obtain a proof of Theorems~\ref{cor:scc}, \ref{cor:scc.small.omg} and Corollary~\ref{cor:RN.instab}.
\end{itemize}
The remaining sections contain the proofs of the theorems:
\begin{itemize}
\item In Section~\ref{sec.main.theorem.C0.stability}, we prove the stability theorem (Theorem~\ref{main.theorem.C0.stability}). 
\item In Section \ref{sec.blow.up}, we prove the instability theorem (Theorem~\ref{final.blow.up.step}). 
\item In Section~\ref{sec.nonpert}, the stability and instability theorems are then ``globalized'' (Theorem~\ref{thm.nonpert}). 
\item In Section~\ref{sec.main.theorem.C2}, we prove the theorem on $C^2$ future-inextendibility (Theorem~\ref{main.theorem.C2}). 
\end{itemize}
Finally, we have an appendix with three sections.
\begin{itemize}
\item In Appendix~\ref{sec.subext.pf}, we give a proof of Kommemi's theorem that the limits along the event horizons are always subextremal in the model under consideration. In Appendix~\ref{sec:appendix.gauge}, we give a discussion regarding the gauge condition that we impose on the event horizon. In Appendix~\ref{sec:christoffel}, we study in more detail the blow-up behavior of the solution in the $C^{0}$ extension constructed in Theorems~\ref{main.theorem.C0.stability} and \ref{thm.nonpert}, and show that the Christoffel symbols fail to be locally $L^{p}$-integrable for every $p > 1$.
\end{itemize}

\subsection*{Acknowledgements} We thank Mihalis Dafermos, Jan Sbierski and Yakov Shlapentokh-Rothman for many insightful discussions. We are grateful for comments from Mihalis Dafermos and Maxime van de Moortel on an earlier version of the manuscript. Much of this work was carried out while J. Luk was at Cambridge University and S.-J. Oh was at UC Berkeley. S.-J. Oh thanks Cambridge University for hospitality during several visits. The authors also thank the Chinese University of Hong Kong for hospitality while some of this work was pursued.

J. Luk is supported in part by a Terman Fellowship. S.-J. Oh was supported by the Miller Research Fellowship from the Miller Institute, UC Berkeley and the TJ Park Science Fellowship from the POSCO TJ Park Foundation.

\section{Einstein--Maxwell--(real)--scalar--field system in spherical symmetry}\label{sec.SS}
The purpose of this preliminary section is to provide the precise setup for the model at hand, namely the Einstein--Maxwell--(real)--scalar--field system in spherical symmetry. We begin with the definition of spherical symmetry in our context.

\begin{definition}[Spherically symmetric solutions]\label{def.SS}
Let $(\mathcal M,g,\phi,F)$ be a suitably regular\footnote{The precise regularity is irrelevant here, since the notion of solutions we work with will be defined later with respect to the reduced system in spherical symmetry. See the well-posedness statements in Propositions~\ref{prop.LWP} and \ref{prop.CIVP.LWP}.} solution to the Einstein--Maxwell--(real)--scalar--field system \eqref{EMSFS}. We say that $(\mathcal M,g,\phi,F)$ is \emph{spherically symmetric} if the following properties hold:
\begin{enumerate}
\item The symmetry group $SO(3)$ acts on $(\mathcal M,g)$ by isometry with spacelike orbits.
\item The metric $g$ on $\mathcal M$ is given by
\begin{equation}\label{SS.metric.1}
g=g_{\mathcal Q}+r^2 d\sigma_{\mathbb S^2},
\end{equation}
where
\begin{equation}\label{SS.metric.2}
g_{\mathcal Q}=-\f{\Omg^2 }{2}(du\otimes dv+dv\otimes du)
\end{equation}
is a Lorentzian metric on the $2$-dimensional manifold $\mathcal Q=\mathcal M/SO(3)$ and $r$ is defined to be the area radius function of the group orbit, i.e.,
$$r=\sqrt{\f{\mbox{Area}({\boldsymbol \pi}^{-1}(p))}{4\pi}},$$
for every $p\in \mathcal Q$, where ${\boldsymbol \pi}$ is natural projection ${\boldsymbol \pi}:\mathcal M\to \mathcal Q$ taking a point to the group orbit it belongs to. Here, as in the introduction, $d\sigma_{\mathbb S^2}$ denotes the standard round metric on $\mathbb S^2$ with radius $1$.
\item The function $\phi$ at a point $x$ depends only on ${\boldsymbol \pi}(x)$, i.e., for $p\in \mathcal Q$ and $x,y\in {\boldsymbol \pi}^{-1}(p)$, it holds that $\phi(x)=\phi(y)$.
\item The Maxwell field $F$ is invariant under pullback by the action (by isometry) of $SO(3)$ on $\calM$.
Moreover, there exists ${\bf e}:\mathcal Q\to \mathbb R$ such that
$$F=\f{\bfe}{2({\boldsymbol \pi}^* r)^2}{\boldsymbol \pi}^*(\Omg^2\,du\wedge dv).$$
\end{enumerate}
\end{definition}

It is well-known that for this system, the real-valued function ${\bf e}$ is in fact a \underline{constant}.

In spherical symmetry, the Einstein-Maxwell-(real)-scalar-field system thus reduces to the following system of coupled wave equations for $(r,\phi,\Omega)$
\begin{equation}\label{WW.SS}
\left\{
\begin{aligned}
\rd_u\rd_v r=&-\f {\Omg^2}{4 r}-\f{\rd_u r\rd_v r}{r}+\f{\Omg^2 {\bf e}^2}{4r^3},\\
\rd_u\rd_v \phi=&-\f{\rd_v r\rd_u\phi}{r}-\f{\rd_u r\rd_v\phi}{r},\\
\rd_u\rd_v\log\Omg=&-\rd_u\phi\rd_v\phi-\f {\Omg^2 {\bf e}^2}{2 r^4}+\f{\Omg^2}{4r^2}+\f{\rd_u r\rd_v r}{r^2}.
\end{aligned}
\right.
\end{equation}
The solutions moreover satisfy the following Raychaudhuri equations:
\begin{equation}\label{eqn.Ray}
\left\{
\begin{aligned}
\rd_v\left(\f{\rd_v r}{\Omg^2}\right)=&-\f{r(\rd_v\phi)^2}{\Omg^2},\\
\rd_u\left(\f{\rd_u r}{\Omg^2}\right)=&-\f{r(\rd_u\phi)^2}{\Omg^2}.
\end{aligned}
\right.
\end{equation}
If one solves the characteristic initial value problem with initial data posed on two intersecting null hypersurface, then one can view the equations \eqref{eqn.Ray} as constraint equations for the initial data. It is easy to check that if the constraint equations \eqref{eqn.Ray} are initially satisfied, then they are propagated by \eqref{WW.SS}.

\subsection{Formulation in terms of Hawking mass}
It will be convenient for our discussion to introduce the following notations:
\begin{equation}\label{lambda.nu}
\lambda:=\rd_v r,\quad \nu:=\rd_u r.
\end{equation}
Define also the \emph{Hawking mass} $m:\mathcal Q\to \mathbb R$ by
\begin{equation}\label{Hawking.def}
m:=\f r2\left(1-g_{\mathcal Q}(\nabla r,\nabla r)\right)=\f r 2\left(1+\f{4\lambda\nu}{\Omg^2}\right),
\end{equation}
where $g_{\mathcal Q}$ is as defined in \eqref{SS.metric.2}, as well as the \emph{modified Hawking mass}
\begin{equation}\label{varpi.def}
\varpi = m + \frac{\e^{2}}{2 r}.
\end{equation}
As a consequence of \eqref{WW.SS}, \eqref{eqn.Ray}, \eqref{Hawking.def}, \eqref{varpi.def}, the following equations hold:
\begin{equation} \label{eq:EMSF-r-phi-m}
\left\{
\begin{aligned}
\rd_{u} \rd_{v} r = & \frac{2(\varpi - \frac{\e^{2}}{r})}{r^{2}} \frac{\rd_{u} r \rd_{v} r}{1-\mu}, \\
\rd_{u} \rd_{v} \phi = & - \frac{\rd_{v} r \rd_{u} \phi}{r} - \frac{\rd_{u} r \rd_{v} \phi}{r}, \\
	\rd_{v} \varpi =& \frac{1}{2} \frac{1-\mu}{\rd_{v} r} r^{2} (\rd_{v} \phi)^{2}, \\
	\rd_{u} \varpi =& \frac{1}{2} \frac{1-\mu}{\rd_{u} r} r^{2} (\rd_{u} \phi)^{2},
\end{aligned}
\right.
\end{equation}
where we denote
\begin{equation}\label{mu.def}
\mu:= \f{2m}{r}.
\end{equation}

\subsection{Cauchy problem formulation}
We give a brief discussion of Cauchy problem formulations of the Einstein--Maxwell--(real)--scalar--field system in spherical symmetry. 
\begin{definition}[Cauchy data]\label{def.Cauchy.data}
A \emph{Cauchy initial data set} for the Einstein--Maxwell--(real)--scalar--field system in spherical symmetry consists of a curve $\Sgm_{0}$ (without boundary), a collection of six real-valued functions $(r, f, h, \ell, \phi, \dot{\phi})$ on $\Sgm_{0}$ and a real number ${\bf e}$. We require $r\in C^2(\Sgm_0;\mathbb R)$, $f,\ell,\phi\in C^1(\Sgm_0;\mathbb R)$ and $h,\dot{\phi}\in C^0(\Sgm_0;\mathbb R)$. Moreover, $f, r$ are required to be strictly positive everywhere on $\Sgm_0$. For $\Sgm_0$ parametrized\footnote{At this point, we allow $\rho$ to have either finite or infinite range. We will require $\Sgm_0=\mathbb R$ later in Definition~\ref{def:adm-data}.} by $\rho$, the collection of functions together with ${\bf e}$ give rise to \emph{geometric data} consisting of the following:
\begin{enumerate}
\item The initial hypersurface $\underline{\Sgm_{0}} = \Sgm_{0} \times \bbS^{2}$ is endowed with the intrinsic Riemannian metric
\begin{equation*}
\hat{g} = f^{2}(\rho) \, \ud \rho^2+ r^{2}(\rho) \, \ud \sgm_{\bbS^{2}}.
\end{equation*}

\item The symmetric $2$-tensor $\hat{k}$ on the initial hypersurface $\underline{\Sgm_{0}}$ (which will be the second fundamental form of the solution) given by
\begin{equation*}
	\hat{k} = h(\rho) \, \ud \rho^{2} + \ell(\rho) \, \ud \sgm_{\bbS^{2}}.
\end{equation*}

\item The initial data on $\underline{\Sgm_{0}}$ for the matter fields\footnote{We abuse notation slightly here, where $\phi$ is used to both denote the scalar field in the spacetime and its restriction to the initial slice $\underline{\Sgm_{0}}$.}
$$(\phi, \underline{n} \phi)\restriction_{\underline{\Sgm_{0}}}=(\phi, \dot{\phi}),\quad F(\underline{n}, \rd_\rho)\restriction_{\underline{\Sgm_{0}}}=\f{{\bf e}f}{r^2},$$
where $\underline{n}$ denotes the unique future-directed unit normal to $\underline{\Sgm_0}$ in $\mathcal M$.
\end{enumerate}
Moreover, the following constraint equations are satisfied: 
\begin{gather}
R_{\hat{g}}-|\hat{k}|^2_{\hat{g}}+(\mbox{tr}_{\hat{g}}\hat{k})^2=4T(\underline{n}, \underline{n})=2\dot{\phi}^2+\f{2}{f^2}(\rd_\rho\phi)^2+\f{2{\bf e}^2}{r^4},\label{ham.con}\\
\quad(\mbox{div}_{\hat{g}} \hat{k})_\rho-\rd_{\rho}(\mbox{tr}_{\hat{g}} \hat{k})=2T(\underline{n}, \rd_\rho)=2\dot{\phi}(\rd_{\rho}\phi).\label{mom.con}
\end{gather}
Here, $T = T^{(sf)}_{\mu\nu}+T^{(em)}_{\mu\nu}$ (cf.~\eqref{EMSFS}) and $R_{\hat{g}}$ is the scalar curvature of $\hat{g}$.
\end{definition}

Any Cauchy initial data can be related to initial data for $(r, \phi, \log \Omg)$ in a double null coordinate system. We give the relation in the following lemma, but will omit the proof, which is a straightforward computation.
\begin{lemma} \label{lem:cauchy-to-char}
Consider a parametrization $\rho \mapsto \Sgm_{0}(\rho)$ of the initial curve $\Sgm_{0}$, and consider a double null coordinate system $(u, v)$ on $\calM$ normalized by the conditions
\begin{equation*}
\frac{\ud u}{\ud \rho} = -1, \quad 
\frac{\ud v}{\ud \rho} = 1 \quad \hbox{ on } \Sgm_{0}.
\end{equation*}
Then the following identities hold on $\Sgm_{0}$: 
\begin{gather*}
	\rd_{u} \restriction_{\Sgm_{0}} = \frac{1}{2} (-\rd_{\rho} + f n), \quad 
	\rd_{v} \restriction_{\Sgm_{0}} = \frac{1}{2} (\rd_{\rho} + f n), \\
	\rd_{u} r \restriction_{\Sgm_{0}} = - \frac{1}{2} \rd_{\rho} r + \frac{f}{2 r} \ell, \quad
	\rd_{v} r \restriction_{\Sgm_{0}} = \frac{1}{2} \rd_{\rho} r + \frac{f}{2 r} \ell, \\
	\rd_{u} \phi \restriction_{\Sgm_{0}} = \frac{1}{2} ( - \rd_{\rho} \phi + f \dot{\phi}), \quad
	\rd_{v} \phi \restriction_{\Sgm_{0}} = \frac{1}{2} ( \rd_{\rho} \phi + f \dot{\phi}), \\
	\rd_{u} \log \Omg \restriction_{\Sgm_{0}} = \frac{1}{2f} (- \rd_{\rho} f + h), \quad
	\rd_{v} \log \Omg \restriction_{\Sgm_{0}} = \frac{1}{2f} (\rd_{\rho} f + h),\\
	\Omg \restriction_{\Sgm_{0}}=f,
\end{gather*}
where $n$ denotes the unique future-directed unit normal to $\Sgm_0$ in $\mathcal Q$.
\end{lemma}

By the Choquet-Bruhat--Geroch theorem \cite{CBG} (suitably adapted to the present setting with matter fields), the initial Cauchy data set gives rise to a unique maximal globally hyperbolic future development. This development also inherits the spherical symmetry of the data. In fact, it is easy to verify that the maximal globally hyperbolic future development of the initial data corresponds to the maximal future $(u,v)$ domain for which \eqref{WW.SS} and \eqref{eqn.Ray} are solved. We summarize the result below but omit the details.

\begin{proposition}[Local well-posedness in $C^1$ for the Cauchy problem]\label{prop.LWP}
Given a Cauchy initial data set as in Definition~\ref{def.Cauchy.data}, there exist a unique maximal future\footnote{Here, causality is to be understood with respect to the standard Minkowski metric $g_{\mathbb R^2} = -\f 12(\ud u\otimes \ud v + \ud v \otimes \ud u)$ in the $(u,v)$-plane.} domain $\mathcal Q\subset \mathbb R^2$ of the $(u,v)$-plane, a constant $\e$ and functions $(r,\phi,\Omg)\in C^2(Q)\times C^1(Q)\times C^1(Q)$ (with $r,\,\Omg>0$) such that both \eqref{WW.SS} and \eqref{eqn.Ray} are satisfied, and that the initial Cauchy data (as in Definition~\ref{def.Cauchy.data}) and the initial gauge conditions (as in Lemma~\ref{lem:cauchy-to-char}) are achieved.

Let $(\mathcal M, g, \phi, F)$ to be related to $(\mathcal Q, r, \phi, \Omg, \e)$ according Definition~\ref{def.SS}. Then, moreover, $(\mathcal M, g, \phi, F)$ the maximal globally hyperbolic future development of the given Cauchy initial data set for the system \eqref{EMSFS}.
\end{proposition}

\subsection{Characteristic initial value problem formulation}

As is well-known, \eqref{EMSFS} can also be solved via a \emph{characteristic initial value problem}. This is particularly relevant in this paper for the study of the interior region, where the question of the stability and instability of the Cauchy horizon is most conveniently phrased in terms of a characteristic initial value problem, see Section~\ref{sec.precise}. The analogue of Proposition~\ref{prop.LWP} in the setting of the characteristic initial value problem is given in the following proposition (whose proof we again omit):
 
\begin{proposition}[Local well-posedness in $C^1$ for the characteristic initial value problem]\label{prop.CIVP.LWP}
Let $\uC_{in}$ and $C_{out}$ be two transversally intersecting null curves parametrized by $\uC_{in}=\{(u, v_1): u\in [u_1, u_2 )\}$ and $C_{out}=\{(u_1, v): v\in [v_1, v_2 )\}$, where $u_1, v_1\in \mathbb R$ and $u_2, v_2\in \mathbb R\cup\{+\infty\}$ with $u_1<u_2$, $v_1<v_2$. Given \emph{characteristic initial data}
\begin{itemize}
\item a constant $\e$,
\item $(r,\phi,\Omg)\in C^2(\uC_{in})\times C^1(\uC_{in})\times C^1(\uC_{in})$ with $r>0$ and $\Omg>0$, and
\item $(r,\phi,\Omg)\in C^2(C_{out})\times C^1(C_{out})\times C^1(C_{out})$ with $r>0$ and $\Omg>0$
\end{itemize}
such that
\begin{itemize}
\item the values of $(r,\phi,\Omg)$ at $(u_1, v_1)$ coincide,
\item the Raychaudhuri equations \eqref{eqn.Ray} are satisfied on $\uC_{in}$ and $C_{out}$.
\end{itemize}
Then there exists a unique maximal globally hyperbolic future development $(r,\phi,\Omg,\e)$ in the $(u,v)$ coordinate system such that $\e$ is a constant and $(r,\phi,\Omg)\in C^2\times C^1\times C^1$ obeys the system of wave equations \eqref{WW.SS}. Moreover, the Raychaudhuri equations \eqref{eqn.Ray} are also satisfied. In particular, $(\mathcal M, g, \phi, F)$, which is related to $(\mathcal Q, r, \phi, \Omg, \e)$ as in Definition~\ref{def.SS}, is a solution to \eqref{EMSFS}.
\end{proposition}

\section{Formulation of strong cosmic censorship}\label{sec.SCC}
In this section we give a precise formulation of strong cosmic censorship for maximal globally hyperbolic future development of admissible initial data (Theorem~\ref{cor:scc}, see also Theorem~\ref{cor:scc.small.omg}). Before that, we need to introduce three important notions:
\begin{itemize}
\item A class of admissible Cauchy data: This includes the assumptions of the regularity and asymptotic flatness of the initial data. We will impose an additional \emph{future admissibility condition}, see the last point in Definition~\ref{def:adm-data}, which can be thought of as an analogue of Christodoulou's ``no anti-trapped surface'' condition suitably adapted to the case of $2$-ended asymptotically flat data, see \cite{D3}. This will be introduced in Section~\ref{sec.adm.data}.
\item Topologies on the set of admissible Cauchy data: The various topologies will allow us to make precise the \emph{genericity} condition, which can be understood as openness and density statements in appropriate topologies. This will be introduced in Section~\ref{sec:topology}.
\item The constants $\mathfrak L_{(\omg)\infty}$ and $\mathfrak L'_{(\omg)\infty}$: For every admissible Cauchy data set, the constants $\mathfrak L_{(\omg)\infty}$ and $\mathfrak L_{(\omg)\infty}'$ are introduced to give an easy-to-check \emph{criterion} which, when verified, would guarantee that the maximal globally hyperbolic future development is $C^2$-future-inextendible. These constants will be featured in the statement of strong cosmic censorship and will be introduced in Section~\ref{subsec:L-cauchy}.
\end{itemize}

We will then give the main statement of strong cosmic censorship in Section~\ref{sec.statement}. This is followed in Section~\ref{SCC.small.omg} by a (much) simpler statement if a sufficiently slowly decaying inverse polynomial tail is allowed in the initial data. We will end the section with a discussion on the global stability and instability of the Reissner--Nordstr\"om Cauchy horizon in Section~\ref{sec.instab.RN}. \textbf{All of the results stated in this section will be proven in Section~\ref{sec.pf.SCC}.}

\subsection{Admissible Cauchy data}\label{sec.adm.data}
We now define the class of Cauchy initial data for which we will establish strong cosmic censorship.
\begin{definition} [Admissible Cauchy initial data] \label{def:adm-data}
Let $\omg_{0} > 2$. An \emph{$\omg_{0}$-future-admissible spherically symmetric $2$-ended asymptotically flat Cauchy initial data set with non-vanishing charge} (in short, an $\omg_{0}$-admissible initial data set) is a Cauchy initial data set $\Tht = (r, f, h, \ell, \phi, \dot{\phi}, \e)$ on $\Sgm_{0} = \bbR$ satisfying the following properties:
\begin{enumerate}
\item $\phi,\, f \in C^2(\Sigma_0;\mathbb R)$, $\dot{\phi},\, h \in C^1(\Sigma_0;\mathbb R)$ (i.e., are more regular\footnote{Note that while the regularity assumptions in Definition~\ref{def.Cauchy.data} are sufficient for the local existence and uniqueness result in Proposition~\ref{prop.LWP} (and it fact also for global decay result, see \cite{DRPL}), we need additional regularity in order to control the difference of the gauges in two different solutions  in \cite{LO.exterior}. This is of course unsurprising in view of the quasilinear nature of the problem.} than required in Definition~\ref{def.Cauchy.data}).
\item The following \emph{asymptotic flatness conditions}\footnote{These conditions are essentially the asymptotic flatness condition in \cite[Section~A.4]{DRPL} with $\alp = 1$. } hold as $\rho \to \pm\infty$ (i.e., towards each end):
\begin{equation} \label{eq:adm-id-af}
\begin{gathered}
	f(\rho) - 1 = O_{2}(\abs{\rho}^{-1}),\quad
	h(\rho) = O_{1}(\abs{\rho}^{-2}), \\
	r(\rho) - \abs{\rho} = O_{2}(\log \abs{\rho}), \quad
	\ell(\rho) = O_{1}(1)
\end{gathered}
\end{equation}
Here the notation $O_{i}(\abs{\rho}^{-n})$ denotes that the function on the LHS is $O(|\rho|^{-n})$ and the $j$-th derivative is $O(|\rho|^{-n-j})$ for all $j\leq i$. In the case $n = 0$, we simply write $O_{i}(1) = O_{i}(| \rho |^{0})$. The notation $O_{i}(\log \abs{\rho})$ is defined similarly. 

\item 
The following asymptotic flatness conditions hold for the scalar field: As $\rho \to \pm \infty$,
\begin{equation} \label{eq:adm-id-phi}
	\phi (\rho) = O_{2}(\abs{\rho}^{-\omg_{0}}), \quad 
	\dot{\phi}(\rho) = O_{1}(\abs{\rho}^{-\omg_{0}-1}). 
\end{equation}
Furthermore, the following limits exist:
\begin{equation} \label{eq:adm-id-limits}
\begin{aligned}
	\mathfrak{L}_{(\omg_{0}) 0}[\Tht] := & \lim_{\rho \to \infty} \frac{r^{\min\set{\omg_{0}, 3}}}{\rd_{\rho} r + \frac{f\ell}{r}} \left(\rd_{\rho} (r \phi) + \frac{f \ell}{r} \phi + f r \dot{\phi} \right), \\ 
	\mathfrak{L}_{(\omg_{0}) 0}'[\Tht] := & \lim_{\rho \to - \infty} \frac{r^{\min\set{\omg_{0}, 3}}}{-\rd_{\rho} r + \frac{f\ell}{r}} \left(-\rd_{\rho} (r \phi) + \frac{f \ell}{r} \phi + f r \dot{\phi} \right).
\end{aligned}
\end{equation}
\item $\e\neq 0$.

\item The following \emph{future admissibility condition} holds: There exist $\rho_1< \rho_2$ such that 
\begin{equation} \label{eq:adm-id-adm}
	\left(- \rd_{\rho} r + \frac{f \ell}{r}  \right) (\rho) < 0 \hbox{ for all } \rho \geq \rho_{1}, \quad \hbox{ and } \quad
	\left(\rd_{\rho} r + \frac{f \ell}{r}\right) (\rho) < 0 \hbox{ for all } \rho \leq \rho_{2}.
\end{equation}
\end{enumerate}
We denote the set of all $\omg_{0}$-admissible initial data sets by $\calA \calI \calD(\omg_{0})$.
\end{definition}

\begin{remark}[Future admissibility condition in double null coordinates]\label{rmk:adm}
By Lemma~\ref{lem:cauchy-to-char}, the future admissibility condition can be understood as requiring that $(\rd_u r)\restriction_{\Sigma_0}(\rho)<0$ for $\rho\geq \rho_{1}$ and $(\rd_v r)\restriction_{\Sigma_0}(\rho)<0$ for $\rho\leq \rho_{2}$. In particular, there exist \emph{trapped surfaces}\footnote{We say that a sphere of symmetry $S={\bf \pi}^{-1}(p)$ (for some $p\in \mathcal Q$) is a trapped surface if $(\rd_u r)(p)<0$ and $(\rd_v r)(p)<0$.} in $\Sigma_0$.
\end{remark}

\begin{remark}[Stability of future admissibility condition]\label{rmk:adm-stab}
The future admissibility condition is manifestly stable against small perturbations (say, with respect to the distances defined in Section~\ref{sec:topology}).
\end{remark}

\begin{example}[Reissner--Nordstr\"om initial data]\label{ex:RN}
An example of an $\omg_{0}$-future-admissible Cauchy initial data in the sense of Definition~\ref{def:adm-data} (for any $\omg_{0} > 2$) is the Cauchy initial data set of a subextremal Reissner--Nordstr\"om solution with $\e \neq 0$ on a suitable initial hypersurface $\underline{\Sgm_{0}}$. More precisely, we take $\underline{\Sgm_{0}}$ to be a (sufficiently smooth) spherically symmetric Cauchy hypersurface that intersects the black hole region (shaded region in Figure~\ref{fig:RNadm}). Note that $\phi = \dot{\phi} = 0$ in this case. We also note that, by (an extension of) Birkhoff's theorem for the Einstein--Maxwell system, any admissible initial data set satisfying Definition~\ref{def:adm-data} with $\phi = \dot{\phi} = 0$ is necessarily an embedded Cauchy hypersurface in some subextremal Reissner--Nordstr\"om solution.
\end{example}

\begin{figure}[h]
\begin{center}
\def\svgwidth{160px}
{\small 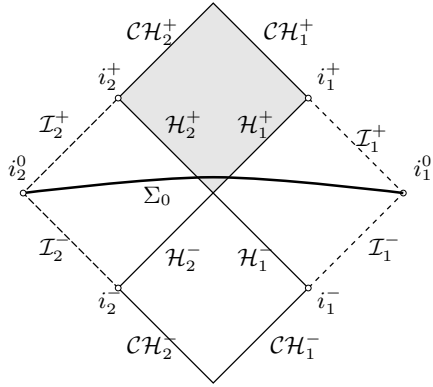 }
\caption{Penrose diagram of a subextremal Reissner--Nordstr\"om spacetime. A Cauchy hypersurface $\Sigma_0$, which intersects the black hole region (shaded) and hence gives rise to admissible Cauchy data, is depicted. For further description of the diagram in the future of $\Sgm_{0}$, we refer to Theorem~\ref{thm:kommemi}. The analogous structures in the past of $\Sgm_{0}$ are marked with the superscript ${}^{-}$. } \label{fig:RNadm} 
\end{center}
\end{figure}

\subsection{Topologies on the set of admissible Cauchy data}\label{sec:topology}

To formulate strong cosmic censorship, we next introduce various topologies on the set of admissible initial data $\calA \calI \calD(\omg_{0})$ to measure the size of initial data perturbations.
\begin{definition} [Distances $d_{k,\omg}$ on $\calA \calI \calD(\omg_{0})$] \label{def:adm-top}
Given any positive integer $k$ and real numbers $\omg, \omg_{0} > 2$, we define the distance $d_{k, \omg}$ on the set $\calA \calI \calD(\omg_{0})$ of $\omg_{0}$-admissible initial data with two asymptotic ends (cf. Definition~\ref{def:adm-data}) as follows (we allow\footnote{This in particular happens when the initial data are not $k$-times differentiable.} $d_{k, \omg} (\Tht, \overline{\Tht}) = \infty$): 
\begin{equation} \label{eq:d-k-omg}
\begin{aligned}
	d_{k, \omg} (\Tht, \overline{\Tht}) := 
	& \nrm{\brk{\rho} \log (f/\fbg)(\rho)}_{C^{0}} 
	+ \sum_{i = 1}^{k} \left(  \nrm{\brk{\rho}^{1+i} \rd_{\rho}^{i} \log (f /\fbg)(\rho)}_{C^{0}} 
					+ \nrm{\brk{\rho}^{1+i} \rd_{\rho}^{i-1} (h - \hbg)(\rho)}_{C^{0}} \right) \\
	& + \nrm{\log^{-1}(1+\brk{\rho}) (r - \rbg) (\rho)}_{C^{0}} 
	+ \sum_{i=1}^{k} \left(  \nrm{\brk{\rho}^{i} \rd_{\rho}^{i} (r - \rbg) (\rho)}_{C^{0}} 
					+ \nrm{\brk{\rho}^{i-1} \rd_{\rho}^{i-1} (f \ell - \overline{f \ell})(\rho)}_{C^{0}} \right) \\
	& + \nrm{\brk{\rho}^{\omg} (\phi - \phibg)(\rho)}_{C^{0}} 
	+ \sum_{i=1}^{k} \left(  \nrm{\brk{\rho}^{\omg+i} \rd_{\rho}^{i} (\phi - \phibg)}_{C^{0}} 
	+ \nrm{\brk{\rho}^{\omg+i} \rd_{\rho}^{i-1} (f \dot{\phi} - \overline{f \dot{\phi}})(\rho)}_{C^{0}} \right) \\
  & + |{\bf e}- \overline{\bf e}|												.
\end{aligned}\end{equation}
Here, $\brk{\rho} = (1 + \rho^{2})^{1/2}$ and $\rho_{\pm} = \max\set{0, \pm \rho}$.
\end{definition}
 
\begin{remark} \label{rem:adm-top}
Some simple remarks concerning the distances $d_{k, \omg}$ are in order.
\begin{enumerate}
\item Note that $d_{k, \omg} \leq d_{k', \omg}$ for $k \leq k'$, and $d_{k, \omg} \leq d_{k, \omg'}$ for $\omg \leq \omg'$.
\item For $k \geq 1$ and $\omg \geq \omg_{0}$, we have
\begin{equation*}
	\abs{\mathfrak{L}_{(\omg_{0}) 0}[\Tht] - \mathfrak{L}_{(\omg_{0}) 0}[\overline{\Tht}]}
	+ \abs{\mathfrak{L}_{(\omg_{0}) 0}'[\Tht] - \mathfrak{L}_{(\omg_{0}) 0}'[\overline{\Tht}]}
	\leq C d_{k, \omg} (\Tht, \overline{\Tht}).
\end{equation*}
In fact, if $\omg > \omg_{0}$, then $d_{k, \omg} (\Tht, \overline{\Tht}) < \infty$ implies 
$\mathfrak{L}_{(\omg_{0}) 0}[\Tht] = \mathfrak{L}_{(\omg_{0}) 0}[\overline{\Tht}]$ and $\mathfrak{L}'_{(\omg_{0}) 0}[\Tht] = \mathfrak{L}'_{(\omg_{0}) 0}[\overline{\Tht}]$.
\end{enumerate}
\end{remark}

With the help of Definition~\ref{def:adm-top}, we introduce the subclass of $C^{k}_{\omg}$ initial data sets in $\calA \calI \calD(\omg_{0})$.
\begin{definition}[$C^k_{\omg}$ initial data]\label{Ck.def}
For $k\in \mathbb N$ with $k\geq 2$ and $\omg, \omg_0>2$, we say that $\Theta\in \calA \calI \calD(\omg_{0})$ is $C^k_\omg$ if 
$d_{k,\omg}(\Theta,\Theta_{RN, M,\e})<\infty$ for some $\Theta_{RN, M,\e}$ which is an admissible smooth Cauchy initial data set for a fixed Reissner--Nordstr\"om solution with parameter $0<|\e|< M$ such that outside a compact set $[-R,R]$, $(r,f,h,\ell,\phi,\dot{\phi},\e)=(|\rho|,(1-\f{2M}{|\rho|}+\f{\e^2}{\rho^2})^{-\f 12},0,0,0,0,\e)$.
\end{definition}

\begin{remark}
Some remarks regarding this definition are in order.
\begin{enumerate}
\item By definition, any $\Theta \in \calA \calI \calD(\omg_{0})$ is automatically $C^2_{\omg_{0}}$.
\item If $d_{k,\omg}(\Theta,\Theta_{RN, M_1,\e})<\infty$ for some $M_1>|\e|$, then $d_{k,\omg}(\Theta,\Theta_{RN, M_2,\e})<\infty$ for any $M_2>|\e|$.
\end{enumerate}
\end{remark}

\subsection{Definition of $\mathfrak{L}_{(\omg_{0}) \infty}$ for developments of $\omg_{0}$-admissible Cauchy data}\label{subsec:L-cauchy}
Next, we define a quantity, which we denote by $\mathfrak{L}_{(\omg_{0}) \infty}$, that seeks to capture the asymptotic size of the incoming scalar field radiation along future null infinity $\NI$ towards timelike infinity. In order to discuss the definition, we will necessarily consider the maximal globally hyperbolic future development of initial data. Here, we only need to use the fact that the maximal globally hyperbolic future development has two components of future null infinity. We refer the reader to Theorem~\ref{thm:kommemi} for further discussions about the relevant geometric notions.

Let $(\mathcal M, g, \phi, F)$ be the maximal globally hyperbolic future development of an $\omg_{0}$-admissible initial data $\Tht$ (cf. Definition~\ref{def:adm-data}; we take $\omg_{0} > 2$), and consider the asymptotically flat end where $\rho \to \infty$ on the initial hypersurface $\Sgm_{0}$. Denote the component of future null infinity corresponding to this asymptotically flat end by $\NI$ and let it be parametrized by the $u$-coordinate.

\begin{definition}[Definition of $\mathfrak{L}_{(\omg_{0}) \infty}$]\label{def.L}
Let $(u, v)$ be a double null coordinates normalized as in Lemma~\ref{lem:cauchy-to-char}. To define $\mathfrak{L}_{(\omg_{0}) \infty}$ in the case $\omg_{0} \geq 3$, we first introduce
\begin{align} \label{eq:L-def}
	\mathfrak{L} :=& \int_{\NI} 2 M(u) \Phi(u) \Gmm(u) \, \ud u,
\end{align}
where $M(u) = \lim_{v \to \infty} \varpi(u, v)$, $\Phi(u) = \lim_{v \to \infty} r \phi(u, v)$ and $\Gmm(u) = \lim_{v \to \infty} \frac{\rd_{u} r}{1- \mu}(u, v)$.
We then define
\begin{equation} \label{eq:Linfty-def}
\mathfrak{L}_{(\omg_{0}) \infty} := \begin{cases}\mathfrak{L}_{(\omg_{0}) 0} + \mathfrak{L}\quad &\mbox{if }\omg_{0}\geq 3, \\
\mathfrak{L}_{(\omg_{0}) 0}\quad &\mbox{if }\omg_{0} \in (2,3),
\end{cases}
\end{equation}
where $\mathfrak{L}_{(\omg_{0}) 0} = \mathfrak{L}_{(\omg_{0}) 0}[\Tht]$ depends only on the initial data and is defined as in \eqref{eq:adm-id-limits}. 
\end{definition}

\begin{remark}[Well-definedness of $\mathfrak L$]
We first note that the definition of $\mathfrak L$ is independent of the choice of the $u$-coordinate. Moreover, the quantity $\mathfrak{L}$ is in fact well-defined as a \emph{finite} real number: Indeed, Theorem~\ref{thm:DR-decay} below shows that the limits as $v\to \infty$ exist and that the quantity $2 M(u) \Phi(u) \Gmm(u)$ is integrable towards the future. The fact that it is also past-integrable follows from asymptotic flatness and will be proven in \cite{LO.exterior}.
\end{remark}

\begin{definition}[Definition of $\mathfrak L_{(\omg)\infty}'$]
Corresponding to the other asymptotically flat end (where $\rho \to - \infty$ on $\Sgm_{0}$), we analogously define $\mathfrak{L}'$ and $\mathfrak{L}_{(\omg) \infty}'$ (switching the roles of $u$ and $v$). 
\end{definition}
 
\subsection{Statement of the strong cosmic censorship theorem in spherical symmetry}\label{sec.statement}

Before we give the statement of strong cosmic censorship, we need to make precise the notion of $C^2$-future-inextendibility:
\begin{definition}[$C^2$-future-inextendibility]\label{def.FE}
A $C^{2}$ Lorentzian $4$-manifold $(\mathcal M,g)$ is said to be \emph{future-extendible with a $C^2$ Lorentzian metric} if there exists a time-oriented, connected $C^2$ Lorentzian $4$-manifold $(\widetilde{\mathcal M}, \tilde{g})$ and an isometric embedding $\iota: \mathcal M\to \widetilde{\mathcal M}$ such that $\iota(\mathcal M)$ is a proper subset of $\widetilde{\mathcal M}$ and moreover for every $p\in \widetilde{\mathcal M}\setminus \mathcal M$, $I^+(p)\cap \mathcal M=\emptyset$. $(\mathcal M,g)$ is said to be \emph{future-inextendible with a $C^2$ Lorentzian metric} (in short, $C^2$-future-inextendible) otherwise.
\end{definition}

\begin{definition}[$C^2$-future-inextendibility of solutions to \eqref{EMSFS}] We say that a solution $(\mathcal M, g, \phi, F)$ to \eqref{EMSFS} is $C^2$-future-inextendible if the underlying Lorentzian manifold $(\mathcal M, g)$ is $C^2$-future-inextendible (as in Definition~\ref{def.FE}).
\end{definition}

The following is the strong cosmic censorship theorem for $\omg_{0}$-admissible initial data, with $\omg_0 \geq 3$.
\begin{theorem}[Strong cosmic censorship in spherical symmetry for $2$-ended asymptotically flat data] \label{cor:scc}
Let $\omg_{0} \geq 3$, and let $\calG$ be the subset of $\calA \calI \calD(\omg_{0})$ (cf. Definition~\ref{def:adm-data}) consisting of elements whose maximal globally hyperbolic future developments satisfy $\mathfrak{L}_{(\omg_{0}) \infty} \neq 0$ and $\mathfrak{L}_{(\omg_{0}) \infty}' \neq 0$ (cf. Definition~\ref{def.L}). Such a set $\calG$ obeys the following properties:
\begin{enumerate}
\item ($C^2$-future-inextendibility of solutions arising from the generic set) The maximal globally hyperbolic future development of any element in $\mathcal G$ is $C^2$-future-inextendible (cf. Definition~\ref{def.FE}).
\item (Openness of the generic set with respect to $d_{1, 2+}$) If $\overline{\Tht} \in \mathcal G$, then for every $\om>2$, there exists $\ep>0$ such that all $\Tht \in \calA \calI \calD(\omg_{0})$ with $d_{1, \omg}(\Tht, \overline{\Tht}) + \abs{\mathfrak{L}_{(\omg_{0}) 0}[\Tht] - \mathfrak{L}_{(\omg_{0}) 0}[\overline{\Tht}]} + \abs{\mathfrak{L}'_{(\omg_{0}) 0}[\Tht] - \mathfrak{L}'_{(\omg_{0}) 0}[\overline{\Tht}]} <\epsilon$ are in $\mathcal G$. 
\item (Density of the generic set) 
Let $\overline{\Tht} \in \calA \calI \calD (\omg_0) \setminus \calG$. Then there exists a one-parameter family $(\Tht_{\eps})_{\eps \in (-\eps_{\ast}, \eps_{\ast})} \subseteq \calA \calI \calD (\omg_0)$ of admissible initial data sets (for some $\eps_{\ast} = \eps_{\ast}(\overline{\Tht}) > 0$) such that 
\begin{itemize}
\item $\Tht_{0} = \overline{\Tht}$, 
\item $\Tht_{\eps} \in \calG$ for all $\eps \in (-\eps_{\ast}, \eps_{\ast})\setminus\{0\}$,
\item if $\overline{\Theta}\in C^k_{\omg_0}$ for $k\in\mathbb N$, $k\geq 2$, then $\eps \mapsto \Tht_{\eps}$ is continuous with respect to $d_{k, \omg}$ for all $\omg > 2$, 
\item $\mathfrak{L}_{(\omg_{0}) 0}[\Tht_{\eps}] = \mathfrak{L}_{(\omg_{0}) 0}[\overline{\Tht}]$ and $\mathfrak{L}_{(\omg_{0}) 0}'[\Tht_{\eps}] = \mathfrak{L}_{(\omg_{0}) 0}'[\overline{\Tht}]$ for all $\eps \in (-\eps_{\ast}, \eps_{\ast})$.
\end{itemize}
In particular, $\calG\cap \left(\cap_{k\geq 2} C^k_{\omg}\right)$ is dense in $\calA \calI \calD (\omg_0)\cap C^\ell_{\omg_0}$ with respect to $d_{\ell, \omg}$ for any $\ell\geq 2$ and $\omg > 2$.
\end{enumerate}
\end{theorem}

\begin{remark}[Topologies for the initial data] \label{rem:scc-top}
As a quick consequence of Theorem~\ref{cor:scc}, observe that $\calG$ is open and dense with respect to $d_{2, \omg_{0}}$, with respect to which $\calA \calI \calD(\omg_{0})$ is an open subset of a complete metric space (we omit the obvious proof). Theorem~\ref{cor:scc}, however, is much stronger in that it allows for different topologies for the openness and density statements.
\end{remark}

\begin{remark}[$\omg$-weights in the density statement] 
Notice that the density statement holds for an \emph{arbitrary} $\omg>2$. In particular, $\overline{\Theta}$ is \underline{not} required to be in $C^k_\omg$. In such a case, the perturbations $\Tht_\eps$ can nonetheless still be constructed so that $d_{k,\omg}(\overline{\Theta}, \Theta_\eps)\to 0$ as $\eps\to 0$.
\end{remark}

\begin{remark}[Density of the generic set in $C^\infty_{\omg_0}$]\label{rmk.Cinfty}
Theorem~\ref{cor:scc} in fact implies that $\calG$ (when restricted to appropriately weighted $C^\infty$ data sets) is dense with respect to a weighted $C^\infty$ topology defined as follows. Let $\omg_0 \geq 3$ and define $C^\infty_{\omg_0}$ as the following space:
$$C^\infty_{\omg_0}:=\cap_{k\geq 2} C^k_{\omg_0}, $$
where $C^k_{\omg_0}$ is as in Definition~\ref{Ck.def}.
For every $\omg >2$, define a distance on $C^\infty_{\omg_0}$ by
\begin{equation}\label{Cinfty.dist}
d_\omg(\Theta, \overline{\Theta})= \sum_{k=1}^{\infty}2^{-k}\f{{d}_{k,\omg}(\Theta, \overline{\Theta})}{1+{d}_{k,\omg}(\Theta, \overline{\Theta})}
\end{equation}
In particular, when $\omg=\omg_0$, the distance $d_{\omg_0}$ (which is obviously complete) makes $C^\infty_{\omg_0}$ a Fr\'echet space. It follows from Theorem~\ref{cor:scc} that in fact for every $\omg > 2$, $\mathcal G\cap C^\infty_\omg$ is dense in $\calA \calI \calD (\omg_0)\cap C^\infty_{\omg_0}$ with respect to the distance defined by \eqref{Cinfty.dist}.
\end{remark}

\subsection{Strong cosmic censorship in the case $2 < \omg_{0} < 3$}\label{SCC.small.omg}

Our strategy also gives a strong cosmic censorship theorem for $\omg_{0}$-admissible initial data with $2 < \omg_{0} < 3$. In this case, however, the nature of the density statement is qualitatively different, since the backscattering effect (captured by $\mathfrak{L}$) is too weak to modify the leading order tail of $\rd_{v}(r \phi)$ near the future endpoint of $\NI$ (captured by $\mathfrak{L}_{(\omg_{0}) \infty}$).

\begin{theorem}[Strong cosmic censorship in the case $2 < \omg_{0} < 3$]\label{cor:scc.small.omg}
Let $\omg_{0} \in (2,3)$, and let $\calG$ be the subset of $\calA \calI \calD(\omg_{0})$ (cf. Definition~\ref{def:adm-data}) consisting of elements whose maximal globally hyperbolic future developments satisfy $\mathfrak{L}_{(\omg_{0}) \infty} \neq 0$ and $\mathfrak{L}_{(\omg_{0}) \infty}' \neq 0$ (cf. Definition~\ref{def.L}). Such a set $\calG$ obeys the same properties as (1) and (2) in Theorem~\ref{cor:scc}, but instead of $(3)$, it satisfies the following:
{\it \begin{enumerate}
\item [(3')] Let $\overline{\Tht} \in \calA \calI \calD (\omg_0) \setminus \calG$. Then there exists a one parameter family $(\Tht_{\eps})_{\eps \in (-\eps_{\ast}, \eps_{\ast})} \subseteq \calA \calI \calD (\omg_0)$ of admissible initial data sets (for some $\eps_{\ast} = \eps_{\ast}(\overline{\Tht}) > 0$) such that 
\begin{itemize}
\item $\Tht_{0} = \overline{\Tht}$, 
\item $\Tht_{\eps} \in \calG$ for all $\eps \in (-\eps_{\ast}, \eps_{\ast}) \setminus \{0\}$,
\item if $\overline{\Theta}\in C^k_{\omg_0}$ for $k\in\mathbb N$, $k\geq 2$, $\eps \mapsto \Tht_{\eps}$ is continuous with respect to $d_{k, \omg_{0}}$.
\end{itemize} 
In particular, $\calG\cap \left(\cap_{k\geq 2} C^k_{\omg}\right)$ is dense in $\calA \calI \calD (\omg_0)\cap C^\ell_{\omg_0}$ with respect to $d_{\ell, \omg_0}$ for any $\ell\geq 2$.
\end{enumerate}
}
\end{theorem}

\begin{remark}[$\mathfrak L_{(\omg_0)\infty}$ depends \emph{only} on the initial data for $\omg_0\in (2,3)$]
The proof of Theorem~\ref{cor:scc.small.omg} is much simpler than its counterpart (Theorem~\ref{cor:scc}) when $\omg_0\geq 3$. In fact, the philosophies of the proofs for Theorem~\ref{cor:scc.small.omg} and Theorem~\ref{cor:scc} are rather different. For Theorem~\ref{cor:scc.small.omg}, since when $\omg_0 \in (2,3)$, $\mathfrak L_{(\omg_0)\infty}=\mathfrak L_{(\omg_0)0}$ and $\mathfrak L_{(\omg_0)\infty}'=\mathfrak L_{(\omg_0)0}'$ depend only on the initial data, it suffices to perturb the \emph{incoming part} of the data with a (precise) polynomial tail to modify $\mathfrak L_{(\omg_0)0}$ and $\mathfrak L_{(\omg_0)0}'$. In contrast, in the proof of Theorem~\ref{thm:instability} in \cite{LO.exterior}, we add a (smooth) compactly supported perturbation to the \emph{outgoing part} of the data to perturb the \emph{dynamically defined} quantity $\mathfrak{L}$, rather than $\mathfrak{L}_{(\omg_{0}) 0}$ on $\Sgm_{0}$. On a technical level, in the case $\omg_0\in (2,3)$, we can completely bypass Theorem~\ref{thm:L-stability}, which is in some sense the most technically involved step of the proof.
\end{remark}

\begin{remark}[Topology of the initial data]
We note explicitly that the difference between the statements (3) in Theorem~\ref{cor:scc} and (3') in Theorem~\ref{cor:scc.small.omg} is that in (3'), the perturbations are only close to $\overline{\Tht}$ with respect to $d_{k, \omg_{0}}$, instead of with respect to $d_{k,\omg}$ for arbitrary $\omg>0$. This is because in the case $\omg_0\in (2,3)$, we are forced to add in a polynomial tail to make $\mathfrak L_{(\omg_0)0}$ and $\mathfrak L'_{(\omg_0)0}$ non-zero. Any perturbation with a faster decay rate as $r \to \infty$ \emph{cannot} alter $\mathfrak{L}_{(\omg_{0}) \infty}$ and $\mathfrak L'_{(\omg_0)0}$, which restricts the topology in (3') to $d_{k, \omg_{0}}$.
\end{remark}

\subsection{Specializing to a (small) neighborhood of Reissner--Nordstr\"om}\label{sec.instab.RN}
Since there exist Cauchy hypersurfaces in subextremal Reissner--Nordstr\"om spacetimes with non-vanishing charge for which the induced data are future-admissible (cf. Example~\ref{ex:RN}) and the admissibility condition is a stable property (cf. Remark~\ref{rmk:adm-stab}), it follows a fortiori from Theorem~\ref{cor:scc} that there exist arbitrarily small perturbations of Reissner--Nordstr\"om data such that the maximal globally hyperbolic future developments are $C^2$-future-inextendible.

In fact, more can be said in the case of solutions arising from small perturbations of Reissner--Nordstr\"om data. First, a result of Dafermos \cite{D3} shows that the interior of the black hole (and similarly the white hole) has a global bifurcate Cauchy horizon and the boundary does not contain any spacelike portion. (In the language of Theorem~\ref{thm:kommemi} in Section~\ref{sec:mghd} below, this means that $\mathcal S=\emptyset$.) Second, though perhaps not so physically relevant, one can use the main theorem together with Cauchy stability to obtain the following result on the instability of Reissner--Nordstr\"om, which is global to the future and to the past:

\begin{corollary}[Instability of Reissner--Nordstr\"om]\label{cor:RN.instab}
Given $(M, \e)$ with $0<|\e|<M$, there exists a sequence $(\Tht_i)_{i\in \mathbb N}$ of smooth perturbations of Reissner--Nordstr\"om data with parameters $M$ and $\e$ such that
\begin{itemize}
\item $\Tht_i\in \calA \calI \calD (\omg)\cap C^\infty_{\omg}$ for all $\omg > 2$, for all $i\in \mathbb N$ (cf. Definition~\ref{def:adm-data} and Remark~\ref{rmk.Cinfty}).
\item The initial data for $\phi$ and $\dot{\phi}$ for $\Tht_i$ are compactly supported for all $i\in \mathbb N$. Moreover, ``the support is uniformly bounded for all $i$'' in the sense that $\displaystyle\sup_{i\in \mathbb N}\sup_{supp (\phi)\cup supp(\dot{\phi})} r<\infty$.
\item As $i\to \infty$, $\Tht_i\to \Tht_{RN, M, \e}$ with respect to the distance $d_{\omg}$ for any $\omg > 2$ (cf. Definition~\ref{Ck.def} and \eqref{Cinfty.dist}). 
\item (Dafermos \cite{D3}) For each $i \in \bbN$, the maximal globally hyperbolic future-and-past development arising from $\Tht_i$ has the identical Penrose diagram as a subextremal Reissner--Nordstr\"om spacetime, as depicted in Figure~\ref{fig:RNadm}. In particular, it has global bifurcate Cauchy horizons both to the future and to the past.
\item For every $i\in \mathbb N$, the maximal globally hyperbolic future-and-past development are $C^2$-\underline{future}-and-\underline{past}-inextendible.
\end{itemize}
\end{corollary}

\section{Maximal globally hyperbolic future development of admissible Cauchy data} \label{sec:mghd}
We now begin discussing the ingredients of the proof of Theorem~\ref{cor:scc} (and Theorem~\ref{cor:scc.small.omg} and Corollary~\ref{cor:RN.instab}). In the next few sections, we will build up to the proof in Section~\ref{sec.pf.SCC}. In this section, we collect some known results regarding the maximal globally hyperbolic future development of admissible Cauchy data. This corresponds to Steps~1 and 2(a) in the discussion in Section~\ref{sec.main.structure} in the introduction.

We begin with a preliminary characterization of the future boundary of the maximal globally hyperbolic future development of an admissible data set can be obtained\footnote{The original theorem of Kommemi applies in the case with one asymptotically flat end. The ideas, however, can be applied in our setting, see discussions in \cite{D3}.} from the work of Kommemi \cite{Kommemi}. 
\begin{theorem} \label{thm:kommemi}
Let $(\calM, g, \phi, F)$ be the maximal globally hyperbolic future development of an admissible Cauchy initial data set (with arbitrary $\omg_{0} >2$, cf. Definition~\ref{def:adm-data}), and denote by $(\calQ = \calM / SO(3), g_{\calQ})$ the quotient Lorentzian manifold. Then the following statements hold:
\begin{enumerate}
\item $(\mathcal Q,g_{\mathcal Q})$ can be conformally embedded\footnote{Note that this is equivalent to $(\mathcal Q,g_{\mathcal Q})$ having a global system of double null coordinates, cf. \eqref{SS.metric.2}.} into a bounded subset of $\mathbb R^{1+1}$. 
\item Let $\mathcal Q^+$ be the closure of $\mathcal Q$ with respect to the topology induced by the conformal embedding described in part (1). Then the boundary\footnote{We abuse notation slightly to name the image of $\mathcal Q$ under the conformal embedding also as $\mathcal Q$. We will similarly do this for subsets of $\mathcal Q$, such as $\Sigma_0$.} of $\mathcal Q$ in $\mathcal Q^+$ has the following components: 
\begin{enumerate}
\item The initial hypersurface $\Sigma_0$.
\item Spatial infinities $i^0_1$ and $i^0_2$ which are the end-points of $\Sigma_0$ in $\mathcal Q^+$, with the convention that $i^0_1$ is the end-point with $\rho\to \infty$ and $i^0_2$ is the end-point with $\rho\to -\infty$.
\item Two connected components of null infinity, denoted by $\mathcal I^+_1$ and $\mathcal I^+_2$ respectively, each of which is an open null segment\footnote{The fact that it is open and that $r$ does not diverge to $\infty$ along $\EH_1$ and $\EH_2$ (see Definition~\ref{def.EH} below), follows from \cite{DafTrapped}.}, defined as the part of the boundary such that the $r$ diverges to $\infty$ along a transversal null curve towards $\mathcal I^+_1$ and $\NI_2$. 
\item Timelike infinities $i^+_1$ and $i^+_2$, which are defined to be future end-points of $\mathcal I^+_1$ and $\mathcal I^+_2$ respectively.
\item The Cauchy horizons \footnote{In the general setting of \cite{Kommemi}, $\CH_1$ and $\CH_2$ may be empty. Nevertheless, it is non-empty in our setting thanks to Theorem \ref{main.theorem.C0.stability}.} $\CH_1$ and $\CH_2$, which are defined to be half open\footnote{Both $\CH_1$ and $\CH_2$ are chosen to include their future endpoints. Therefore, in the case where $\mathcal S$ is empty (or contains only a single point of $\mathcal Q$), by our convention the bifurcation sphere is part of both $\CH_1$ and $\CH_2$.} null segments emanating from future null infinities $\mathcal I^+_1$ and $\mathcal I^+_2$ respectively such that the area-radius function $r$ extends continuously to $\CH_1\cup \CH_2$ and is strictly positive except possibly at the future endpoints of $\CH_1$ or $\CH_2$.
\item A (possibly empty) achronal set\footnote{In \cite{Kommemi}, Kommemi further distinguishes the sets for which $r$ extends to $0$ into null segments emanating from the endpoints of $\CH_1$ or $\CH_2$ and another piece which does not intersect any null rays emanating from future null infinity. We do not need such distinction here and will simply consider one achronal set $\mathcal S$ on which $r$ extends to $0$.} $\mathcal S$ which is defined to be the subset of the boundary on which $r$ extends continuously to $0$.
\end{enumerate}
Moreover, $\mathcal Q^+$ can be given by the Penrose diagram in Figure~\ref{fig:Kommemi}.
\end{enumerate}
\end{theorem}

\begin{figure}[h] 
\begin{center}
\def\svgwidth{250px}
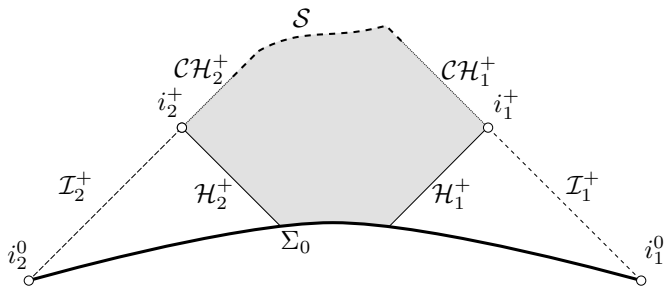 
\caption{Penrose diagram of the maximal globally hyperbolic future development. The black hole interior is the shaded region. The black hole exterior consists of two connected components in white.}  \label{fig:Kommemi}
\end{center}
\end{figure}

\subsection{Subextremality of the event horizons}
Given the boundary characterization, one can define the event horizons, see Definition~\ref{def.EH}. We will recall a result of Kommemi that under the admissibility condition, the event horizons are \emph{subextremal}, see Proposition~\ref{prop.subextremality}. The most important consequence of the subextremality statement is that the Price's law decay of Dafermos--Rodnianski applies. We will discuss this in the next subsection.

We now begin with the definition of the event horizons. For this, introduce the notation $\ub(v)$ and $\underline{v}(u)$ so that $(u,\underline{v}(u)), (\ub(v),v)\in \Sigma_0$.
\begin{definition}[Event horizons]\label{def.EH}
Given the maximal globally hyperbolic future development of an admissible Cauchy initial data set (with arbitrary $\omg>2$, cf.~Definition~\ref{def:adm-data}), define the \emph{event horizon} $\EH_1:=\{(u,v): u=u_{\EH_1},\, v\geq \underline{v}(u_{\EH}) \}$, where $u_{\EH_1}:=\sup \{u: \sup r(u', \cdot)=\infty \mbox{ for all }u' < u\}$ and $\underline{v}(u_{\EH})$ is such that $(u_{\EH},\underline{v}(u_{\EH}))\in \Sigma_0$. Abusing notation, we will also refer to the set $\EH_1\times \mathbb S^2\subset \mathcal M$ as the \emph{event horizon}. Notice that this is well-defined by Theorem~\ref{thm:kommemi} and can be viewed as the past-directed null curve emanating from $i^+_1$.

We also define the event horizon $\EH_2$ emanating from $i^+_2$ in a completely analogous manner, namely, $\EH_2:=\{(u,v): v=v_{\EH_2},\, u\geq \ub(v_{\EH_2}) \}$, where $v_{\EH_{2}}:=\sup \{v: \sup r(\cdot, v')=\infty \mbox{ for all }v' < v\}$ and $\ub(v_{\EH_2})$ is such that $(\ub(v_{\EH_2}),v_{\EH_2})\in \Sigma_0$.
\end{definition}

Importantly, the event horizons $\EH_1$ and $\EH_2$ are \emph{subextremal}, in the sense that $\sup \varpi(u_{\EH_1}, \cdot) >|{\bf e}|$ and $\sup{ \varpi(\cdot,v_{\EH_2})} >|{\bf e|}$. This is a result in the PhD thesis of Kommemi \cite{KomThe}. We include a proof in  Appendix~\ref{sec.subext.pf} for the convenience of the reader.
\begin{proposition}[Subextremality of the event horizons, Kommemi \cite{KomThe}]\label{prop.subextremality}
The following strict inequalities hold:
$$\sup \varpi(u_{\EH_1}, \cdot)>|{\bf e}|,\quad \sup \varpi(\cdot, v_{\EH_{2}})>|{\bf e}|.$$
\end{proposition}

\subsection{Price's law decay theorem of Dafermos--Rodnianski}
By the subextremality assertion (Proposition~\ref{prop.subextremality}), the Price's law decay theorem due to Dafermos--Rodnianski is applicable. For our purposes in this paper, it suffices to only state a subset of the estimates that are proven in \cite{DRPL}.
\begin{theorem}[Price's law decay, Dafermos--Rodnianski \cite{DRPL}] \label{thm:DR-decay}
Let $(\mathcal M, g, \phi, F)$ be the maximal globally hyperbolic future development of an admissible Cauchy initial data set with some $\omg_{0} >2$ (cf. Definition~\ref{def:adm-data}). Let $(u,v)$ be a coordinate system defined in connected component of the exterior region bounded in the future by $\EH_1$ and $\NI_1$ such that $v=\infty$ at\footnote{This statement, and the statement $\EH_1=\{(u,v):u=\infty\}$ is to be understood in terms of the conformal embedding of $\mathcal Q$ given in Theorem~\ref{thm:kommemi}, i.e., for every fixed $u$ and a sequence $v_i\to \infty$, $(u,v_i)$ corresponds to a sequence of points converging to $\NI_1$ with respect to the topology induced by the conformal embedding.} $\NI_1$ and $u=\infty$ at $\EH_1$ with the normalizations
\begin{equation}\label{PL.normal}
\lim_{v\to \infty}\f{(-\rd_u r)}{1-\mu}(u,v)=1,\quad C_0^{-1}\leq \lim_{u\to \infty}\f{\rd_v r}{1-\mu}(u,v)\leq C_0, 
\end{equation}
for some constants $C_0$.
Then, for every $\eta>0$, there exists some $B>0$ depending on $C_0$, $\eta$, as well as the solution, such that the following (future) decay estimates hold along the event horizon $\EH_1$:
$$|\phi|(\infty,v)+|\rd_v\phi|(\infty,v)\leq B v^{-\min\{\omg,3\}+\eta}\quad\mbox{for }v\geq 1$$
and the following (future) decay estimates hold along future null infinity $\NI_1$:
$$r|\phi|(u,\infty)\leq B u^{-\min\{\omg, 3\}+1}\quad\mbox{for }u\geq 1.$$
\end{theorem}
An analogous statement obviously holds in the exterior region bounded in the future by $\EH_2$ and $\NI_2$ with obvious modifications.

\begin{remark}[Completeness of null infinity and event horizon]\label{rmk.completeness}
It is implicit in Theorem~\ref{thm:DR-decay} that both $\NI_1$ and $\EH_1$ are complete\footnote{See \cite{DafTrapped} for a definition of the completeness of null infinity}. This follows from the existence of double null coordinates $u, v$ which satisfy the normalization \eqref{PL.normal} and have infinite range, as required by Theorem~\ref{thm:DR-decay}.
\end{remark}

\begin{remark}[Applicability of Theorem~\ref{thm:DR-decay} in our setting]
At this point, it may not be entirely obvious that the results in \cite{DRPL}, which are originally stated for a characteristic initial value problem, are applicable in our setting where we start with a Cauchy initial data set. However, using asymptotic flatness, one can indeed reduce the problem to the characteristic initial value problem considered in \cite{DRPL}. Moreover, the existence of a double null coordinate system $(u, v)$ satisfying the conditions of Theorem~\ref{thm:DR-decay} (and hence completeness of $\NI_{1}$ and $\EH_{1}$ by Remark~\ref{rmk.completeness}) can be proven. Detailed proofs of these statements can be found in \cite[Section~5]{LO.exterior}. 
\end{remark}

\section{Main theorems for the interior region proven in this paper}\label{sec.precise}
In this section, we state the main theorems for the interior region. These are the precise versions of the results corresponding to Steps~2(b), 3(b), 3(c) and 5 in Section~\ref{sec.main.structure}. \textbf{These results will be proven in Section~\ref{sec.main.theorem.C0.stability}, \ref{sec.blow.up}, \ref{sec.nonpert} and \ref{sec.main.theorem.C2} respectively.} We also refer the reader to the beginning of each of those sections for brief discussions of the ideas of the proofs.

\subsection{Geometry of Reissner--Nordstr\"om interior} \label{sec.geometry}
In this subsection, we briefly digress to discuss the geometry of the interior of Reissner--Nordstr\"om. This will provide the language and notations for some of the statements in the remainder of the paper.

Let ${\bf e}$ and $M$ be real numbers such that $0<|{\bf e}|<M$. \emph{The interior of the Reissner--Nordstr\"om black hole} is the manifold $\mathbb R^2\times \mathbb S^2$ together with the metric \eqref{RN.metric}:
\begin{equation*}
g=- \left( 1-\f {2M} r+\f{{\bf e}^2}{r^2} \right) \,dt^2+ \left( 1-\f {2M} r+\f{{\bf e}^2}{r^2}\right)^{-1} \, dr^2+r^2 d\sigma_{\mathbb S^2},
\end{equation*}
where $r$ ranges over $(r_-,r_+)$ with $r_\pm:=M\pm\sqrt{M^2-{\bf e}^2}$ and $t$ ranges over $(-\infty,\infty)$. We will attach the event horizon $\mathcal H^+_{total} = \mathcal{H}^{+}_{1} \cup \mathcal{H}^{+}_{2}$ and Cauchy horizon $\mathcal{CH}^+_{total} = \mathcal{CH}^{+}_{1} \cup \mathcal{CH}^{+}_{2}$ as boundaries of the interior of the Reissner--Nordstr\"om black hole such that $(\mathbb R^2\times \mathbb S^2)\cup \mathcal{H}_{total}^{+} \cup \mathcal{CH}_{total}^{+}$ is a manifold with corners (see Section \ref{sec.horizons} and the shaded region in Figure \ref{fig:RN}).

In Section~\ref{sec.null.1}, we will introduce a set of null coordinates $(u,v)$ and put the metric in the form \eqref{SS.metric.1}, \eqref{SS.metric.2} as discussed in Section~\ref{sec.SS}. We then introduce the notions of the event horizon and the Cauchy horizon in Section~\ref{sec.horizons}. In Sections~\ref{sec.null.2} and \ref{sec.null.3}, we will introduce two more systems of null coordinates, namely, the $(U,v)$ and $(u,V)$ coordinate system, which are regular at the event horizon and Cauchy horizon respectively. In the paper (especially in the proof of Theorem \ref{main.theorem.C0.stability}), we will use all of these coordinate systems to compare the metric of the spacetime solution in question with that of the Reissner--Nordstr\"om spacetime.
 
\subsubsection{The $(u,v)$ coordinate system}\label{sec.null.1}
We first define the $r^*$ coordinate in the interior of the Reissner--Nordstr\"om black hole:
$$r^*=r+ \left( M+\frac{2M^2-{\bf e}^2}{2\sqrt{M^2-{\bf e}^2}} \right) \log (r_+-r) + \left( M-\frac{2M^2-{\bf e}^2}{2\sqrt{M^2-{\bf e}^2}} \right)\log (r-r_-).$$
Define then the null coordinates
$$v=\frac 12(r^*+t),\quad u=\frac 12(r^*-t),$$
which implies
$$\frac{\rd}{\rd v}= \frac{\rd}{\rd r^*}+\frac{\rd}{\rd t},\quad\frac{\rd}{\rd u}= \frac{\rd}{\rd r^*}-\frac{\rd}{\rd t}.$$
According to \eqref{RN.metric}, in this coordinate system, the Reissner--Nordstr\"{o}m metric takes the form
$$g_{RN}=-\f{\Omg_{RN}^2}{2} (du\otimes dv+dv\otimes du)+ r_{RN}^2 d\sigma_{\mathbb S^2},$$
where $\Omg_{RN}^2=-4\left(1-\frac{2M}{r_{RN}}+\frac{{\bf e}^2}{r_{RN}^2}\right)$. Here and below, we use the notation $r_{RN}(u,v)$ to denote the Reissner--Nordstr\"om area radius function $r$ in the $(u,v)$ coordinates that we just defined. The following holds:
\begin{equation}\label{int.r}
\rd_v r_{RN}=\rd_u r_{RN}=1-\f{2M}{r_{RN}}+\f{{\bf e}^2}{r_{RN}^2}.
\end{equation}
Define $\kappa_+>0$ and $\kappa_->0$ to be\footnote{Note that this is different from the definition in \cite{LO.instab} in which $\kappa_-$ is taken to be negative.} 
$$\kappa_+=\f{r_+-r_-}{2 r_+^2},\quad \kappa_-=\f{r_+-r_-}{2 r_-^2}.$$
The coordinate $r^*$ can then be alternatively written as
$$r^*=r_{RN}+\f{1}{2\kappa_+}\log (r_+-r_{RN}) -\f{1}{2\kappa_-}\log (r_{RN}-r_-).$$
We compute that when $r_{RN}$ is close to $r_+$, we have
$$r_+-r_{RN}=e^{-2\kappa_+ r_+}(r_+-r_-)^{\f{\kappa_+}{\kappa_-}}e^{2\kappa_+ r^*}(1+O(r_+-r_{RN})).$$
In other words, as $r_{RN}\to r_+$,
\begin{equation} \label{eq:Omg-H}
\f 14\Omg_{RN}^2=-\rd_u r_{RN}=-\rd_v r_{RN}=\f{e^{-2\kappa_+ r_+}(r_+-r_-)^{1+\f{\kappa_+}{\kappa_-}}}{r_+^2}e^{2\kappa_+ r^*}(1+O(r_+-r_{RN})).
\end{equation}
On the other hand, when $r_{RN}$ is close to $r_-$, we have
$$r_{RN}-r_-=e^{2\kappa_-r_-}(r_+-r_-)^{\f{\kappa_-}{\kappa_+}}e^{-2\kappa_-r^*}(1+O(r_{RN}-r_-)).$$
As a consequence, as $r_{RN}\to r_-$,
\begin{equation} \label{eq:Omg-CH}
\f 14\Omg_{RN}^2=-\rd_u r_{RN}=-\rd_v r_{RN}=\f{e^{-2\kappa_-r_-}(r_+-r_-)^{1+\f{\kappa_-}{\kappa_+}}}{r_-^2}e^{-2\kappa_-r^*}(1+O(r_{RN}-r_-)).
\end{equation}
Moreover, for every $r_1$, $r_2$ such that $r_-<r_1<r_2<r_+$, there exists $C_{r_1,r_2,M,{\bf e}}$ depending on $r_1$, $r_2$, $M$ and ${\bf e}$ such that the following holds whenever $r_{RN}(u,v)\in [r_1,r_2]$:
$$\f 14\Omg_{RN}^2(u,v)=-\rd_u r_{RN}(u,v)=-\rd_v r_{RN}(u,v)\leq C_{r_1,r_2,M,{\bf e}}.$$

\subsubsection{The event horizon and the Cauchy horizon}\label{sec.horizons}
In this subsection, we introduce the event horizon and the Cauchy horizon. For this purpose, we will need to introduce a few additional sets of null coordinates. Define the functions $U_{\EH}(u)$, $U_{\CH}(u)$, $V_{\EH}(v)$ and $V_{\CH}(v)$ which are smooth and strictly increasing functions of their arguments and satisfy the following ODEs:
\begin{equation}\label{U.EH.def}
\f{dU_{\EH}}{du}=e^{2\kappa_+ u} \mbox{ and }\, U_{\EH}(u)\to 0 \mbox{ as }u\to -\infty;
\end{equation}
\begin{equation}\label{U.CH.def}
\f{dU_{\CH}}{du}=e^{-2\kappa_- u} \mbox{ and }\, U_{\CH}(u)\to 1 \mbox{ as }u\to +\infty;
\end{equation}
\begin{equation}\label{V.EH.def}
\f{dV_{\EH}}{dv}=e^{2\kappa_+ v}\mbox{ and }V_{\EH}(v)\to 0\mbox{ as }v\to -\infty;
\end{equation}
\begin{equation}\label{V.CH.def}
\f{dV_{\CH}}{dv}=e^{-2\kappa_- v}\mbox{ and }V_{\CH}(v)\to 1\mbox{ as }v\to +\infty.
\end{equation}
In the $(U_{\EH},V_{\EH})$ coordinate system, we attach the boundaries $\EH_1:=\{U_{\EH}=0\}$ and $\EH_2:=\{V_{\EH}=0\}$. Denote also the \emph{event horizon} as $\EH_{total}=\EH_1\cup\EH_2$. In the $(U_{\CH},V_{\CH})$ coordinate system, we attach the boundaries $\CH_1:=\{V_{\CH}=1\}$ and $\CH_2 :=\{U_{\CH}=1\}$. Denote also the \emph{Cauchy horizon} as $\CH_{total}=\CH_1\cup\CH_2$. Notice that in our convention, the \emph{bifurcation sphere of $\EH_{total}$}, which is given by $\{(U_{\EH},V_{\EH}):U_{\EH}=V_{\EH}=0\}$ belongs to both $\EH_1$ and $\EH_2$ (in fact, it is precisely $\calH_{1}^{+} \cap \calH_{2}^{+}$). Similarly for the the \emph{bifurcation sphere of $\CH_{total}$}. Notice that the metric extends smoothly to the boundaries. We refer the reader to the shaded region of Figure~\ref{fig:RN} for a depiction of these boundaries.

\subsubsection{The $(U,v)$ coordinate system}\label{sec.null.2}
Here and in Section \ref{sec.null.3}, we will compute the metric in two more systems of null coordinates. In particular, these coordinate systems will be regular near $\EH_1$ and $\CH_{1}$ respectively. In view of the fact that Theorems \ref{main.theorem.C0.stability} and \ref{final.blow.up.step} are stated and proved in the neighbourhood of one componenent of the event horizon and one component of the Cauchy horizon, without loss of generality, we will consider a neighbourhood of $\EH_1$ and $\CH_1$. {\bf In particular, it will be convenient to introduce the convention that $\EH$ and $\CH$ (without subscripts) refer to $\EH_1$ and $\CH_1$ respectively and that $U=U_{\EH}$ and $V=V_{\CH}$.} We will often abuse terminology to also call $\EH$ and $\CH$ the \emph{event horizon} and the \emph{Cauchy horizon} respectively.

Using the above convention, $U=U_{\EH}$ is given by \eqref{U.EH.def}, i.e.,
\begin{equation}\label{U.def}
\f{dU}{du}=e^{2\kappa_+ u} \mbox{ and }\, U(u)\to 0 \mbox{ as }u\to -\infty.
\end{equation}
To distinguish the metric in this coordinate system to that in the $(u,v)$ coordinate, we use an extra index $\mathcal H$. The metric in the coordinate system $(U,v)$ takes the form
$$g_{RN,\mathcal H}=-\f{\Omg_{RN,\mathcal H}^2}{2} (dU\otimes dv+dv\otimes dU)+ r_{RN}^2 d\sigma_{\mathbb S^2},$$
where 
$$\Omg_{RN,\mathcal H}^2=-\frac{4}{2 \kpp_{+}}\left(1-\f{2M}{r_{RN}}+\f{{\bf e}^2}{r_{RN}^2}\right)U^{-1}.$$ 
Notice that for every fixed $v$, $\lim_{U\to 0} \Omg_{RN,\mathcal H}^2(U,v)$ is non-zero, i.e., $(U,v)$ is a regular coordinate system near the event horizon. In fact, we have
$$\lim_{U\to 0} \Omg_{RN,\mathcal H}^2(U,v)=\f{4e^{-2\kappa_+ r_+}(r_+-r_-)^{1+\f{\kappa_+}{\kappa_-}}}{r_+^2}e^{2\kappa_+ v}.$$
Moreover, for every $r_0\in (r_-,r_+)$, we have the following estimates if $r_{RN}(U,v)\in [r_{0}, r_{+}]$:
\begin{equation}\label{RN.H.bound}
\Omg_{RN,\mathcal H}^2+|\rd_U r_{RN}|\leq C_{r_0, M,{\bf e}}e^{2\kappa_+ v},\quad |\rd_v r_{RN}|\leq C_{r_0, M,{\bf e}}Ue^{2\kappa_+ v}
\end{equation}
where $C_{r_0,M,{\bf e}}>0$ is a constant depending only on $r_0$ and the parameters of the Reissner--Nordstr\"om spacetime $M$ and ${\bf e}$.
\subsubsection{The $(u,V)$ coordinate system}\label{sec.null.3}
Finally, we introduce the coordinate system $(u,V)$ by defining $V=V_{\CH}$ as in \eqref{V.CH.def}, i.e.,  
\begin{equation}\label{V.def}
\f{dV}{dv}=e^{-2\kappa_- v}\mbox{ and }V(v)\to 1\mbox{ as }v\to \infty.
\end{equation}
We use the index $\mathcal C\mathcal H$ to denote the metric in this coordinate system, i.e.
$$g_{RN,\mathcal C\mathcal H}=-\f{\Omg_{RN,\mathcal C\mathcal H}^2}{2} (du\otimes dV+dV\otimes du)+ r_{RN}^2 d\sigma_{\mathbb S^2},$$
where 
$$\Omg_{RN,\mathcal C\mathcal H}^2=-\frac{4}{2 \kpp_{-}}\left(1-\f{2M}{r_{RN}}+\f{{\bf e}^2}{r_{RN}^2}\right) (1-V)^{-1}.$$ 
In this coordinate system, the $\CH$ is given by $\{V=1\}$. Notice that in Reissner--Nordstr\"om, $(u,V)$ is a smooth coordinate system up to the Cauchy horizon.

In the proofs of Theorems \ref{main.theorem.C0.stability} and \ref{thm.nonpert}, we will also show that the metric extends continuously to $\{V=1\}$ in the $(u,V)$ coordinate system. On the other hand, we will of course not show that the metric extends smoothly as in the Reissner--Nordstr\"om case - indeed, in view of Theorems \ref{final.blow.up.step} and \ref{main.theorem.C2}, this is false for a large class of initial data.

\subsection{Stability and instability of the Cauchy horizon}

We first state the stability theorem (cf. Step~2(b) in Section~\ref{sec.main.structure}), which, except for the precise quantitative rates, was first proven by Dafermos in \cite{D2}. The theorem is formulated below in a gauge that is most convenient for deriving estimates in this paper. We will then compare this with an ``equivalent gauge''\footnote{Here, we say that two null coordinates $v$ and $v'$ are equivalent if $c_0^{-1}<\f{dv}{dv'}(v')<c_0$ for some $c_0>0$.} in Remark~\ref{gauge.equiv}, which is more convenient for applying Theorem~\ref{thm:DR-decay}. 

To give the statement of the theorem, we will use three coordinate systems: $(u,v)$, $(U,v)$ and $(u,V)$. They should be thought of as analogues of the corresponding coordinate systems on the Reissner--Nordstr\"om spacetime introduced in Sections~\ref{sec.null.1}, \ref{sec.null.2} and \ref{sec.null.3} respectively.

We will choose our coordinates as follows. First, we pick the coordinate system $(U,v)$, where $U$ and $v$ are normalized by the gauge conditions \eqref{gauge.1} and \eqref{gauge.2} below. We then define $u(U)$ and $V(v)$ which relate to $U$ and $v$ according to \eqref{U.def} and \eqref{V.def} respectively.

\begin{theorem}\label{main.theorem.C0.stability}
Fix $M$, ${\bf e}$ and $s$ such that $0<|{\bf e}|<M$ and $s>1$. Consider the characteristic initial value problem (cf. Proposition~\ref{prop.CIVP.LWP}) with $\e\neq 0$ and data given on $C_{-\infty}$ and $\underline{C}_1$ satisfying the constraint equations \eqref{eqn.Ray} such that the following hold for some $E>0$: 
\begin{enumerate}
\item $C_{-\infty}:=\{(U,v): U=0,\, v\geq 1\}$ (which will be viewed as the event horizon $\EH$) is an affine complete null hypersurface approaching a subextremal Reissner--Nordstr\"om event horizon with $0<|{\bf e}|<M$. More precisely, in the gauge\footnote{Here, $r_\pm$ and $\kappa_\pm$ are defined with respect to the fixed parameters $M$ and ${\bf e}$.}\footnote{To see that one can indeed choose such a gauge, see Remark~\ref{gauge.equiv} and Appendix~\ref{sec:appendix.gauge}.} 
\begin{equation}\label{gauge.1}
\Omg^2_{\mathcal H}=\f{4e^{-2\kappa_+ r_+}(r_+-r_-)^{1+\f{\kappa_+}{\kappa_-}}}{r_+^2}e^{2\kappa_+ v},
\end{equation}
we have
\begin{enumerate}
\item $r\to M+\sqrt{M^2-{\bf e}^2}$ as $v\to \infty$;
\item the following decay rate holds for all $v\geq 1$ for the scalar field and its derivative:
\begin{equation}\label{phi.EH.decay}
|\phi|(0,v)+|\rd_v\phi|(0,v)\leq E v^{-s};
\end{equation}
\end{enumerate}
\item On $\underline{C}_1:=\{(U,v): v=1,\, U\leq U_0\}$ (where $U_0>0$), after normalizing $U$ by 
\begin{equation}\label{gauge.2}
\rd_U r=-1,
\end{equation}
the following holds for all $U\leq U_0$:
\begin{equation}\label{rdUphi.initial}
|\rd_U\phi|(U,1)\leq E.
\end{equation}
\end{enumerate}
Then, by restricting to some nonempty, connected subset $\underline{C}'_1:=\{(U,v):v=1,\,U\leq U_s\}\subset \underline{C}_1$ for some $U_s>0$ sufficiently small, the globally hyperbolic future development of the data on $\underline{C}'_1\cup C_{-\infty}$ has a Penrose diagram given by Figure \ref{fig:maintheorem} (Section~\ref{sec.prev-works-nonlinear}). Moreover, in the $(u,V)$ coordinate system\footnote{We remind the reader again that once the $(U,v)$ coordinate system is defined, the $(u,v)$ and $(u,V)$ coordinate systems are given by the conditions \eqref{U.def} and \eqref{V.def}.}, we can attach the null boundary $\CH:=\{V=1\}$ to the spacetime such that the metric extends continuously to $\CH$. 

In addition, for $u_s=u(U_s)$, the following estimates hold in the $(u,v)$ coordinate system for $u<u_s$ and for some constant $C>0$ depending on $M$, ${\bf e}$, $E$ and $s$:
$$|\phi|(u,v)+|r-r_{RN}|(u,v)+|\log\Omg-\log\Omg_{RN}|(u,v) \leq C(v^{-s}+|u|^{-s+1}),$$
$$|\rd_v\phi|(u,v)+|\rd_v(r-r_{RN})|(u,v)+|\rd_v(\log\Omg-\log\Omg_{RN})|(u,v) \leq Cv^{-s}.$$
Furthermore, for every $A\in \mathbb R$, there exists $C>0$ depending on $A$, $M$, ${\bf e}$, $E$, $U_0$ and $s$ such that the following estimates hold in the $(u,v)$ coordinate system for $u<u_s$:
$$|\rd_u\phi|(u,v)+|\rd_u(r-r_{RN})|(u,v)+|\rd_u(\log\Omg-\log\Omg_{RN})|(u,v) \leq 
\begin{cases}
C\Omg_{RN}^2 v^{-s} &\mbox{for }u+v\leq A\\
C|u|^{-s} &\mbox{for }u+v\geq A.
\end{cases}
$$
\end{theorem}

\begin{remark}[Alternative formulation of the gauge condition]\label{gauge.equiv}
While the gauge condition on the event horizon as in Theorem~\ref{main.theorem.C0.stability} is suitable for proving estimates, it is less convenient to apply the theorem with this gauge condition. For that purpose, we note that one can alternatively consider the gauge condition $\f{\rd_v r}{1-\mu}=1$ on the set $C_{-\infty}$ (with the conditions on $\underline{C}_1$ remain unchanged). This will be discussed in detail in Proposition~\ref{prop:gauge} in Appendix~\ref{sec:appendix.gauge}.
\end{remark}

\begin{remark}[Assumptions in Theorem~\ref{main.theorem.C0.stability} \emph{always} hold in the global setting]\label{C0.stab.assumptions}
In fact, the assumptions in Theorem~\ref{main.theorem.C0.stability} \emph{always} hold for $C_{-\infty}=\EH_1$ (or for $\EH_2$, with appropriate changes of $u$ and $v$) with some $s>2$, and an arbitrary choice of a transversal initial null curve $\underline{C}_{1}$. We check this for $\EH_1$, as $\EH_2$ is similar. 

We first normalize $U$ by the gauge condition \eqref{gauge.2} on $\underline{C}_{1}$ (which will later coincide with $\set{(U, v): v=1, \, U \leq U_{0}}$). Note that this is possible because the future-admissibility condition in Definition~\ref{def:adm-data} guarantees that for any point on $\EH_1$, $r$ is decreasing towards the increasing $u$ direction; see Lemma~\ref{no.trapped} in Appendix~\ref{sec.subext.pf}.

Next, instead of (1), it suffices to verify the conditions in Proposition~\ref{prop:gauge}. Note that
\begin{enumerate}
\item $\EH_1$ is affine complete in view of Remark~\ref{rmk.completeness}.
\item $r$ and $\varpi$ have limits as in (b) and (c) in Remark~\ref{gauge.equiv} thanks to \eqref{r.varpi.limits} and \eqref{r.varpi.poly} in Appendix~\ref{sec.subext.pf}. Moreover, $\rd_{v} r \geq 0$ on $\EH_{1}$ (see Lemma~\ref{no.trapped}) and $\inf_{\EH_{1}} r > 0$ since $r$ is initially positive.
\item By Theorem~\ref{thm:DR-decay}, \eqref{phi.EH.decay} holds in the gauge $\f{\rd_v r}{1-\mu}=1$ .
\end{enumerate}
Moreover, by the $C^{1}$ regularity of the initial data, \eqref{rdUphi.initial} follows from standard local existence after a suitable translate of $v$.

Thus, as pointed out in celebrated work of Dafermos--Rodnianski \cite{DRPL}, one has \emph{unconditionally}  that the boundary of the maximal globally hyperbolic development has non-empty $\CH_1$ and $\CH_2$.

\end{remark}

Next, we state our instability theorem (cf. Step~3(b) in Section~\ref{sec.main.structure}):
\begin{theorem}\label{final.blow.up.step}
Assume that the assumptions of Theorem \ref{main.theorem.C0.stability} hold with $s>2$. If, moreover, there exists a \underline{non-integer}\footnote{This is simply to guarantee that $\alp_0>\alp'$.} \footnote{Notice that $\alp'$ cannot be too small, as it would be inconsistent with the assumptions in Theorem \ref{main.theorem.C0.stability}.} $\alpha'$ with $\alp_0:=\lceil \alp' \rceil\in [3, 4s-2)$ such that on $C_{-\infty}$
\begin{equation}\label{final.blow.up.step.assumption}
\int_{1}^\infty v^{\alpha'}(\rd_v\phi\restriction_{C_{-\infty}})^2(v)\, dv =\infty,
\end{equation}
then, after taking $u_s$ more negative if necessary, the following holds for all $u<u_s$:
\begin{equation}\label{blow.up.interior}
\int_{-u_s}^{\infty} \log_{+}^{\alpha_0}(\f 1{\Omg}) (\rd_v\phi)^2(u,v) dv=\infty,
\end{equation}
where $\log_+ x=\begin{cases}
\log x\quad\mbox{if }x\geq e\\
1\quad\mbox{otherwise}.
\end{cases}
$
Moreover, there exists $u_\lambda<u_s$ such that for every $u<u_\lambda$, the following blow up holds:
\begin{equation}\label{blow.up.interior.dvr}
\lim_{v\to \infty}\f{\lambda}{\Omg^2}(u,v)=-\infty.
\end{equation}
In particular, \eqref{blow.up.interior} implies that the scalar field is not in $W^{1,2}_{loc}$ in the $C^0$ extension constructed in Theorem \ref{main.theorem.C0.stability} and \eqref{blow.up.interior.dvr} implies that the metric is not in $C^1$ in the $C^0$ extension constructed in Theorem \ref{main.theorem.C0.stability}.
\end{theorem}

Theorems~\ref{main.theorem.C0.stability} and \ref{final.blow.up.step} give a fairly complete picture of nonlinear stability and instability properties of the Cauchy horizon in a region that is \emph{perturbative} in the sense that it is $C^{0}$ close to a Reissner--Nordstr\"om spacetime. Our next result combines these results with Theorem~\ref{thm:kommemi} to extend this picture to a suitable global non-perturbative setting (cf. Step~3(c) in Section~\ref{sec.main.structure}). For simplicity, we formulate the theorem only for $\CH_{1}$; an entirely symmetric statement holds for $\CH_{2}$.

\begin{theorem} \label{thm.nonpert}
Let $(\calM, g, \phi, F)$ be the maximal globally hyperbolic future development of an admissible Cauchy initial data set (with arbitrary $\omg_{0}>2$, cf. Definition~\ref{def:adm-data}). In a neighborhood of $\EH_1$ in the interior of the black hole, consider the null coordinates $(u, V)$ characterized by \eqref{gauge.1}, \eqref{V.EH.def} on $\EH_{1}$ for $V$, and \eqref{gauge.2}, \eqref{U.EH.def} for $u$. Denote by $p_{\CH_{1}} = (u_{\CH_{1}}, 1) \in \calQ^{+}$ the future endpoint\footnote{The existence of such a  $u_{\CH_{1}} \in (-\infty, \infty]$ is a straightforward consequence of the coordinate choice; see Lemma~\ref{lem.u-reg}.} of $\CH_{1}$. Then the metric components $\Omg^{2}(u, V)$ and $r(u, V)$, as well as the scalar field $\phi(u, V)$, extend continuously to $\CH_{1} \setminus \set{p_{\CH_{1}}} = \set{(u, V) : -\infty < u < u_{\CH_{1}}, V = 1}$.

Moreover, if the lower bound \eqref{final.blow.up.step.assumption} holds on $\EH_{1}$, then for every $u \in (-\infty, u_{\CH_{1}})$, the following blow up of $\rd_{V} \phi$ and $\rd_{V} r$ hold:
\begin{align} 
	\int_{0}^{1} \frac{(\rd_{V} \phi)^{2}}{\Omg^{2}}(u, V) d V = \infty, \label{eq:nonpert-blowup-phi} \\
	\lim_{V \to 1} \frac{\rd_{V} r}{\Omg^{2}}(u, V) = - \infty.		\label{eq:nonpert-blowup-dvr}
\end{align}
In particular, the scalar field is not in $W^{1,2}_{loc}$ and the metric is not in $C^1$ in the above $C^0$ extension obtained by adjoining $\set{(u, V) : -\infty < u < u_{\CH_{1}}, V = 1}$.
\end{theorem}
\begin{remark} \label{rem.nonpert-bfsph}
According to Theorem~\ref{thm:kommemi}, if $\calS = \emptyset$ then $\CH_{1} \cap \CH_{2}$ consists of a bifurcation sphere with positive area radius. In this case, the metric components and the scalar field extend continuously in the $(U_{\CH_{2}}, V_{\CH_{1}})$ coordinates to the bifurcate sphere $(U_{\CH_{2}}, V_{\CH_{1}}) = (1, 1)$, where $V_{\CH_{1}}$ coincides with $V$ in Theorem~\ref{thm.nonpert} and $U_{\CH_{2}}$ is defined analogously to $V$ but with respect to $\EH_{2}$, i.e., $v$ in \eqref{gauge.1} is replaced by $u$, and \eqref{V.CH.def} is replaced by \eqref{U.CH.def}. We refer to Remark~\ref{rem.nonpert-bfsph-pf} for a proof.
\end{remark}

\subsection{$C^2$-future-inextendibility}

The inextendibility statement for the interior region in Theorem~\ref{thm.nonpert} depends on the particular $C^{0}$ extension in the $(u, V)$ coordinates. Our final result asserts a geometric formulation of inextendibility of the whole maximal globally hyperbolic future development, namely \emph{$C^{2}$-future-inextendibility}, which is independent of such a choice (cf. Step~5 in Section~\ref{sec.main.structure}). Notice that since by Remark~\ref{C0.stab.assumptions}, the assumptions in Theorem~\ref{main.theorem.C0.stability} always hold with $s>2$ when considering the maximal globally hyperbolic future development, we do not need to state those assumptions explicitly in Theorem~\ref{main.theorem.C2}.
\begin{theorem}\label{main.theorem.C2}
Let $(\calM, g, \phi, F)$ be the maximal globally hyperbolic future development of an admissible Cauchy initial data set (with arbitrary $\omg_{0} > 2$, cf. Definition~\ref{def:adm-data}). Assume furthermore that the lower bound \eqref{final.blow.up.step.assumption} holds on $\EH_{1}$, as well as on $\EH_{2}$ (with $v$ replaced by $u$).
Then $(\calM, g)$ is future-inextendible with a $C^2$ Lorentzian metric.
\end{theorem}

\section{Main theorems for the exterior region proven in \cite{LO.exterior}}  \label{sec:main-thm}

We now discuss the main results for the exterior region, which correspond to Steps~3(a) and 4 in Section~\ref{sec.main.structure}. \textbf{All of the results in this section are proven in our companion paper \cite{LO.exterior}.}

In all the results in this section, the quantities $\mathfrak{L}_{(\omg_{0}) 0}$, $\mathfrak{L}$ and $\mathfrak{L}_{(\omg_{0}) \infty}$ play an important role. \textbf{All statements in this section concerning $\mathfrak{L}_{(\omg_{0}) 0}$, $\mathfrak{L}$ and $\mathfrak{L}_{(\omg_{0}) \infty}$ on one asymptotically flat end applies equally well to $\mathfrak{L}'_{(\omg_{0}) 0}$, $\mathfrak{L}'$ and $\mathfrak{L}'_{(\omg_{0}) \infty}$ (respectively) on the other asymptotically flat end.} We focus on the first asymptotically flat end (i.e., towards which $v$ and $-u$ increase) for the sake of concreteness.

Our first main theorem states that nonvanishing of $\mathfrak{L}_{(\omg_{0}) \infty}$ implies some integrated lower bound for the incoming radiation $\rd_{v} \phi$ along $\EH$. 

\begin{theorem} \label{thm:blowup}
For $\omg_{0} > 2$, let $\Tht = (r, f, h, \ell, \phi, \dot{\phi}, \e)$ be an $\omg_{0}$-admissible data set, and let $(\calM, g, \phi, F)$ be the corresponding maximal globally hyperbolic future development. Suppose that
\begin{equation*}
	\mathfrak{L}_{(\omg_{0}) \infty} \neq 0.
\end{equation*}
Then for an advanced null coordinate $v$ such that
\begin{equation*}
	C^{-1} < \inf_{\EH} \frac{\rd_{v} r}{1-\mu} \leq \sup_{\EH} \frac{\rd_{v} r}{1-\mu} < C
\end{equation*}
for some $C > 0$, we have
\begin{equation*}
	\int_{\EH} v^{\alp} (\rd_{v} \phi)^{2} \, \ud v = \infty
\end{equation*}
for every $\alp >  \min\{ 2 \omg_{0} + 1 , 7 \}$. 
\end{theorem}

Our next theorem asserts stability of the quantity $\mathfrak{L}$, which is the dynamically defined part of $\mathfrak{L}_{(\omg_{0}) \infty}$ in the case $\omg_{0} \geq 3$ (observe that $\mathfrak{L}_{(\omg_{0}) 0}$ is determined by the initial data). 
\begin{theorem} \label{thm:L-stability}
Fix $\omg_0>2$. Let $\Tht = (r, f, h, \ell, \phi, \dot{\phi}, \e)$ and $\overline{\Tht} = (\rbg, \fbg, \hbg, \ellbg, \phibg, \dphibg, \ebg)$ be $\omg_0$-admissible data sets (cf. Definition~\ref{def:adm-data}) such that $d_{1, \omg_0}^{+}(\Tht, \overline{\Tht}) < \eps$, where 
\begin{equation*}
\begin{aligned}
	d_{1, \omg_0}^{+} (\Tht, \overline{\Tht}) := 
	& \nrm{\brk{\rho_{+}} \log (f / \fbg)(\rho)}_{C^{0}} 
	+ \nrm{\brk{\rho_{+}}^{2} \rd_{\rho} \log (f / \fbg)(\rho)}_{C^{0}} 
	+ \nrm{\brk{\rho_{+}}^{2}  (h - \hbg) (\rho)}_{C^{0}} \\
	& + \nrm{\log^{-1}(1+\brk{\rho_{+}}) (r - \rbg)(\rho)}_{C^{0}} 
	+ \nrm{\brk{\rho_{+}} \rd_{\rho} (r - \rbg)(\rho)}_{C^{0}} 
	+ \nrm{( f \ell - \overline{f}\overline{ \ell})(\rho)}_{C^{0}}  \\
	& + \nrm{\brk{\rho_{+}}^{\omg_0} (\phi - \phibg)(\rho)}_{C^{0}} 
	+ \nrm{\brk{\rho_{+}}^{\omg_0+1} \rd_{\rho} (\phi- \phibg)(\rho)}_{C^{0}} 
	+ \nrm{\brk{\rho_{+}}^{\omg_0+1} (f \dot{\phi} - \overline{f} \overline{\dot{\phi}})(\rho)}_{C^{0}} 
	+ |{\bf e} - \ebg|.
\end{aligned}
\end{equation*}
Here, $\brk{\rho_{+}} := (1 + \rho_{+}^{2})^{1/2}$ and $\rho_{+} := \max\set{0, \rho}$.

Then, for $\mathfrak L:=\mathfrak L[\Tht]$ and $\overline{\mathfrak L}:=\mathfrak L[\overline{\Tht}]$, there exists a constant $C_{\overline{\Tht}}$, which depends only on $\overline{\Tht}$, such that
\begin{equation*}
	\Abs{\mathfrak{L} - \overline{\mathfrak{L}}} \leq C_{\overline{\Tht}} \, \eps.
\end{equation*}
\end{theorem}

Our final theorem concerns instability of the condition $\mathfrak{L}_{(\omg_{0}) \infty} = 0$. 
\begin{theorem} \label{thm:instability}
For $\omg_{0} \geq 3$, let $\overline{\Tht} = (\rbg, \fbg, \hbg, \ellbg, \phibg, \dphibg, \ebg)$ be an $\omg_{0}$-admissible data set. 
Suppose that
\begin{equation*}
	\mathfrak{L}_{(\omg_{0}) \infty} [\overline{\Tht}] = 0.
\end{equation*}
Then for some $\eps_{\ast} = \eps_{\ast}(\overline{\Tht}) > 0$, there exists a one parameter family $(\Tht_{\eps})_{\eps \in (-\eps_{\ast}, \eps_{\ast})}$ of $\omg_{0}$-admissible initial data sets such that 
\begin{itemize}
\item $\Tht_{0} = \overline{\Tht}$, 
\item $\mathfrak{L}_{(\omg_{0}), \infty}[\Tht_{\eps}] \neq 0$ for all $\eps \in (-\eps_{\ast}, \eps_{\ast})\setminus\{0\}$,
\item if $\overline{\Theta}\in C^k_{\omg_0}$ for $k\in\mathbb N$, $k\geq 2$, then $\eps \mapsto \Tht_{\eps}$ is continuous with respect to $d_{k, \omg}$ for all $\omg > 2$, 
\item $\mathfrak{L}_{(\omg_{0}) 0}[\Tht_{\eps}] = \mathfrak{L}_{(\omg_{0}) 0}[\overline{\Tht}]$ and $\mathfrak{L}_{(\omg_{0}) 0}'[\Tht_{\eps}] = \mathfrak{L}_{(\omg_{0}) 0}'[\overline{\Tht}]$ for all $\eps \in (-\eps_{\ast}, \eps_{\ast})$.
\end{itemize}
In fact, there exist $\bar{\rho}_{2} > \bar{\rho}_{1} \gg 1$ such that 
\begin{itemize}
\item $\Tht_{\eps} = \overline{\Tht}$ in $\set{\rho \in \Sgm_{0} : \rho < \bar{\rho}_{1}}$ for all $\eps \in (-\eps_{\ast}, \eps_{\ast})$,
\item denoting $\Tht_\eps = (r_\eps, f_\eps, h_\eps, \ell_\eps, \phi_\eps, \dot{\phi}_\eps, \e_\eps)$, it holds that
$$\overline{\phi}=\phi_\eps,\quad \fbg \overline{\dot{\phi}}=f_\eps \dot{\phi}_\eps $$
in $\set{\rho \in \Sgm_{0} : \rho > \bar{\rho}_{2}}$ for all $\eps \in (-\eps_{\ast}, \eps_{\ast})$.
\end{itemize}
\end{theorem}

\section{Proof of strong cosmic censorship (Theorem~\ref{cor:scc}, Theorem~\ref{cor:scc.small.omg} and Corollary~\ref{cor:RN.instab})}\label{sec.pf.SCC}

Combining the results in Sections~\ref{sec:mghd}, \ref{sec.precise} and \ref{sec:main-thm}, we now give the proof of strong cosmic censorship. We first prove the main strong cosmic censorship theorem, which holds for $\omg_0\geq 3$.

\begin{proof}[Proof of Theorem~\ref{cor:scc}]
\pfstep{Step 1} The inextendibility assertion (first statement) follows from Theorems~\ref{main.theorem.C2} and \ref{thm:blowup} (suitably applied to both asymptotically flat ends).

\pfstep{Step 2} The openness assertion (second statement) follows from Theorem~\ref{thm:L-stability} (where we take $\omg_{0}$ to be the smaller number among $\omg_{0}$ and $\omg$ in Theorem~\ref{cor:scc}.(2)) applied to both asymptotically flat ends. Indeed, Theorems~\ref{thm:L-stability} implies that $\mathfrak L$ and $\mathfrak L'$ are stable, while by assumption $\mathfrak L_{(\omg_0) 0}$ and $\mathfrak L_{(\omg_0)0}'$ are only allowed to be modified by a small perturbation, hence (recalling \eqref{eq:Linfty-def}) $\mathfrak L_{(\omg_0) \infty}$ and $\mathfrak L_{(\omg_0)\infty}'$ are also stable.

\pfstep{Step 3} The density assertion (third statement) is a mere restatement of Theorem~\ref{thm:instability} applied to both asymptotically flat ends.
\end{proof}

Next, we turn to the case where $\omg_0\in (2,3)$.
\begin{proof}[Sketch of proof of Theorem~\ref{cor:scc.small.omg}]
This is much simpler than the proof of Theorem~\ref{cor:scc} since for $\omg_0\in (2,3)$, $\mathfrak L_{(\omg_0)\infty}$ and $\mathfrak L'_{(\omg_0)\infty}$ are \underline{not} dynamical, i.e., $\mathfrak L_{(\omg_0)\infty}=\mathfrak L_{(\omg_0)0}$, $\mathfrak L'_{(\omg_0)\infty}=\mathfrak L'_{(\omg_0)0}$ can be computed from initial data alone.

The inextendibility assertion again follows immediately from Theorems~\ref{main.theorem.C2} and \ref{thm:blowup}. The openness assertion is obvious, in view of the fact $\mathfrak L_{(\omg_0)\infty}=\mathfrak L_{(\omg_0)0}$, $\mathfrak L'_{(\omg_0)\infty}=\mathfrak L'_{(\omg_0)0}$. Finally, for the density assertion, say, for $\mathfrak L_{(\omg_0)0}$, it suffices to construct a one-parameter family of initial data such that the incoming part $\rd_u(r\phi)=-\rd_\rho(r\phi) + f r \dot{\phi} + \frac{f \ell \phi}{r}$ (cf. Lemma \ref{lem:cauchy-to-char}) has an $\eps \rho^{-\omg_0}$ tail as $\rho\to \infty$. Such a one-parameter family of initial data (satisfying constraint) is easily seen to exist if one follows the ideas for solving the constraint equations in the proof of Theorem~\ref{thm:instability} in \cite{LO.exterior}. We omit the details.
\end{proof}

Finally, we prove the instability of both the future and past smooth Cauchy horizons of Reissner--Nordstr\"om.
\begin{proof}[Proof of Corollary~\ref{cor:RN.instab}]
Given $i\in \mathbb N$, we construct $\Tht_i$ as follows. By Theorem~\ref{thm:instability} (and considerations in Remark~\ref{rmk.Cinfty}), one can perturb a future-admissible hypersurface $\Sigma_0$ (i.e., one on which the induced data are future-admissible) in Reissner--Nordstr\"om to obtain a perturbation $\tilde{\Tht}_i\in \cap_{\omg>2} (\calA \calI \calD (\omg)\cap C^\infty_\omg)$ with compactly supported scalar field such that $d_3(\tilde{\Tht}_i,\Tht_{RN,M,\e})<2^{-i}$ and $\mathfrak L_{(3)\infty}[\tilde{\Tht}_i]\neq 0$ and $\mathfrak L_{(3)\infty}'[\tilde{\Tht}_i]\neq 0$. If the maximal globally hyperbolic future-and-past development is $C^2$-past-extendible, we take $\Tht_i=\tilde{\Tht}_i$. Otherwise, one can, by Cauchy stability, solve towards the past to obtain a past-admissible Cauchy data set. Both the openness and density assertions in Theorem~\ref{cor:scc} and Cauchy stability imply the existence of a $\Tht_i\in \cap_{\omg>2} (\calA \calI \calD (\omg)\cap C^\infty_\omg)$, again with compactly supported scalar field, such that 
\begin{itemize}
\item on $\Sigma_0$, $d_3(\Tht_i,\tilde{\Tht}_i)<2^{-i}$,
\item $\mathfrak L_{(3)\infty}[{\Tht}_i]\neq 0$ and $\mathfrak L_{(3)\infty}'[{\Tht}_i]\neq 0$
\item the maximal globally hyperbolic future-and-past development arising from $\Tht_i$ is $C^2$-future-and-past-inextendible.
\end{itemize}
We now show that the sequence $(\Tht_i)_{i\in \mathbb N}$ satisfies the desired properties.
\begin{itemize}
\item \emph{Assertion 1: $\Tht_i\in \cap_{\omg>2} (\calA \calI \calD (\omg)\cap C^\infty_\omg)$ for all $i\in \mathbb N$.} This holds by construction.
\item \emph{Assertion 2: Compact support of initial scalar field.} Note that the final assertion in Theorem~\ref{thm:instability} (when suitably applied to both asymptotically flat ends) implies that the support of the perturbations of $\phi$ and $\dot{\phi}$ is contained in a fixed large interval $[-\bar{\rho_2}, \bar{\rho}_2]$. Since the scalar field vanishes on Reissner--Nordstr\"om, the assertion follows.
\item \emph{Assertion 3: $\Tht_i\to \Tht_{RN, M, \e}$ as $i\to \infty$ with respect to the metric $d_\omg$.} For $\omg=3$, it follows from the construction (and triangle inequality) that $d_3(\Tht_i,\Tht_{RN,M,\e})<2^{-i+1}\to 0$. That this also holds for all $\omg>3$ follows from the (uniform) compact support of the scalar field.
\item \emph{Assertion 4: Existence of future-and-past bifurcate Cauchy horizons.} This is proven in \cite{D3}.
\item \emph{Assertion 5: Future-and-past inextendibility.} This holds by construction. \qedhere
\end{itemize}
\end{proof}

This concludes the proof of strong cosmic censorship, assuming the results in Sections~\ref{sec.precise} and \ref{sec:main-thm}. The remainder of this paper will be devoted to the proofs of the theorems in Sections~\ref{sec.precise}.

\section{Stability of the Cauchy horizon: Proof of Theorem \ref{main.theorem.C0.stability}}\label{sec.main.theorem.C0.stability}

In this section, we carry out the proof of Theorem \ref{main.theorem.C0.stability}. First, we begin with a brief discussion on the ideas of the proof in Section \ref{stab.ideas}. We then introduce the quantity $\Psi$ that we will bound in Section \ref{sec.data.bound}. In the same section, we derive the bounds for the metric components on the initial hypersurface. Then, in Section \ref{sec.strategy}, we will discuss the strategy of the proof, in which we will in particular introduce a partition of the spacetime into the red-shift and blue-shift regions\footnote{We comment on the nomenclature of these regions. The naming of the regions as the red-shift and blue-shift regions originates in \cite{D1, D2}. In these works, Dafermos estimated the gauge invariant derivatives of $\phi$ including $\f{\rd_u\phi}{\rd_u r}$ and $\f{\rd_v\phi}{\rd_v r}$ and these quantities tend to decay or grow in the red-shift or blue-shift region respectively. In the present work, since we bound coordinate derivatives of the scalar field, such growth or decay is not manifest. Nevertheless, it should be noted that our choice of the norms is inspired in part by the red-shift/blue-shift estimates in \cite{D1, D2}.}. We will then derive estimates for $\Psi$ separately in the red-shift and blue-shift regions in Sections \ref{sec.RS} and \ref{sec.BS} respectively. We will then prove the $C^0$ stability statement in Section \ref{sec.C0} and conclude the proof of Theorem \ref{main.theorem.C0.stability}. Finally, we end with Section \ref{add.est} in which we prove some additional estimates for $\rd_v r$ which will be used in the proof of the instability theorem.

\subsection{Ideas of the proof}\label{stab.ideas}

Our proof is strongly inspired by the recent work \cite{DL} on the stability of the Kerr Cauchy horizon. More precisely, it is based on the following ideas:
\begin{enumerate}
\item We view the problem as a stability problem and control only the quantities that remain close to the background Reissner--Nordstr\"om spacetime. We bound the difference of these quantities with their background Reissner--Nordstr\"om value and call these differences $\Psi$. These quantities include the scalar field, the metric components, as well as their appropriately weighted derivatives\footnote{Notice that we control for instance $\rd_v\Psi$, but $\rd_v$ is in fact a degenerate derivative near the Cauchy horizon.}. In particular, unlike in \cite{D2}, no estimates are derived for the Hawking mass, which according to \cite{D2} blows up for a large subclass of initial data.
\item In spherical symmetry, instead of working in $L^2$ based spaces as in \cite{DL}, one can directly control $\Psi$ and its derivatives in $L^\infty$ with weights in $u$ and $v$. The key is to prove estimates for $\Psi$ and its derivatives which are weighted in $|u|^s$ or $v^s$ with no loss in the polynomial power $s$ compared to initial data. To this end, we first observe that the background values of\footnote{See definitions in Sections \ref{sec.SS} and \ref{sec.geometry}.} $\Omg$, $\rd_u r$ and $\rd_v r$ in the $(u,v)$ coordinate system are exponential functions in $u+v$ and one can use basic calculus facts such as Lemmas \ref{lemma.1} and \ref{lemma.2} to ensure that the estimates for $\rd_u\Psi$ and $\rd_v\Psi$ do not lose in the polynomial weights. For $\Psi$ itself, we will bound a degenerate quantity $e^{-\kappa_- (u+v)}\Psi$ in order to to ensure that we have no loss in the polynomial weight.
\item Finally, we point out the source of smallness that we can exploit in this problem. By restricting ourselves to the future development of $\underline{C}'_1\cup C_{-\infty}$, we create a smallness parameter which is sufficient to deal with all terms that are nonlinear in $\Psi$. On the other hand, to deal with the linear terms, we face the problem that we need to prove estimates in a ``large'' region of spacetime and this largeness is reflected by the fact that the $r$ difference between two points in this region need not be small. To handle this, we divide the spacetime into finitely many regions each of which has small $r$ difference. This is a way of exploiting the linearity of these terms and can be viewed as a naive substitute of Gr\"onwall's inequality in this setting with two variables\footnote{A similar idea was exploited to derive the BV estimates for the scalar field in Section 13 of \cite{D2}.}.
\end{enumerate}

\subsection{Setup and bounds for the initial data}\label{sec.data.bound}

As mentioned before, we will work with all of the following coordinate systems: $(u,v)$, $(U,v)$ and $(u,V)$. The $(U,v)$ coordinates are chosen on the initial hypersurfaces by the gauge conditions \eqref{gauge.1} and \eqref{gauge.2}. We then define the $u$ and $V$ coordinates by \eqref{U.def} and \eqref{V.def} respectively.

For any fixed choice of null coordinate system, we write the metric in the form given by \eqref{SS.metric.1} and \eqref{SS.metric.2}. We will label the metric components\footnote{Notice that the values of $r$ are independent of the choice of the null coordinates.} by $(\Omg^2_{\mathcal H}, r)$, $(\Omg^2, r)$ and $(\Omg^2_{\mathcal C\mathcal H}, r)$ in the $(U,v)$, $(u,v)$ and $(u,V)$ coordinate systems respectively.

In the proof of the main theorem, we will estimate the scalar field as well as the difference of the metric components with the corresponding Reissner--Nordstr\"om spacetime. We find it convenient to introduce that quantity $\Psi$ which we will estimate in the proof of Theorem \ref{main.theorem.C0.stability}. In the coordinate system $(U,v)$, define
\begin{equation}\label{Psi.def}
\Psi=\begin{bmatrix}
\phi\\r-r_{RN}\\
\log\f{\Omg_\mathcal H}{\Omg_{RN,\mathcal H}}.
\end{bmatrix}
\end{equation}
Notice that here we have taken the functions $r_{RN}$ and $\Omg_{RN,\mathcal H}$ to be the metric components of a Reissner--Nordstr\"om spacetime with parameters $M$ and ${\bf e}$ equal to that in Theorem \ref{main.theorem.C0.stability}. In the $(u,v)$ coordinate system, we also define
$$\Psi=\begin{bmatrix}
\phi\\r-r_{RN}\\
\log\f{\Omg}{\Omg_{RN}}
\end{bmatrix}$$
We make a similar definition in the $(u,V)$ coordinate system. 

The functions $r$ and $\phi$ are manifestly independent of the choice of the null coordinates. Notice that we also have $\log\f{\Omg_\mathcal H}{\Omg_{RN,\mathcal H}}=\log\f{\Omg}{\Omg_{RN}}=\log\f{\Omg_{\mathcal C\mathcal H}}{\Omg_{RN,\mathcal C\mathcal H}}$. We will therefore refer to $\Psi$ without mentioning explicitly which coordinate system we are considering.

To proceed, we further introduce the notation that
$$|\Psi|^2:=|\phi|^2+|r-r_{RN}|^2+ \left\vert \log\f{\Omg_\mathcal H}{\Omg_{RN,\mathcal H}} \right\vert^{2}.$$
We also define $|\rd_v\Psi|$, $|\rd_V\Psi|$, $|\rd_u\Psi|$ and $|\rd_U\Psi|$ in the obvious manner.

Before we end this subsection, we prove some bounds for $\Psi$ and its derivatives \emph{on the initial hypersurfaces}. This will be the starting point of the later subsections in which we bound $\Psi$ in the spacetime solution. 
\begin{proposition}\label{initial.bound}
Under the assumptions of Theorem \ref{main.theorem.C0.stability}, there exists constant $D=D(M,{\bf e},s,E)>0$ such that on $\{U=0\}$, we have
\begin{equation}\label{initial.horizon}
|\Psi|(0,v)+|\rd_v\Psi|(0,v)\leq Dv^{-s};
\end{equation}
and on $\{v=1\}$, we have
\begin{equation}\label{initial.1}
|\rd_U\Psi|(U,0)\leq D.
\end{equation}
Moreover, by choosing $D=D(M,{\bf e},s,E)>0$ larger if necessary, we have the additional estimate on $\{U=0\}$,
\begin{equation}\label{initial.horizon.2}
|\rd_U\Psi|(0,v)\leq De^{2\kappa_+ v}v^{-s};
\end{equation}
and on $\{v=1\}$, we also have
\begin{equation}\label{initial.1.2}
|\Psi|(U,1)+|\rd_v\Psi|(U,1)\leq D.
\end{equation}
\end{proposition}
\begin{proof}
On $\{U=0\}$, we recall the gauge condition \eqref{gauge.1}:
$$\Omg_{\mathcal H}^2=\f{4e^{-2\kappa_+ r_+}(r_+-r_-)^{1+\f{\kappa_+}{\kappa_-}}}{r_+^2}e^{2\kappa_+ v}.$$
Hence, by definition, for $\Omg_{\mathcal H}$ in the $(U,v)$ coordinate system, we have 
\begin{equation}\label{Omg.initial.horizon}
\left(\log\Omg_{\mathcal H}-\log\Omg_{RN,\mathcal H}\right)(0,v)=0.
\end{equation}
Since $\rd_v\phi$ is given, $r$ can be computed by \eqref{eqn.Ray}, namely
$$\rd_v(e^{-2\kappa_+ v}\rd_v r)=-e^{-2\kappa_+ v}r(\rd_v\phi)^2.$$
Integrating this equation from $v=1$ to $v=\infty$, we obtain
$$\lim_{v\to \infty} (e^{-2\kappa_+ v}\rd_v r(0,v) + \int_1^v e^{-2\kappa_+ v'}r(\rd_v\phi)^2(0,v')\, \ud v') = e^{-2\kappa_+}\rd_v r(0,1).$$
We claim that 
$$\lim_{v\to \infty} \int_1^v e^{-2\kappa_+ v'}r(\rd_v\phi)^2(0,v')\, \ud v' = e^{-2\kappa_+}\rd_v r(0,1).$$
If not, since we also know that $\rd_v r(0,v)\geq 0$ (by Lemma~\ref{no.trapped}), it then follows that there exists $c>0$ so that for $v$ sufficiently large, $\rd_v r(0,v) \geq ce^{2\kappa_+ v}$. This then contradicts the fact that $r\to M+\sqrt{M^2-{\bf e}^2}$ as $v\to \infty$. Now the claim also implies that $\lim_{v\to \infty} e^{-2\kappa_+ v}\rd_v r(0,v) = 0$. Therefore, integrating the $\rd_v(e^{-2\kappa_+ v}\rd_v r)$ equation above, we obtain
\begin{equation}\label{r.initial.horizon}
|\rd_v r|(0,v)\leq C_{M,{\bf e},s} {\bf e}^2 v^{-2s},\quad |r-r_{RN}|=|r-(M+\sqrt{M^2-{\bf e}^2})|(0,v)\leq C_{M,{\bf e},s} {\bf e}^2 v^{-2s+1}.
\end{equation}
We then note that by the assumptions of Theorem \ref{main.theorem.C0.stability}, $\phi$ and $\rd_v\phi$ obey the desired estimate. This fact, together with \eqref{Omg.initial.horizon} and \eqref{r.initial.horizon}, prove \eqref{initial.horizon}.

We turn to the proof of \eqref{initial.1}.
On $\{v=1\}$, we set the gauge by the condition 
$\rd_U r=-1.$
Combined with the bound \eqref{RN.H.bound} for $\rd_{U} r_{RN}$, \eqref{initial.1} follows for $\rd_{U} (r - r_{N})$.
Moreover, by \eqref{eqn.Ray}, we have
$$\rd_U \left( \f{1}{\Omg^2_{\mathcal H}} \right)=\f{r(\rd_U\phi)^2}{\Omg_{\mathcal H}^2},$$
which implies that for $U\leq U_s\leq 1$,
$$|\rd_U\log\Omg_{\mathcal H}|(U,1)\leq C_{M,{\bf e}} (\rd_{U} \phi)^2(U, 1).$$
These observations, together with the assumption of Theorem \ref{main.theorem.C0.stability} on $\rd_U\phi$, imply \eqref{initial.1}.

Once we have proven \eqref{initial.horizon} and \eqref{initial.1}, it is easy to see that \eqref{initial.horizon.2} and \eqref{initial.1.2} follow directly from integration and the use of the equations \eqref{WW.SS}.
\end{proof}

\subsection{Strategy of the proof}\label{sec.strategy}

From this point onwards, we work under the assumptions of Theorem \ref{main.theorem.C0.stability}. In order to prove Theorem \ref{main.theorem.C0.stability}, we will estimate $\Psi$ (as defined in \eqref{Psi.def}) and its derivatives. Since we are allowed to restrict to the future development of a subset $\underline{C}_1'$ of $\underline{C}_1$, we only need to prove the bounds in $0\leq U\leq U_s$ for some $U_s$ sufficiently small (or equivalently $-\infty<u\leq u_s$ for some $u_s$ sufficiently negative). These estimates will be proven in the following two steps:
\begin{enumerate}
\item First, we prove estimates in a neighborhood of the event horizon, $\{(u,v):v\geq 1,\, u+v\leq A_1\}$ for $A_1$ sufficiently negative. This will be called the red-shift region $\mathcal R$. The choice of $A_1$ depends on the size of the initial data (in particular on $D$ in Proposition \ref{initial.bound}).
\item Second, after $A_1$ is fixed, we further divide the remaining region. Define $\mathcal B_i:=\{(u,v): A_i\leq u+v< A_{i+1}, u\leq u_i\}$ for $i=1,\ldots, n$, where $A_1<A_2<\dots<A_n<A_{n+1}:=\infty$, $u_1>u_2>\dots >u_n$ (all of these parameters will be chosen in the course of the proof). The choice of $A_i$ depends only on $A_1$, the parameters of the background Reissner--Nordstr\"om spacetime $M$, ${\bf e}$, and also $s$. In particular it does not depend explicitly on the size of the data. We then show that as long as $u_i$ is sufficient negative and $u_i\leq u_{i-1}$ (in a way that can depend on the size of the data), one can obtain good estimates for the solution in each of these regions. Finally, choosing $u_s\leq u_n$, we conclude the proof of the main theorem.

We will call $\mathcal B:=(\cup_{i=1}^n \mathcal B_i)\cap \{(u,v): u\leq u_s\}$ the blue-shift region.  
\end{enumerate}

A depiction of the partition of the spacetime and the various parameters involved is given in Figure \ref{fig:setup}.
\begin{figure}[htbp]
\begin{center}
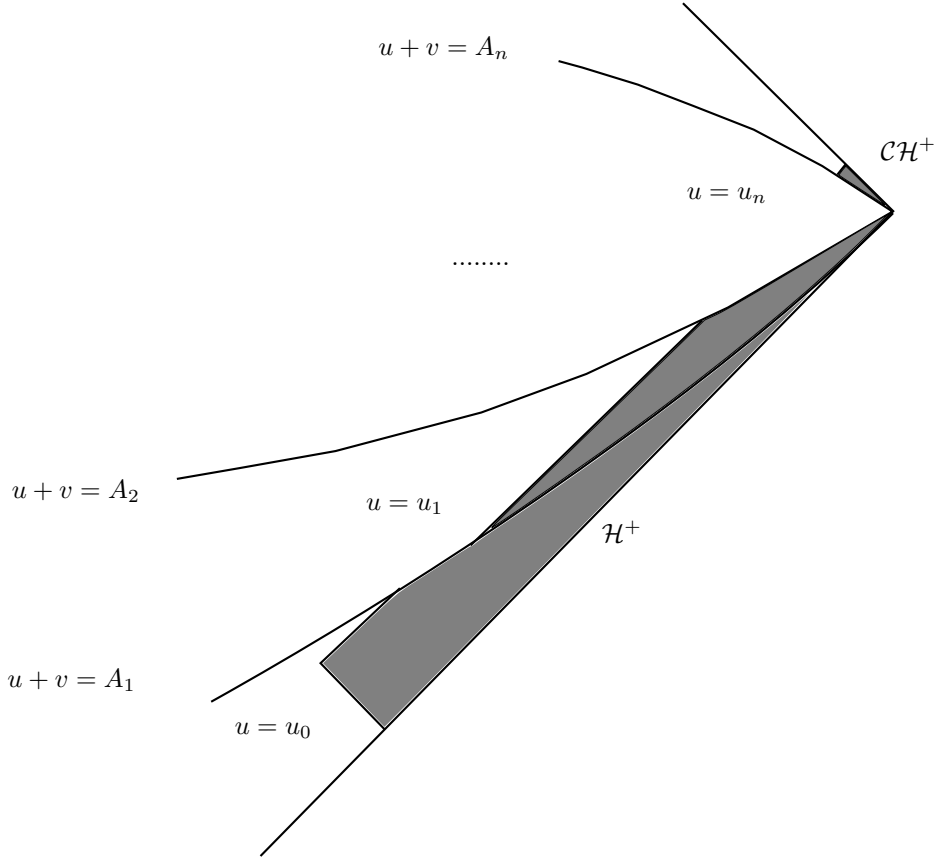 
\caption{Setup for the proof of main theorem} \label{fig:setup}
\end{center}
\end{figure}

We will carry out the estimates in the two steps above in the next two subsections. We then combine these bounds and conclude the proof of the main theorem in Section \ref{sec.C0}.

\subsection{The red-shift region}\label{sec.RS}

In this subsection, we work with the $(U,v)$ coordinate system since it is regular in the red-shift region. Define the red-shift region as follows:
$$\mathcal R:=\{(U,v): U e^{2 \kappa_+ v}\leq \de, \,v\geq 1,\, U\geq 0\},$$
where $\de>0$ is a small constant to be chosen later. Moreover, as we will show later, the choice of $\de$ depends only on $M$, ${\bf e}$, $s$ and $D$. Notice that this definition of $\mathcal R$ is the same as that in Section \ref{sec.strategy} if $A_1:=\f{1}{2\kappa_+}\log (2\kappa_+\de)$.

We proceed by writing down a schematic wave equation $\Psi$ under certain bootstrap assumptions. Assume that the following bound holds throughout $\mathcal R$:
\begin{equation}\label{BA.C0}
|\Psi|\leq 4D.
\end{equation}
We also need the following bootstrap assumptions:
\begin{equation}\label{BA.Omg}
|\log\Omg_{\mathcal H}-\log \Omg_{RN,\mathcal H}|\leq \f 1{100},
\end{equation}
and
\begin{equation}\label{BA.RS}
|\rd_U\Psi|(u,v)\leq \Delta e^{2\kappa_+ v}v^{-s},
\end{equation}
where $\Delta$ is a large constant to be chosen later.

Subtracting the equations for the Reissner--Nordstr\"om solution from that for $(\mathcal M,g,\phi,F)$ and using \eqref{BA.C0}, we obtain
\begin{equation}\label{RS.eqn}
|\rd_U\rd_v\Psi|\leq C_{M,{\bf e},s,D}\left( e^{2\kappa_+ v}(|\rd_v\Psi|+|\Psi|)+ |U|e^{2\kappa_+ v} |\rd_U\Psi|+|\rd_U\Psi\rd_v\Psi|\right),
\end{equation}
where $C_{M,{\bf e},s,D}$ is a constant depending only on $M$, ${\bf e}$, $s$ and $D$. Notice that in the above we have used the estimates \eqref{RN.H.bound} on Reissner--Nordstr\"om spacetime and also the bound
\begin{equation}\label{Omg.diff.bound}
\f{\Omg^2_{\mathcal H}-\Omg^2_{RN,\mathcal H}}{\Omg^2_{RN,\mathcal H}}\leq C\left(\log\Omg_{\mathcal H}-\log\Omg_{RN,\mathcal H}\right)
\end{equation}
which holds under the bootstrap assumption \eqref{BA.Omg}.

Our goal in the remainder of this subsection is to use the equation \eqref{RS.eqn} to obtain estimates for $\Psi$, $\rd_v\Psi$ and $\rd_u\Psi$. In particular, we will recover the bootstrap assumptions \eqref{BA.C0}, \eqref{BA.Omg} and \eqref{BA.RS}. We can immediately observe that in fact we can obtain a bound for $\Psi$ using \eqref{BA.RS}. In particular, since $v\geq 1$, the estimate in the following proposition improves the bootstrap assumptions \eqref{BA.C0} and \eqref{BA.Omg}:
\begin{proposition}\label{RS.P}
Under the bootstrap assumptions \eqref{BA.C0}, \eqref{BA.Omg} and \eqref{BA.RS}, for $\de=\de(\Delta,D)>0$ sufficiently small, the following estimates hold for $(U,v)\in \mathcal R$:
$$|\Psi|(U,v)\leq 2D v^{-s}$$
and 
$$|\log\Omg_{\mathcal H}-\log \Omg_{RN,\mathcal H}|(U,v)\leq \f 1{200}.$$
\end{proposition}
\begin{proof}
We apply the bound for the initial data in Proposition \ref{initial.bound}, the bootstrap assumption \eqref{BA.RS} and the fact that $Ue^{2\kappa_+ v}\leq \de$ to get
\begin{equation*}
\begin{split}
|\Psi|(U,v)\leq &|\Psi|(0,v)+\int_0^U |\rd_U\Psi|(U',v) dU' \\
\leq &Dv^{-s}+\int_0^U \Delta e^{2\kappa_+ v} v^{-s}\,dU' \leq (D+\Delta\de) v^{-s}\leq 2Dv^{-s},
\end{split}
\end{equation*}
after choosing $\de$ to be sufficiently small. To show the second bound stated in the proposition, notice that by \eqref{Omg.initial.horizon}, $\log\Omg_{\mathcal H}-\log \Omg_{RN,\mathcal H}$ is vanishing for $U=0$. We can therefore repeat the argument above without the $Dv^{-s}$ term, i.e., we have
$$|\log\Omg_{\mathcal H}-\log \Omg_{RN,\mathcal H}|(U,v)\leq \int_0^U \Delta e^{2\kappa_+ v} v^{-s}\,dU' \leq \Delta\de v^{-s}\leq \f{1}{200}$$
for $\de$ sufficiently small.
\end{proof}
Using this we get a bound for $\rd_v\Psi$:
\begin{proposition}\label{RS.rvP}
Under the bootstrap assumptions \eqref{BA.C0}, \eqref{BA.Omg} and \eqref{BA.RS}, for $\de=\de(M,{\bf e},\Delta,D)>0$ sufficiently small, the following estimate holds for $(U,v)\in \mathcal R$:
$$|\rd_v\Psi|(U,v)\leq 2D v^{-s}.$$
\end{proposition}
\begin{proof}
Integrating \eqref{RS.eqn} and using Proposition \ref{initial.bound}, the bootstrap assumption \eqref{BA.RS} and Proposition \ref{RS.P}, we obtain
\begin{equation*}
\begin{split}
&|\rd_v\Psi|(U,v)\leq |\rd_v\Psi|(0,v)+\int_0^U |\rd_U\rd_v\Psi|(U',v)dU'\\
\leq &D v^{-s}+C_{M,{\bf e},s,D}\int_0^U \left(e^{2\kappa_+ v}(|\rd_v\Psi|+|\Psi|)+ |U'|e^{2\kappa_+ v} |\rd_U\Psi|+|\rd_U\Psi||\rd_v\Psi|\right)(U',v)\, dU'\\
\leq & (D+2 C_{M,{\bf e},s,D} U e^{2\kappa_+ v}+C_{M,{\bf e},s,D}\Delta U^2 e^{4\kappa_+ v}) v^{-s}\\
&+C_{M,{\bf e},s,D}\int_0^U \left(e^{2\kappa_+ v}+|\rd_U\Psi|(U',v)\right)|\rd_v\Psi|(U',v)\, dU'.
\end{split}
\end{equation*}
Recalling that in the red-shift region $U e^{2\kappa_+ v}\leq \de$, we can thus choose $\de$ to be sufficiently small such that the first term is $\leq \f{3D}{2} v^{-s}$. We then apply Gr\"onwall's inequality to get
$$|\rd_v\Psi|(U,v)\leq \f{3D}2 v^{-s}e^{C_{M,{\bf e},s,D}Ue^{2\kappa_+ v}+C_{M,{\bf e},s,D}\int_0^U|\rd_U\Psi|(U',v)\, dU'}\leq \f{3D}2 v^{-s}e^{C_{M,{\bf e},s,D}\de+C_{M,{\bf e},s,D}\Delta\de v^{-s}}.$$
Choosing $\de$ to be sufficiently small depending on $M$, ${\bf e}$, $D$ and $\Delta$, we thus have obtained the desired conclusion.
\end{proof}
Our goal is then to use the above bounds to control $\rd_U\Psi$ and to improve the bootstrap assumption \eqref{BA.RS}. Before we proceed to the estimates for $\rd_U\Psi$, we first prove a simple lemma:
\begin{lemma}\label{lemma.1}
For every real numbers $s\geq 0$ and $v\geq 1$, we have
$$\int_1^v e^{2\kappa_+ v'} (v')^{-s} \,dv'\leq C_{\kappa_+,s}e^{2\kappa_+ v} v^{-s},$$
for some $C_{\kappa_+,s}>0$ depending only on $\kappa_+$ and $s$.
\end{lemma}
\begin{proof}
Integrating by parts twice, we obtain
\begin{equation*}
\begin{split}
\int_1^v e^{2\kappa_+ v'} (v')^{-s} \,dv'\leq &\f{1}{2\kappa_+}e^{2\kappa_+ v} v^{-s}+\f{s}{2\kappa_+}\int_1^v e^{2\kappa_+ v'} (v')^{-s-1} \,dv'\\
\leq &\f{1}{2\kappa_+}e^{2\kappa_+ v} v^{-s}+\f{s}{4\kappa_+^2}e^{2\kappa_+ v} v^{-s-1}+\f{s(s+1)}{4\kappa_+^2}\int_1^v e^{2\kappa_+ v'} (v')^{-s-2} \,dv'.
\end{split}
\end{equation*}
The last term can be bounded by
$$\leq \f{s(s+1)}{4\kappa_+^2}e^{2\kappa_+ v} v^{-s}\int_1^v (v')^{-2} \,dv'\leq \f{s(s+1)}{4\kappa_+^2}e^{2\kappa_+ v} v^{-s}.$$
Combining all the bounds above gives the desired conclusion.
\end{proof}
Finally, we use the equation \eqref{RS.eqn} and the above lemma to estimate $\rd_U\Psi$ and close the bootstrap assumption \eqref{BA.RS}.
\begin{proposition}\label{RS.ruP}
Under the bootstrap assumptions \eqref{BA.C0}, \eqref{BA.Omg} and \eqref{BA.RS}, for $\de=\de(M,{\bf e},\Delta,D)>0$ sufficiently small, the following estimate holds for $(U,v)\in \mathcal R$:
$$|\rd_U\Psi|(U,v)\leq C_{M,{\bf e},s,D} e^{2\kappa_+ v} v^{-s},$$
where $C_{M,{\bf e},s,D}$ is a constant depending only on $M$, ${\bf e}$, $s$ and $D$.
\end{proposition}
\begin{proof}
Integrating \eqref{RS.eqn} in the $v$-direction, applying Proposition \ref{initial.bound} and using Propositions \ref{RS.P} and \ref{RS.rvP} together with Lemma \ref{lemma.1}, we get
\begin{equation*}
\begin{split}
|\rd_U\Psi|(U,v)\leq & D+C_{M,{\bf e},s,D}\int_1^v \left(e^{2\kappa_+ v'}(|\rd_v\Psi|+|\Psi|)+ |U|e^{2\kappa_+ v} |\rd_U\Psi|+|\rd_U\Psi||\rd_v\Psi|\right)\, dv'\\
\leq &D+C_{M,{\bf e},s,D} \int_1^v e^{2\kappa_+ v'} (v')^{-s}\, dv'+C_{M,{\bf e},s,D}\int_1^v (|U|e^{2\kappa_+ v}+|\rd_v\Psi|)|\rd_U\Psi|\, dv'\\
\leq &D+C_{M,{\bf e},s,D} e^{2\kappa_+ v} v^{-s} +C_{M,{\bf e},s,D}\int_1^v (|U|e^{2\kappa_+ v}+|\rd_v\Psi|)|\rd_U\Psi|\, dv'.
\end{split}
\end{equation*}
Finally, we apply Gr\"onwall's inequality to get the desired conclusion.
\end{proof}
Choosing $\Delta>C_{M,{\bf e},s,D}$, we have thus improved the bootstrap assumption \eqref{BA.RS} and completed the proof of the following result:
\begin{theorem}\label{RS.prop}
There exist constants $\de=\de(M,{\bf e},s,D)>0$ and $C_{M,{\bf e},s,D}>0$ such that the following estimate holds for $(U,v)\in \mathcal R$:
$$|\Psi|(U,v)+|\rd_v\Psi|(U,v)+e^{-2\kappa_+ v}|\rd_U\Psi|(U,v)\leq C_{M,{\bf e},s,D}  v^{-s}.$$
\end{theorem}
\begin{proof}
By \eqref{initial.horizon}, \eqref{initial.1}, \eqref{initial.horizon.2}, \eqref{initial.1.2} and \eqref{Omg.initial.horizon}, the bootstrap assumptions \eqref{BA.C0}, \eqref{BA.Omg} and \eqref{BA.RS} hold on the initial hypersurfaces. We then choose $\Delta>C_{M,{\bf e},s,D}$, where $C_{M,{\bf e},s,D}$ is the constant appearing in the statement of Proposition \ref{RS.ruP}. After choosing\footnote{Notice that since the choice of $\Delta$ depends only on $M$, ${\bf e}$, $s$ and $D$, we can choose $\de$ to depend only on $M$, ${\bf e}$, $s$ and $D$.} $\de=\de(M,{\bf e},s,D)>0$ sufficiently small,  Propositions \ref{RS.P}, \ref{RS.rvP} and \ref{RS.ruP} then imply that the bootstrap assumptions \eqref{BA.C0}, \eqref{BA.Omg} and \eqref{BA.RS} can in fact be improved and hold with a better constant. A standard continuity argument then shows that the desired estimates hold.
\end{proof}

We now fix a $\de$ so that Theorem~\ref{RS.prop} holds.

\subsection{The blue-shift region}\label{sec.BS}
We now turn to the blue-shift region. In this region, it is convenient to use the $(u,v)$ coordinate system. Recall that $U$ and $u$ are related by \eqref{U.def}. The blue-shift region is given in the $(u,v)$ coordinates by
$$\mathcal B:=\{(u,v) :u+v\geq \f{1}{2\kappa_+}\log (2\kappa_+\de)=:A_1,\, u\leq u_s,\,v\geq 1\}. $$
As mentioned before, we will define a partition $\mathcal B_i$ so that $\mathcal B=(\cup_{i=1}^n \mathcal B_i)\cap \{(u,v): u\leq u_s\}$. Given $\de_B>0$ (which will be chosen below; see Theorem~\ref{BS.main}) and $A_1=\f{1}{2\kappa_+}\log (2\kappa_+\de) \in (-\infty,\infty)$ (which is fixed by Theorem~\ref{RS.prop}), define a sequence $A_1< A_2<...< A_n$ such that 
$$\min\{e^{-\kappa_- A_i}-e^{-\kappa_- A_{i+1}},e^{-2\kappa_- A_i}-e^{-2\kappa_- A_{i+1}}\}=\de_B.$$
We require $n$ to be such that 
$$e^{-2\kappa_- A_n}\leq \de_B.$$
Moreover, we will use the convention that 
$$A_{n+1}=\infty.$$
Define now $\mathcal B_i$ by
$$\mathcal B_i:=\{(u,v):A_i\leq u+v\leq A_{i+1},\,u\leq u_i,\, v\geq 1\},$$
where $u_i$ is a decreasing sequence to be chosen below (see Theorem~\ref{BS.main}).

We now state the following main estimates of this subsection:
\begin{theorem}\label{BS.main}
Consider the following statements
\begin{equation}\label{BS.main.data}
\sup_{u+v=A_i,\, u\leq u_i,\, v\geq 1}|u|^s (|\rd_u\Psi|(u,v)+|\rd_v\Psi|(u,v)+|\Psi|(u,v))\leq B_i,
\end{equation}
and
\begin{equation}\label{BS.main.est}
\begin{split}
\sup_{(u,v)\in\mathcal B_i} \bb( |u|^s(|\rd_u\Psi|(u,v)+|e^{-\kappa_- (v+u)}\Psi|(u,v)) &\\
+(v-A_i)^s|\rd_v\Psi|(u,v)+|u|^{s-1}|\Psi|(u,v) \bb) &\leq C_{M, {\bf e},s,A_1} B_i
\end{split}
\end{equation}
on the hypersurface $\{u+v=A_i\}$ and in the region $\mathcal B_i$ respectively.

Then, given $A_1\in (-\infty,\infty)$ and $B_i\geq 0$, there exist $\de_B=\de_B(M,{\bf e},s,A_1)>0$ sufficiently small\footnote{We emphasize that the choice of $\de_B$ is independent of $B_i$.} and $u_i=u_i(M,{\bf e},s,B_i)$ sufficiently negative such that \eqref{BS.main.data} implies \eqref{BS.main.est} with a constant $C_{M,{\bf e},s,A_1}$ depending only on $M$, ${\bf e}$, $s$ and $A_1$.
\end{theorem}
In order to prove the estimates, we make the following bootstrap assumptions:
\begin{equation}\label{BS.C0}
\sup_{(u,v)\in \mathcal B_i}|\Psi|(u,v)\leq \f{1}{100}
\end{equation}
and
\begin{equation}\label{BA.1}
\sup_{(u,v)\in \mathcal B_i}|u|^s|\rd_u\Psi|(u,v)\leq 2B_i.
\end{equation}
The bound \eqref{BA.1} obviously holds initially by \eqref{BS.main.data}. Notice also that by \eqref{BS.main.data}, \eqref{BS.C0} holds initially on $\{(u,v):u+v=A_i,\,u\leq u_i,\,v\geq 1\}$ if $u_i$ is chosen to be sufficiently negative.

The equations for $\Psi$ together with \eqref{BS.C0} imply that
\begin{equation}\label{BS.eqn}
|\rd_u\rd_v\Psi|\leq C_{M,{\bf e},A_1} \bb( e^{-2\kappa_-(v+u)} (|\Psi|+|\rd_v\Psi|+|\rd_u\Psi|)+|\rd_u\Psi\rd_v\Psi| \bb).
\end{equation}
Here, we have used an analogue of \eqref{Omg.diff.bound} which holds in this setting due to \eqref{BS.C0}.

Our goal is to show that under the bootstrap assumptions \eqref{BS.C0} and \eqref{BA.1}, we can use the equation \eqref{BS.eqn} to control $\Psi$ and its derivatives. Before we proceed to the estimates, we need a calculus lemma, which can be viewed as an analogue of Lemma \ref{lemma.1} in the blue-shift region:
\begin{lemma}\label{lemma.2}
For $1\leq \alp\leq 2$ and $u \leq -1$, we have
$$\int_{-v+A_i}^u e^{-\alp\kappa_-(v+u')}|u'|^{-s}\, du'\leq Ce^{-\alpha\kappa_- A_i}(v-A_i)^{-s},$$
where the constant $C$ depends on $M$, ${\bf e}$ and $s$.
\end{lemma}
\begin{proof}
Integrating by parts, we obtain\footnote{Notice that the boundary term at $u'=u$ has a good sign.}
\begin{equation*}
\begin{split}
&\int_{-v+A_i}^u e^{-\alp\kappa_-(v+u')}|u'|^{-s}\, du' \\
\leq & \f 1{\alp\kappa_-} e^{-\alp\kappa_-A_i}(v-A_i)^{-s}+\f{s}{\alp\kappa_-} \int_{-v+A_i}^u e^{-\alp\kappa_-(v+u')}|u'|^{-s-1}\, du'.
\end{split}
\end{equation*}
The second term can be treated as an error term as follows\footnote{In the case $u\leq -v+A_i+\f s{\alp\kappa_-} \log (v-A_i)$, we can ignore the first integral below and obtain a better bound.}:
\begin{equation*}
\begin{split}
&\int_{-v+A_i}^u e^{-\alp\kappa_-(v+u')}|u'|^{-s-1}\, du' \\
= & \int_{-v+A_i+\f s{\alp\kappa_-} \log (v-A_i)}^u e^{-\alp\kappa_-(v+u')}|u'|^{-s-1}\, du'+\int_{-v+A_i}^{-v+A_i+\f s{\alp\kappa_-} \log (v-A_i)} e^{-\alp\kappa_-(v+u')}|u'|^{-s-1}\, du'\\
\leq & \f{1}{s} e^{-\alp\kappa_-A_i}(v-A_i)^{-s}|u|^{-s}+\f s{\alp\kappa_-} e^{-\alp\kappa_- A_i} (v-A_i)^{-s-1} \log (v-A_i) .
\end{split}
\end{equation*}
Combining the bounds above and using $u+v\geq A_i$ give the desired conclusion.
\end{proof}

We now turn to the estimates for $\Psi$ and its derivatives. Using Lemma \ref{lemma.2}, the bootstrap assumption \eqref{BA.1} immediately implies the following bound on $\Psi$:
\begin{proposition}\label{BS.P}
Under the bootstrap assumptions \eqref{BS.C0} and \eqref{BA.1}, there exist $\de_B=\de_B(M,{\bf e},s,A_1)>0$ sufficiently small and $u_i=u_i(M,{\bf e},s,B_i)$ sufficiently negative such that the following estimate holds for $(u,v)\in \mathcal B_i$:
$$|e^{-\kappa_-(v+u)}\Psi|(u,v)\leq C_{M,{\bf e},s} B_i e^{-\kappa_- A_i} |u|^{-s}.$$
\end{proposition}
\begin{proof}
Using the bootstrap assumption \eqref{BA.1} and Lemma \ref{lemma.2}, we first derive the following bound on a weighted integral of $|\rd_u\Psi|$:
\begin{equation}\label{BS.P.1}
\begin{split}
\int_{-v+A_i}^u e^{-\kappa_-(v+u')}|\rd_u\Psi|(u',v)\, du'\leq& 2 B_i\int_{-v+A_i}^u e^{-\kappa_-(v+u')}|u'|^{-s}\, du'\\
\leq & C_{M,{\bf e},s}B_i e^{-\kappa_- A_i} (v-A_i)^{-s}\leq C_{M,{\bf e},s}B_i e^{-\kappa_- A_i}|u|^{-s}.
\end{split}
\end{equation}
Now, to estimate $e^{-\kappa_-(v+u)}|\Psi|$, notice that
$$\f 12\rd_u(e^{-\kappa_-(v+u)}\Psi)^2=e^{-2\kappa_-(v+u)}\Psi\rd_u\Psi-\kappa_-(e^{-\kappa_- (v+u)}\Psi)^2.$$
The second term has a good sign and therefore by integrating along a constant $v$ curve, we obtain
$$(e^{-\kappa_-(v+u)}\Psi)^2 (u,v)\leq e^{-2\kappa_- A_i}B_i^2(v-A_i)^{-2s}+C\int_{-v+A_i}^u e^{-2\kappa_-(v+u')}|\Psi\rd_u\Psi | (u',v) du'.$$
This implies, after dividing by\footnote{Note that we only apply this division in the case $\sup_{u'\in [v+A_i,u]}e^{-\kappa_-(v+u')}|\Psi| (u',v)\geq e^{-\kappa_- A_i}B_i(v-A_i)^{-s}$ since the estimate is trivial otherwise.} $\sup_{u'\in [v+A_i,u]}e^{-\kappa_-(v+u')}|\Psi| (u',v)$ on both sides, that
\begin{equation}\label{BS.P.2}
e^{-\kappa_-(v+u)}|\Psi| (u,v)\leq Ce^{-\kappa_- A_i}B_i|u|^{-s}+C\int_{-v+A_i}^u e^{-\kappa_-(v+u')}|\rd_u\Psi(u',v)| du'.
\end{equation}
The conclusion follows from combining \eqref{BS.P.1} and \eqref{BS.P.2}.
\end{proof}
We then turn to the estimate for $\rd_v\Psi$:
\begin{proposition}\label{BS.rvP}
Under the bootstrap assumptions \eqref{BS.C0} and \eqref{BA.1}, there exist $\de_B=\de_B(M,{\bf e},s,A_1)>0$ sufficiently small and $u_i=u_i(M,{\bf e},s,B_i)$ sufficiently negative such that the following estimate holds for $(u,v)\in\mathcal B_i$:
$$|\rd_v\Psi|(u,v)\leq C_{M,{\bf e},s,A_1}B_i(v-A_i)^{-s}.$$
\end{proposition}
\begin{proof}
Integrating the equation \eqref{BS.eqn} and using the initial data bound \eqref{BS.main.data}, we have
\begin{equation}\label{BS.rvP.1}
\begin{split}
&|\rd_v\Psi|(u,v)\leq |\rd_v\Psi|(-v+A_i,v)+\int_{-v+A_i}^u |\rd_u\rd_v\Psi| du'\\
\leq &B_i(v-A_i)^{-s}+C_{M,{\bf e},A_1}\int_{-v+A_i}^u \left(e^{-2\kappa_-(v+u')}(|\rd_u\Psi|+|\rd_v\Psi|+|\Psi|)+|\rd_u\Psi||\rd_v\Psi|\right)(u',v) du'.
\end{split}
\end{equation}
Combining the bootstrap assumption \eqref{BA.1} and Lemma \ref{lemma.2} with $\alp=2$ , we get
$$\int_{-v+A_i}^u e^{-2\kappa_-(v+u')}|\rd_u\Psi|(u',v) du' \leq C_{M,{\bf e},s} B_i e^{-2\kappa_- A_i} (v-A_i)^{-s}.$$
Similarly, using the bound in Proposition \ref{BS.P} together with Lemma \ref{lemma.2} with $\alp=1$, we obtain
$$\int_{-v+A_i}^u e^{-2\kappa_-(v+u')}|\Psi|(u',v) du' \leq C_{M,{\bf e},s} B_i e^{-2\kappa_- A_i} (v-A_i)^{-s}.$$
Therefore, we have
\begin{equation}\label{BS.rvP.2}
|\rd_v\Psi|(u,v)\leq C_{M,{\bf e},s,A_1}B_i(v-A_i)^{-s}+C_{M,{\bf e},A_1}\int_{-v+A_i}^u \big(e^{-2\kappa_-(v+u')}|\rd_v\Psi|+|\rd_u\Psi||\rd_v\Psi|\big)(u',v) du'.
\end{equation}
To conclude, notice that
$$\int_{-v+A_i}^u e^{-2\kappa_-(v+u')} du'\leq C_{M,{\bf e}}(e^{-2\kappa_i A_i}-e^{-2\kappa_i A_{i+1}})\leq C_{M,{\bf e}}\delta_B$$
and
$$\int_{-v+A_i}^u |\rd_u\Psi|(u',v) du'\leq C_s B_i|u|^{-s+1}\leq C_s B_i|u_i|^{-s+1}.$$
Therefore, applying Gr\"onwall's inequality to \eqref{BS.rvP.2}, we obtain
$$|\rd_v\Psi|(u,v)\leq C_{M,{\bf e},s,A_1}B_i(v-A_i)^{-s} e^{C_{M,{\bf e}}\delta_B+C_sB_i|u_i|^{-s+1}}\leq C_{M,{\bf e},s,A_1}B_i(v-A_i)^{-s},$$
where in the last inequality, we have chosen $\delta_B=\delta_B(M,{\bf e},s)$ to be sufficiently small and $u_i=u_i(M,{\bf e},s,B_i)$ to be sufficiently negative such that $e^{C_{M,{\bf e}}\delta_B+C_sB_i|u_i|^{-s+1}}\leq 2$.
\end{proof}
This allows us to estimate $\rd_u\Psi$ and recover the bootstrap assumption \eqref{BA.1}.
\begin{proposition}\label{BS.ruP}
Under the bootstrap assumptions \eqref{BS.C0} and \eqref{BA.1}, there exist $\de_B=\de_B(M,{\bf e},s,A_1)>0$ sufficiently small and $u_i=u_i(M,{\bf e},s,B_i)$ sufficiently negative such that the following estimate holds in $\mathcal B_i$:
$$|\rd_u\Psi|(u,v)\leq \f 32 B_i |u|^{-s}.$$
\end{proposition}
\begin{proof}
Using \eqref{BS.eqn} and the initial data bound \eqref{BS.main.data}, we have
\begin{equation}\label{BS.ruP.1}
\begin{split}
&|\rd_u\Psi|(u,v)\leq |\rd_u\Psi|(u,-u+A_i)+\int_{-u+A_i}^v |\rd_v\rd_u\Psi| dv'\\
\leq &B_i|u|^{-s}+C_{M,{\bf e},A_1}\int_{-u+A_i}^v \left(e^{-2\kappa_-(v'+u)}(|\rd_u\Psi|+|\rd_v\Psi|+|\Psi|)+|\rd_u\Psi||\rd_v\Psi|\right)(u,v') dv'.
\end{split}
\end{equation}
By Propositions \ref{BS.P} and \ref{BS.rvP} and $u+v\geq A_i$, we have
\begin{equation*}
\begin{split}
&\int_{-u+A_i}^v \left(e^{-2\kappa_-(v'+u)}(| \Psi |+|\rd_v\Psi|)\right)(u,v') dv'\\
\leq &C_{M,{\bf e},s,A_1}B_i|u|^{-s}\int_{-u+A_i}^v e^{-2\kappa_- (v'+u)} dv'\leq C_{M,{\bf e},s,A_1}B_i\de_B |u|^{-s}\leq \f{B_i}{5}|u|^{-s},
\end{split}
\end{equation*}
after choosing $\delta_B=\delta_B(M,{\bf e},s,A_1)$ to be sufficiently small.
Returning to \eqref{BS.ruP.1}, we get
\begin{equation}\label{BS.ruP.2}
\begin{split}
|\rd_u\Psi|(u,v)\leq &\f{6B_i}{5}|u|^{-s}
+C_{M,{\bf e},A_1}\int_{-u+A_i}^v \left(e^{-2\kappa_-(v'+u)}|\rd_u\Psi|+|\rd_u\Psi||\rd_v\Psi|\right)(u,v') dv'.
\end{split}
\end{equation}
In order to apply Gr\"onwall's inequality, notice that
$$\int_{-u+A_i}^v e^{-2\kappa_-(v'+u)} dv'\leq C_{M,{\bf e}}\de_B$$
and
$$\int_{-u+A_i}^v |\rd_v\Psi|(u,v') dv'\leq C_{M,{\bf e},s,A_1}B_i\int_{-u+A_i}^v (v'-A_i)^{-s} d v' \leq C_{M,{\bf e},s,A_1}B_i|u_i|^{-s+1}$$
where the latter bound is obtained using Proposition \ref{BS.rvP}. Substituting these bounds into \eqref{BS.ruP.2}, we then obtain
\begin{equation*}
\begin{split}
|\rd_u\Psi|(u,v)\leq &\f{6B_i}{5}|u|^{-s}\times e^{C_{M,{\bf e}}\de_B+C_{M,{\bf e},s,A_1}B_i|u_i|^{-s+1}}\leq \f{3B_i}{2}|u|^{-s},
\end{split}
\end{equation*}
after choosing $\delta_B=\delta_B(M,{\bf e},s,A_1)$ to be sufficiently small and $u_i=u_i(M,{\bf e},s,B_i)$ to be sufficiently negative.
\end{proof}
Now, notice that the bound for $\Psi$ derived in Proposition \ref{BS.P} is not sufficiently strong to improve the bootstrap assumption \eqref{BS.C0}. Nonetheless, we also have the following estimate:
\begin{proposition}\label{Psi.C0}
Under the bootstrap assumptions \eqref{BS.C0} and \eqref{BA.1}, there exist $\de_B=\de_B(M,{\bf e},s,A_1)>0$ sufficiently small and $u_i=u_i(M,{\bf e},s,B_i)$ sufficiently negative such that the following estimate holds in $\mathcal B_i$:
$$|\Psi|\leq C_sB_i|u|^{-s+1}.$$
\end{proposition}
\begin{proof}
Using the assumption on the data in Theorem \ref{BS.main} and the bootstrap assumption \eqref{BA.1}, we have
\begin{equation*}
\begin{split}
|\Psi|(u,v)\leq &|\Psi|(-v+A_i, v)+\int_{-v+A_i}^u|\rd_u\Psi|(u',v)\,du'\\
\leq &B_i(v-A_i)^{-s}+2B_i\int_{-v+A_i}^u\f{du'}{|u'|^s}\leq C_sB_i\left((v-A_i)^{-s}+|u|^{-s+1}\right).
\end{split}
\end{equation*}
Finally, we conclude by noting that $v-A_i\geq |u|$.
\end{proof}

According to Proposition \ref{Psi.C0}, we have thus improved the bootstrap assumption \eqref{BS.C0} after choosing $u_i$ to be sufficiently negative depending on $B_i$ and $s$. Moreover, we have also improved the bootstrap assumption \eqref{BA.1} in Proposition \ref{BS.ruP}. Together with Propositions \ref{BS.P} and \ref{BS.rvP}, we conclude the proof of Theorem \ref{BS.main}. Finally, by iterating the estimate in Theorem \ref{BS.main}, we obtain the following bounds in the region $\mathcal B$.
\begin{theorem}\label{BS.prop}
Given the estimates in the red-shift region $\mathcal R$ in Theorem \ref{RS.prop}, there exists $u_s=u_s(M,{\bf e},s,D)$ sufficiently negative such that the following estimates hold for $(u,v)\in\mathcal B$:
$$|\rd_u\Psi|(u,v)\leq C_{M,{\bf e},s,D}|u|^{-s},\quad |\Psi|(u,v)\leq C_{M,{\bf e},s,D}|u|^{-s+1},\quad |\rd_v\Psi|(u,v)\leq C_{M,{\bf e},s,D}v^{-s}.$$
\end{theorem}
\begin{proof}
To prove the estimates, we apply Theorem \ref{BS.main}. First, notice that according to Theorem \ref{BS.main}, $\de_B$ depends only on $M$, ${\bf e}$, $s$ and $A_1$. Therefore, given $A_1$ (which depends on $M$, ${\bf e}$, $s$ and $D$), we can choose $\de_B$ and partition the blue-shift region in $\mathcal B_i$. We now repeatedly apply Theorem \ref{BS.main} for each of the $\mathcal B_i$ to show that the following holds for every $i$, 
\begin{equation}\label{BS.final}
\begin{split}
\sup_{(u,v)\in\mathcal B_i} \bb( |u|^s(|\rd_u\Psi|(u,v)+|e^{-\kappa_- (v+u)}\Psi|(u,v)) &\\
+(v-A_i)^s|\rd_v\Psi|(u,v)+|u|^{s-1}|\Psi|(u,v) \bb)
\leq &\left(C_{M, {\bf e},s,A_1}^*\right)^i e^{2i\kappa_-A_n} C_{M,{\bf e},s,D}^{**},
\end{split}
\end{equation}
where $C_{M, {\bf e},s,A_1}^*$ is the constant in the conclusion of Theorem \ref{BS.main} and $C_{M,{\bf e},s,D}^{**}$ is a constant that exists by the conclusion of Theorem \ref{RS.prop}. From this we can easily obtain the conclusion. First, note that $(v-A_i)^{-s}\leq C_{M,{\bf e},s,D}v^{-s}$. Finally, we observe that the right hand side of \eqref{BS.final} can be bounded uniformly by a constant depending on $M$, ${\bf e}$, $s$ and $D$.
\end{proof}

\subsection{$C^0$ stability and conclusion of proof of Theorem \ref{main.theorem.C0.stability}}\label{sec.C0}

The estimates in the previous two subsections already imply that the spacetime remains regular in the region $u\leq u_s$, $v\geq 1$. In this subsection, we will further show that the metric and the scalar field can be extended in $C^0$ and moreover the solution approaches Reissner--Nordstr\"om in $C^0$ in the sense given by the last statement in Theorem \ref{main.theorem.C0.stability}. 

First, we have the following $C^0$ extendibility statement:
\begin{proposition}\label{C0.1}
In the $(u,V)$ coordinate system, one can attach the boundary $\CH:=\{V=1\}$ such that $r$, $\phi$ and $\log\Omg_{\mathcal C\mathcal H}$ extend continuously to $\CH$.
\end{proposition}
\begin{proof}
Since the metric components of Reissner--Nordstr\"om $r_{RN}$ and $\log\Omg_{RN,\mathcal C\mathcal H}$ are smooth up to $\{V=1\}$, it suffices to show that $\Psi$ is continuous up to $\{V=1\}$.

First, we note that by the conclusions in Theorems \ref{RS.prop} and \ref{BS.prop}, the following holds in the $(u,V)$ coordinate system if $u\leq u_s$, $V\geq \f 12$:
\begin{equation}\label{newcoord.bound}
|\rd_u\Psi|(u,V)\leq C_{M,{\bf e},s,D}|u|^{-s},\quad |\rd_V\Psi|(u,V)\leq C_{M,{\bf e},s,D}|1-V|^{-1}\log^{-s}\left(\f{1}{1-V}\right),
\end{equation}

Given a sequence $u_i\to u$ and $V_i\to 1$, it then suffices to show that $\Psi(u_i,V_i)$ is a Cauchy sequence. Fix $\ep>0$. Notice that there exists $N$ such that for $i,j\geq N$, we have
\begin{equation*}
\begin{split}
&|\Psi(u_i,V_i)-\Psi(u_j,V_j)|\\
\leq &\left|\int_{u_i}^{u_j}|\rd_u\Psi|(u',V_i) \,du'\right|+\left|\int_{V_i}^{V_j}|\rd_V\Psi|(u_j,V')\,d V'\right| \\
\leq &C_{M,{\bf e},s,D}\left(\left|\int_{u_i}^{u_j}|u'|^{-s} \,du'\right|+\left|\int_{V_i}^{V_j}|1-V'|^{-1}\log^{-s}\left(\f{1}{1-V'}\right)\,d V'\right| \right)\leq \ep.
\end{split}
\end{equation*}
This is because for $s>1$, the integrands in both integrals are integrable.
\end{proof}

We now have all the ingredients to conclude the proof of Theorem \ref{main.theorem.C0.stability}.
\begin{proof}[Proof of Theorem \ref{main.theorem.C0.stability}]
By Theorems \ref{RS.prop} and \ref{BS.prop} and standard local existence results, the solution remains regular in the region $\{(u,v):u\leq u_s,\,v\geq 1\}$ as long as $u_s$ is sufficiently negative. In particular, the spacetime indeed has a Penrose diagram given by Figure \ref{fig:maintheorem}. The statement on the continuity of the metric and the scalar field up to $\CH$ is given by Proposition \ref{C0.1}. Finally, the desired estimates are direct consequences of Theorems \ref{RS.prop} and \ref{BS.prop}.
\end{proof}

\subsection{Refined bounds for $\rd_v r$}\label{add.est}

In this final subsection, we record an easy estimate for $\rd_v r$. This estimate is not necessary to close the argument for the stability theorem, but will be useful in the next section.

\begin{proposition}\label{dvr.improved}
For $u_s$ sufficiently negative, we have the following bound to the future of $\{u+v=A_1,\,u\leq u_s\}$:
$$|\rd_v r|(u,v)\leq C_{M,{\bf e},s,D}(\Omg_{RN}^2+v^{-2s}).$$
\end{proposition}
\begin{proof}
We use the equation 
\begin{equation}\label{R.eqn}
\rd_v \left( \f{\rd_v r}{\Omg^2} \right) =-\f{r}{\Omg^2}(\rd_v\phi)^2.
\end{equation}
By the estimates in Theorem \ref{RS.prop}, we have $\f{\rd_v r}{\Omg^2}\leq C_{M,{\bf e},s,D}$ on $\{u+v=A_1\}$. Thus by integrating \eqref{R.eqn}, we get
\begin{equation*}
\begin{split}
\left\vert \f{\rd_v r}{\Omg^2} \right\vert (u,v) \leq &C_{M,{\bf e},s,D}+\int_{-u+A_1}^v \f{r}{\Omg^2}(\rd_v\phi)^2(u,v')\,dv'\\
\leq &C_{M,{\bf e},s,D}+C_{M,{\bf e},s,D}\int_{-u+A_1}^{v} e^{2\kappa_-(v'+u)} (v')^{-2s}\, dv' \leq C_{M,{\bf e},s,D}+C_{M,{\bf e},s,D} e^{2\kappa_-(v+u)} v^{-2s},
\end{split}
\end{equation*}
where for the last inequality we have used\footnote{with $\kappa_+$ replaced by $\kappa_-$.} Lemma \ref{lemma.1}. The final conclusion can be derived after multiplying by $\Omg^2$, which using the estimates in Proposition \ref{Psi.C0}, can be bounded up to a constant by $e^{-2\kappa_-(v+u)}\leq C_{M,{\bf e},s,D}\Omg^2_{RN}$ in this region.
\end{proof}

\section{Blow up on the Cauchy horizon: Proof of Theorem \ref{final.blow.up.step}}\label{sec.blow.up}

The goal of this subsection is to prove Theorem \ref{final.blow.up.step}. We will first briefly describe the main ideas of the proof in Section \ref{idea.instab}. In Section \ref{sec.instab.setup}, we will describe the notations and the setup of the proof. The three main types of $L^2$ estimates,  namely the almost energy conservation, the integrated local energy decay and the red-shift estimates, are proven in Sections \ref{sec.almost.energy}, \ref{sec.ILED} and \ref{sec.red.shift} respectively. In Section \ref{sec.instab.together}, we then put together all the estimates to prove the blow up statement \eqref{blow.up.interior}. Finally, in Section \ref{sec.blow.up.dvr}, we use the blow up \eqref{blow.up.interior} to derive the blow up \eqref{blow.up.interior.dvr}.

\subsection{Idea of the proof}\label{idea.instab}

We prove \eqref{blow.up.interior} in Theorem \ref{final.blow.up.step} by showing its contrapositive, namely, we assume that \eqref{blow.up.interior} fails for some $u<u_s$ and deduce that the condition \eqref{final.blow.up.step.assumption} must also fail on the event horizon. (Once \eqref{blow.up.interior} is proved, \eqref{blow.up.interior.dvr} follows straightforwardly - see Section \ref{sec.blow.up.dvr}.) The main idea behind the proof is to view the \underline{nonlinear} \underline{unknown} spacetime as a perturbation of Reissner--Nordstr\"om and to apply the argument\footnote{Recall that the proof Theorem \ref{linear.thm} consists of an argument in the interior of the black hole region and an argument in the exterior of the black hole region. More precisely, we apply here the argument in \cite{LO.instab} relevant to the interior region: see point (1) in the paragraph after the statement of Theorem \ref{linear.thm}.} in the proof of Theorem \ref{linear.thm}. Here, we in particular use the bounds we obtained in the previous section to show that the spacetime in question is indeed close to Reissner--Nordstr\"om in an appropriate sense. 

We quickly recall here our argument in \cite{LO.instab} for proving Theorem \ref{linear.thm}. The main idea is to recast this as a decay problem. More precisely, we show that if (the linear analogue of) the conclusion of Theorem \ref{final.blow.up.step} is assumed to be false for some $u$, then we can solve the wave equation towards $i^+$ starting from the black hole interior and prove a strong enough decay bound on the event horizon which contradicts (the linear analogue of) the assumption of Theorem \ref{final.blow.up.step}. Here, a key observation is that 
the original ``blue-shift'' effect from the point of view of the forward problem (which is the source of the instability at the first place) becomes a ``red-shift'' effect when the wave equation is solved in this direction. We then combine the following three types of $L^2$ estimates: (1) an energy identity, (2) an integrated local energy decay estimate and (3) the red-shift estimates near both horizons to show the desired decay bounds on the event horizon.


On the other hand, for each of the above estimates, one faces the following challenges in applying the argument in \cite{LO.instab} to the nonlinear setting at hand:
\begin{enumerate}
\item (Almost energy conservation) Unlike in \cite{LO.instab}, an exact conservation law does not hold in our setting. Instead, we only have an ``almost conservation law'' with error terms that decays according to the stability results proved in the previous section. In particular, in order to close the estimate for the almost energy conservation law, we must couple it with both the integrated local energy decay estimate and the red-shift estimates.
\item (Integrated local energy decay estimate) Unlike in exact Reissner--Nordstr\"om, in our setting $\lambda$ and $\nu$ \underline{cannot} be controlled by $\Omg^2$. Instead, the differences $\lambda-\lambda_{RN}$ and $\nu-\nu_{RN}$ only decay polynomially in either $|u|$ or $v$. To deal with this, we in particular use a stronger integrated local energy decay estimate with a weight which is not smooth at the event horizon and the Cauchy horizon.\footnote{This estimate is reminiscent of the ``irregular red-shift vector field'' of Dafermos-Rodnianski \cite{DRS}. }
\item (Red-shift estimates) Moreover, in the proof of the ``red-shift'' estimates near the Cauchy horizon (Proposition \ref{prop:intr:red-shift}), the weaker bounds that we have for $\lambda$ compared to the Reissner--Nordstr\"om case give much less room for the argument. Here, it is crucial that we have obtained the improved estimate in Proposition \ref{dvr.improved}. (It is also for this estimate that we need to impose the condition $\alpha_0<4s-2$.)
\end{enumerate}

\subsection{Setting up the contradiction argument and notations used in this section}\label{sec.instab.setup}


Before we proceed, we set up some notation for this section. For $\tau\geq \tau_0$, where $\tau_0$ is a large parameter to be chosen later, let
\begin{equation*}
	\Gmm_{\tau} := \Gmm^{(1)}_{\tau} \cup \Gmm^{(2)}_{\tau},
\end{equation*}
where
\begin{equation*}
	\Gmm^{(1)}_{\tau} = 	\set{(-\tau, v) : v \geq \tau}, \quad
	\Gmm^{(2)}_{\tau} = \set{(u, \tau) : u \leq -\tau}.
\end{equation*}

We denote by $\EH$ the set $C_{-\infty}$ as in Theorem \ref{main.theorem.C0.stability} and denote by $\CH$ the boundary $\{V=1\}$ as in Proposition \ref{C0.1}. We also define
\begin{equation*}
	\CH(\tau_{1}, \tau_{2}) = \CH \cap \set{-\tau_{2} \leq u \leq -\tau_{1}}, \quad
	\EH(\tau_{1}, \tau_{2}) = \EH \cap \set{\tau_{1} \leq v \leq \tau_{2}}.
\end{equation*}
Denote by $\calD(\tau_{1}, \tau_{2})$ the region bounded by $\Gmm^{(1)}_{\tau_{1}}$, $\Gmm^{(2)}_{\tau_{1}}$, $\CH(\tau_{1}, \tau_{2})$, $\EH(\tau_{1}, \tau_{2})$, $\Gmm^{(1)}_{\tau_{2}}$, $\Gmm^{(2)}_{\tau_{2}}$. A Penrose diagram representation of these objects is provided in Figure~\ref{fig:interior}. Note that we will be integrating on the sets $\calD(\tau_{1}, \tau_{2})$, $\Gmm^{(1)}_{\tau_{1}}$, $\Gmm^{(2)}_{\tau_{1}}$, $\CH(\tau_{1}, \tau_{2})$, $\EH(\tau_{1}, \tau_{2})$, $\Gmm^{(1)}_{\tau_{2}}$, $\Gmm^{(2)}_{\tau_{2}}$. 

\begin{figure}[h]
\begin{center}
\def\svgwidth{220px}
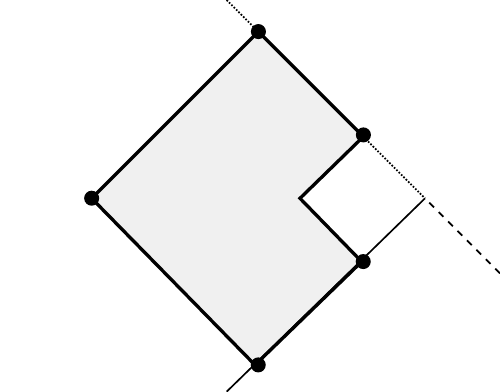 
\caption{} \label{fig:interior}
\end{center}
\end{figure}

We introduce the following conventions for integration: On the null hypersurfaces $\Gmm^{(1)}_\tau$ and $\EH$, we integrate with respect to the measure $dv$; on the null hypersurfaces $\Gmm^{(2)}_\tau$ and $\CH$, we integrate with respect to the measure $du$; in the spacetime region $\calD(\tau_1,\tau_2)$, we integrate with respect to the measure $du\, dv$. Notice that the measure $du\, dv$ is \underline{not} the volume form induced by the metric. To improve readability, we will use the convention that $\int$ denotes an integral along a null hypersurface while $\iint$ denotes an integral in a spacetime region.

We will prove \eqref{blow.up.interior} in Theorem \ref{final.blow.up.step} by contradiction. To show \eqref{blow.up.interior}, it suffices to prove that there exists $\tau_0$ sufficiently large such that $\int_{\Gmm^{(1)}_{\tau_0'}} \log_+^{\alp_{0}}(\f{1}{\Omg}) (\rd_{v} \phi)^{2}=\infty$ for every $\tau_0'\geq \tau_0$. Assume for the sake of contradiction that for every $\tau_0$, there exists $\tau_0'\geq \tau_0$ such that the following holds:
\begin{equation} \label{contra.2}
	E_{\tau_0'}:=\int_{\Gmm^{(1)}_{\tau_0'}} \log_+^{\alp_{0}}(\f{1}{\Omg}) (\rd_{v} \phi)^{2} <\infty
\end{equation}
for $\alp_0$ as in Theorem \ref{final.blow.up.step}.

Our goal now is to show that \eqref{contra.2} together with the estimates derived in Theorem \ref{main.theorem.C0.stability} lead to the bound
\begin{equation}\label{contra.2.goal}
	\int_{\EH\cap\{v\geq \tau_0'\}} \f{v^{\alp_0}}{\log^2(1+v)} (\rd_{v} \phi)^{2} < \infty,
\end{equation}
along $\mathcal H^+$ towards $i^{+}$, which would in particular imply that \eqref{final.blow.up.step.assumption} does not hold. Notice that for $\tau_0$ sufficiently large, $\mathcal D(\tau_0,\infty)$ is indeed in the region where the estimates in Theorem \ref{main.theorem.C0.stability} apply. In particular, near $\CH$, the weight $\log (\frac{1}{\Omg})$ is comparable to $u+v$ (with constants depending on $M, \bfe, E, s$) thanks to \eqref{eq:Omg-CH} and Theorem~\ref{main.theorem.C0.stability}.

Before we proceed, we make a notational convention in the rest of this subsection. {\bf In the remainder of this subsection, we will use $C$ to denote a general (large) constant that depends on $M$, ${\bf e}$, $E$, $s$ and $\alp_0$.} We will use the notation $C_{A_+,A_-}$ to denote a constant that depends on $A_-$, $A_+$, $M$, ${\bf e}$, $E$, $s$ and $\alp_0$ for $A_-$ and $A_+$ that will be defined below (see \eqref{chi1.def}, \eqref{chi2.def} and Propositions \ref{prop:intr:red-shift} and \ref{prop:intr:red-shift:EH}). We will also use $C^{-1}$ to denote a small constant.

In Propositions \ref{prop:intr:energy}-\ref{prop:intr:red-shift:final}, we will derive some estimates for the scalar field in the interior of the black hole. Then, in Proposition \ref{prop:intr:decay}, we will use the assumption \eqref{contra.2} to derive \eqref{contra.2.goal}. 

In the estimates for the scalar field, we will need to define some cutoff functions. Let $\chi_{-}$ and $\chi_{+}$ be smooth positive functions depending only on $u+v$ such that
\begin{equation}\label{chi1.def}
	\chi_{-}(u,v) = 
	\left\{
\begin{array}{cl}	
	1 & \hbox{ for } u+v\geq A_- \\
	0 & \hbox{ for } u+v\leq A_- -1,
\end{array} \right.
\end{equation}
and
\begin{equation}\label{chi2.def}
	\chi_{+}(u,v) = 
	\left\{
\begin{array}{cl}	
	1 & \hbox{ for } u+v\leq A_+ \\
	0 & \hbox{ for } u+v \geq A_+ +1,
\end{array}, \right.
\end{equation}
where $A_- >A_+$ are constants to be determined later. The supports of $\chi_{+}$ and $\chi_{-}$ are depicted in Figure~\ref{fig:interior}.

\subsection{Almost conserved energy}\label{sec.almost.energy}

We begin with an ``almost conserved energy'' in the direction of $i^{+}$. To derive this, we notice that solutions to the linear wave equation on the Reissner--Nordstr\"om spacetime obeys a conservation law and show that in the coupled setting, the deviation from exact Reissner--Nordstr\"om can be controlled.\footnote{In the coupled setting, there is also a conservation law associated to the renormalized Hawking mass. We do not apply this in our setting in particular because the renormalized Hawking mass may be infinite at the Cauchy horizon (see Remark \ref{rem.hawking.mass}).}
\begin{proposition} \label{prop:intr:energy}
Let $\gamma:=\min\{2, s-1, 4s-\alp_0-1\}$. The following holds for any $A_-, A_+\in \mathbb R$ with a constant $C_{A_-,A_+}>0$ depending on $A_-$ and $A_+$: There exists $\tau_0$ sufficiently large such that for every $\tau_{1}, \tau_{2}$ satisfying $\tau_0 \leq \tau_{1} \leq \tau_{2}$, we have\footnote{Notice that the for $\rd_u\phi$, the weight grows towards the event horizon and is degenerate towards the Cauchy horizon. For $\rd_v\phi$, the weight has the opposite behavior, i.e., it grows towards the Cauchy horizon and is degenerate towards the event horizon. As we will see below (see Propositions \ref{prop:intr:red-shift} and \ref{prop:intr:red-shift:EH}), these terms can be controlled since we have stronger bounds for $\rd_u\phi$ near the event horizon and for $\rd_v\phi$ near the Cauchy horizon.}
\begin{equation} \label{eq:intr:energy}
\begin{aligned}
&\hskip-2em
	\int_{\Gmm_{\tau_{2}}^{(1)}} r^{2} (\rd_{v} \phi)^{2} + \int_{\Gmm_{\tau_{2}}^{(2)}} r^{2} (\rd_{u} \phi)^{2} 
	+ \int_{\CH(\tau_{1}, \tau_{2})} r^{2} (\rd_{u} \phi)^{2} + \int_{\EH(\tau_1,\tau_2)} r^{2} (\rd_{v} \phi)^{2} \\
	\leq & \int_{\Gmm^{(1)}_{\tau_{1}}} r^{2} (\rd_{v} \phi)^{2} + \int_{\Gmm^{(2)}_{\tau_{1}}} r^{2} (\rd_{u} \phi)^{2} +C_{A_+,A_-}\tau_1^{-s}\iint_{\calD(\tau_{1}, \tau_{2})} \left( \log_+^{-\gamma}(\f{1}{\Omg})+\chi_+\log_+^\gamma(\f{1}{\Omg}) \right)(\rd_u\phi)^2\\
	&+C_{A_+,A_-}\tau_1^{-s}\iint_{\calD(\tau_{1}, \tau_{2})}\left( \log_+^{-\gamma}(\f{1}{\Omg})+\chi_-\log_+^\gamma(\f{1}{\Omg}) \right)(\rd_v\phi)^2.
\end{aligned}
\end{equation}
\end{proposition}
\begin{proof} 
Integrating by parts over $\calD(\tau_{1}, \tau_{2})$ the identity
\begin{equation*}
	0 = r^{2} (\rd_{u} \phi - \rd_{v} \phi) \bb( \rd_{u} \rd_{v} \phi + \frac{\dvr}{r} \rd_{u} \phi + \frac{\dur}{r} \rd_{v} \phi \bb),
\end{equation*} 
and noticing that $\dvr_{RN} = \dur_{RN}$ in our coordinates, we get
\begin{equation*} 
\begin{aligned}
&\hskip-2em
	\int_{\Gmm_{\tau_{2}}^{(1)}} r^{2} (\rd_{v} \phi)^{2} + \int_{\Gmm_{\tau_{2}}^{(2)}} r^{2} (\rd_{u} \phi)^{2} 
	+ \int_{\CH(\tau_{1}, \tau_{2})} r^{2} (\rd_{u} \phi)^{2} + \int_{\EH} r^{2} (\rd_{v} \phi)^{2} \\
	= &  \int_{\Gmm^{(1)}_{\tau_{1}}} r^{2} (\rd_{v} \phi)^{2} + \int_{\Gmm^{(2)}_{\tau_{1}}} r^{2} (\rd_{u} \phi)^{2} + \iint_{\calD(\tau_{1}, \tau_{2})} r(\dvr-\dvr_{RN}-\dur+\dur_{RN} )(\rd_u\phi)(\rd_v\phi).
\end{aligned}
\end{equation*}
To handle the last term, we divide the integral into three pieces:
$$\iint_{\calD(\tau_{1}, \tau_{2})}=\underbrace{\iint_{\calD(\tau_{1}, \tau_{2})\cap\{u+v\leq A_+\}}}_{=:I}+\underbrace{\iint_{\calD(\tau_{1}, \tau_{2})\cap\{A_+<u+v< A_-\}}}_{=:II}+\underbrace{\iint_{\calD(\tau_{1}, \tau_{2})\cap\{u+v\geq A_-\}}}_{=:III}.$$
For each of these pieces, we use the fact that $r$ is bounded, as well as
$$\sup_{(u,v)\in \calD(\tau_1,\infty)}\left(|\dvr-\dvr_{RN}|+|\dur-\dur_{RN}|\right)(u,v)\leq C\tau_1^{-s}$$
(both of which follow from Theorem~\ref{main.theorem.C0.stability}), and apply the Cauchy-Schwarz inequality. Firstly, for the term $I$ localized near the event horizon, we have the bound
\begin{equation*}
\begin{split}
&\iint_{\calD(\tau_{1}, \tau_{2})\cap\{u+v\leq A_+\}} r(|\dvr-\dvr_{RN}|+|\dur-\dur_{RN}| )|\rd_u\phi||\rd_v\phi|\\
\leq &C_{A_+,A_-}\tau_1^{-s}\left(\iint_{\calD(\tau_{1}, \tau_{2})} \chi_+\log_+^{\gamma}(\f{1}{\Omg})(\rd_u\phi)^2+\iint_{\calD(\tau_{1}, \tau_{2})} \log_+^{-\gamma}(\f{1}{\Omg})(\rd_v\phi)^2\right).
\end{split}
\end{equation*}
For the second term $II$, since we have upper and lower bounds for $r$ and $\Omg$ (depending on $A_+$ and $A_-$) thanks to the Reissner--Nordstr\"om computation in Section~\ref{sec.geometry} and Theorem~\ref{main.theorem.C0.stability}, we can put in weights that degenerate to obtain
\begin{equation*}
\begin{split}
&\iint_{\calD(\tau_{1}, \tau_{2})\cap\{A_+<u+v< A_-\}} r(|\dvr-\dvr_{RN}|+|\dur-\dur_{RN}| )|\rd_u\phi||\rd_v\phi|\\
\leq &C_{A_+,A_-}\tau_1^{-s}\left(\iint_{\calD(\tau_{1}, \tau_{2})} \log_+^{-\gamma}(\f{1}{\Omg})(\rd_u\phi)^2+\iint_{\calD(\tau_{1}, \tau_{2})} \log_+^{-\gamma}(\f{1}{\Omg})(\rd_v\phi)^2\right).
\end{split}
\end{equation*}
Finally, the third integral $III$, i.e., the term localized near the Cauchy horizon, we have 
\begin{equation*}
\begin{split}
&\iint_{\calD(\tau_{1}, \tau_{2})\cap\{u+v\geq A_-\}} r(|\dvr-\dvr_{RN}|+|\dur-\dur_{RN}| )|\rd_u\phi||\rd_v\phi|\\
\leq &C_{A_+,A_-}\tau_1^{-s}\left(\iint_{\calD(\tau_{1}, \tau_{2})} \log_+^{-\gamma}(\f{1}{\Omg})(\rd_u\phi)^2+\iint_{\calD(\tau_{1}, \tau_{2})} \chi_-\log_+^\gamma(\f{1}{\Omg})(\rd_v\phi)^2\right).
\end{split}
\end{equation*}
Combining, we obtain
\begin{equation}\label{intr.cross.bulk}
\begin{split}
&\iint_{\calD(\tau_{1}, \tau_{2})} r(|\dvr-\dvr_{RN}|+|\dur-\dur_{RN}| )|\rd_u\phi||\rd_v\phi|\\
\leq &C_{A_+,A_-}\tau_1^{-s}\iint_{\calD(\tau_{1}, \tau_{2})} \left( \log_+^{-\gamma}(\f{1}{\Omg})+\chi_+\log_+^{\gamma}(\f{1}{\Omg}) \right)(\rd_u\phi)^2\\
	&+C_{A_+,A_-}\tau_1^{-s}\iint_{\calD(\tau_{1}, \tau_{2})} \left( \log_+^{-\gamma}(\f{1}{\Omg})+\chi_-\log_+^{\gamma}(\f{1}{\Omg}) \right)(\rd_v\phi)^2.
\end{split}
\end{equation}
This concludes the proof of the proposition.
\end{proof}

\subsection{Integrated local energy decay estimate}\label{sec.ILED}

Using the energy inequality, we now establish an integrated local energy decay estimate.
\begin{proposition} \label{prop:intr:ILED}
Let $\gamma:=\min\{2, s-1, 4s-\alp_0-1\}$. The following holds for any $A_-, A_+\in \mathbb R$ with a constant $C_{A_-,A_+}>0$ depending on $A_-$ and $A_+$: For $\tau_0$ sufficiently large, and for every $\tau_{1}, \tau_{2}$ such that $\tau_0 \leq \tau_{1} \leq \tau_{2}$, we have
\begin{equation} \label{eq:intr:ILED}
\begin{split}
	&\iint_{\calD(\tau_{1}, \tau_{2})} \log_+^{-\gamma}(\f 1{\Omg}) \bb( (\rd_{u} \phi)^{2} + (\rd_{v} \phi)^{2} \bb)\\
	\leq &C_{A_+,A_-} \bb( \int_{\Gmm^{(1)}_{\tau_{1}}} (\rd_{v} \phi)^{2} + \int_{\Gmm^{(2)}_{\tau_{1}}} (\rd_{u} \phi)^{2} \bb)+C_{A_+,A_-}\tau_1^{-s}\iint_{\calD(\tau_{1}, \tau_{2})} \chi_+\log_+^{\gamma}(\f{1}{\Omg})(\rd_u\phi)^2\\
	&+C_{A_+,A_-}\tau_1^{-s}\iint_{\calD(\tau_{1}, \tau_{2})} \chi_-\log_+^{\gamma}(\f{1}{\Omg})(\rd_v\phi)^2.
\end{split}
\end{equation}
\end{proposition}
\begin{proof} 
The assumption implies that $\gamma>1$. We can thus define $w(u,v)$ to be a positive $C^1$ function as follows:
$$w(u,v)=2-\f{\gamma-1}{2}\int_{-\infty}^{u+v} \f{dx}{(1+|x|)^{\gamma}}.$$
Consequently, $w$ satisfies
$$1\leq w\leq 2,\quad \rd_u w\leq -C_{w}^{-1}\log_+^{-\gamma}(\f 1{\Omg}),\quad \rd_v w\leq -C_{w}^{-1}\log_+^{-\gamma}(\f 1{\Omg})$$
for some positive constant $C_w$.

Given $N \geq 0$ to be chosen, consider
\begin{equation}\label{intr.ILED.bulk.main}
\begin{split}
	0 =& \iint_{\calD(\tau_{1}, \tau_{2})}	
			w^{N} (\rd_{u} \phi + \rd_{v} \phi) \bb( \rd_{u} \rd_{v} \phi + \frac{\dvr}{r} \rd_{u} \phi + \frac{\dur}{r} \rd_{v} \phi \bb) \, \ud u \ud v \\
	=& \iint_{\calD(\tau_{1}, \tau_{2})} - \frac{N}{2} (\rd_{v} w) w^{N-1} (\rd_{u} \phi)^{2} - \frac{N}{2} (\rd_{u} w) w^{N-1} (\rd_{v} \phi)^{2}\\
		& + \iint_{\calD(\tau_{1}, \tau_{2})} \f{w^{N}}{r} \bb( \lambda(\rd_{u} \phi)^{2} + \nu(\rd_{v} \phi)^{2} + (\dvr+\dur) \rd_{u} \phi \rd_{v} \phi \bb) + (\hbox{boundary terms})\\
	\geq & \f{C_{w}^{-1} N}2 \iint_{\calD(\tau_{1}, \tau_{2})} w^{N-1} \log_+^{-\gamma}(\f 1{\Omg}) \left((\rd_{u} \phi)^{2} + (\rd_{v} \phi)^{2}\right)-C \iint_{\calD(\tau_{1}, \tau_{2})} w^N \Omg^2\left( (\rd_u\phi)^2+(\rd_v\phi)^2\right)\\
	& -C \iint_{\calD(\tau_{1}, \tau_{2})} \f{w^{N}}{r} \bb( |\lambda-\lambda_{RN}|(\rd_{u} \phi)^{2} + |\nu-\nu_{RN}|(\rd_{v} \phi)^{2}\bb)\\
	&-C\iint_{\calD(\tau_{1}, \tau_{2})} (|\dvr-\dvr_{RN}|+|\nu-\nu_{RN}|) |\rd_{u} \phi \rd_{v} \phi| - |\hbox{boundary terms}|,
\end{split}
\end{equation}
where in the last line we have used that $\lambda_{RN}=\nu_{RN}=-\Omg^2_{RN}$ and that $\Omg_{RN}$ and $\Omg$ are comparable. 
Choosing $N$ sufficiently large, we see that the first line on the right hand side controls (up to a constant factor) the following:
\begin{equation}\label{intr.ILED.bulk.1}
\begin{split}
&\f{C_{w}^{-1} N}2 \iint_{\calD(\tau_{1}, \tau_{2})} w^{N-1} \log_+^{-\gamma}(\f 1{\Omg}) \left((\rd_{u} \phi)^{2} + (\rd_{v} \phi)^{2}\right)-C \iint_{\calD(\tau_{1}, \tau_{2})} w^N \Omg^2\left( (\rd_u\phi)^2+(\rd_v\phi)^2\right)\\
&\geq C^{-1}	N\iint_{\calD(\tau_{1}, \tau_{2})} w^{N} \log_+^{-\gamma}(\f 1{\Omg}) \left( (\rd_{u} \phi)^{2} + (\rd_{v} \phi)^{2} \right).
\end{split}
\end{equation}
It remains to bound the bulk terms containing $|\dvr-\dvr_{RN}|$ and $|\dur-\dur_{RN}|$ and the boundary terms. We consider each of the bulk terms. Firstly, by Theorem \ref{main.theorem.C0.stability},
\begin{equation}\label{ILED.dvrdiff.1}
|\lambda-\lambda_{RN}|(u,v)\leq C v^{-s}
\end{equation}
and that for $u+v\geq A_+$ and $v \geq \tau_{1}$, we additionally have
\begin{equation}\label{ILED.dvrdiff.2}
|\lambda-\lambda_{RN}|(u,v)\leq C_{A_+} \tau_{1}^{-(s - \gmm)} \log_+^{-\gmm}(\f 1{\Omg})(u,v),
\end{equation}
since $\gmm \leq 2 < s$.
Therefore,
\begin{equation}\label{intr.ILED.bulk.2}
\begin{split}
\iint_{\calD(\tau_{1}, \tau_{2})} \f{w^{N}}{r}  |\lambda-\lambda_{RN}|(\rd_{u} \phi)^{2} 
\leq &C_{A_+}\iint_{\calD(\tau_{1}, \tau_{2})} w^N (\tau_{1}^{-(s-\gmm)} \log_+^{-\gmm}(\f{1}{\Omg})+\tau_1^{-s}\chi_+) (\rd_{u} \phi)^{2} .
\end{split}
\end{equation}
Here, we have used \eqref{ILED.dvrdiff.1} and $v\geq \tau_1$ for $u+v< A_+$ and used \eqref{ILED.dvrdiff.2} for $u+v\geq A_+$.
Similarly, by Theorem \ref{main.theorem.C0.stability}, $|\nu-\nu_{RN}|(u,v)\leq C|u|^{-s}$. Hence, for $u+v\leq A_-$  and $-u \geq \tau_{1}$, we have
$$|\nu-\nu_{RN}|(u,v)\leq C_{A_-} \tau_{1}^{-(s- \gmm)} \log_+^{-\gmm}(\f 1{\Omg})(u,v).$$
Consequently,
\begin{equation}\label{intr.ILED.bulk.3}
\begin{split}
\iint_{\calD(\tau_{1}, \tau_{2})} \f{w^{N}}{r}  |\nu-\nu_{RN}|(\rd_{v} \phi)^{2} 
\leq &C_{A_-}\iint_{\calD(\tau_{1}, \tau_{2})} w^N (\tau_{1}^{-(s-\gmm)} \log_+^{-\gmm}(\f{1}{\Omg})+\tau_1^{-s}\chi_-) (\rd_{v} \phi)^{2} .
\end{split}
\end{equation}
Finally, for the remaining bulk term, we estimate it in a similar manner as \eqref{intr.cross.bulk} to get
\begin{equation}\label{intr.ILED.bulk.4}
\begin{split}
&\iint_{\calD(\tau_{1}, \tau_{2})} \f{w^{N}}{r}  (|\dvr-\dvr_{RN}|+|\nu-\nu_{RN}|) |\rd_{u} \phi \rd_{v} \phi| \\
\leq &C_{A_+,A_-}\tau_1^{-s}\iint_{\calD(\tau_{1}, \tau_{2})} w^N \left( \log_+^{-\gamma}(\f{1}{\Omg})+\chi_+\log_+^{\gamma}(\f{1}{\Omg})\right)(\rd_u\phi)^2\\
	&+C_{A_+,A_-}\tau_1^{-s}\iint_{\calD(\tau_{1}, \tau_{2})} w^N \left( \log_+^{-\gamma}(\f{1}{\Omg})+\chi_-\log_+^{\gamma}(\f{1}{\Omg}) \right)(\rd_v\phi)^2.
\end{split}
\end{equation}
Without loss of generality, we may assume that $1 \leq \tau_{0} \leq \tau_{1}$, which implies $\tau_{1}^{-s} \leq \tau_{1}^{-(s - \gmm)}$.
Combining the estimates in \eqref{intr.ILED.bulk.1}, \eqref{intr.ILED.bulk.2}, \eqref{intr.ILED.bulk.3} and \eqref{intr.ILED.bulk.4}, we then obtain
\begin{equation*}
\begin{split}
	&|\hbox{boundary terms}| \\
	\geq & C^{-1} N \iint_{\calD(\tau_{1}, \tau_{2})} w^{N-1} \log_+^{-\gamma}(\f 1{\Omg}) \left((\rd_{u} \phi)^{2} + (\rd_{v} \phi)^{2}\right)\\
	&-C_{A_+,A_-}\iint_{\calD(\tau_{1}, \tau_{2})} w^N \left( \tau_1^{-(s-\gmm)}\log_+^{-\gamma}(\f{1}{\Omg})+ \tau_1^{-s} \chi_+\log_+^{\gamma}(\f{1}{\Omg}) \right)(\rd_u\phi)^2\\
	&-C_{A_+,A_-}\iint_{\calD(\tau_{1}, \tau_{2})} w^N \left( \tau_1^{-(s-\gmm)} \log_+^{-\gamma}(\f{1}{\Omg})+ \tau_1^{-s} \chi_-\log_+^{\gamma}(\f{1}{\Omg}) \right)(\rd_v\phi)^2.
\end{split}
\end{equation*}
On the other hand, the boundary terms\footnote{One can easily check that the boundary integrals on $\Gamma^{(1)}_{\tau_2}$ and $\EH(\tau_1,\tau_2)$ only have $(\rd_v\phi)^2$ terms and the boundary integrals on $\Gamma^{(2)}_{\tau_2}$ and $\CH(\tau_1,\tau_2)$ only have $(\rd_u\phi)^2$ terms.} can be controlled by Proposition~\ref{prop:intr:energy}, the bound $1 \leq w \leq 2$, and the upper and lower bounds of $r$ so that we have
\begin{equation}\label{interior:ILED:last}
\begin{split}
&N\iint_{\calD(\tau_{1}, \tau_{2})} w^{N} \log_+^{-\gamma}(\f 1{\Omg}) \bb( (\rd_{u} \phi)^{2} + (\rd_{v} \phi)^{2} \bb)\\
\leq &C 2^N\left(\int_{\Gmm^{(1)}_{\tau_{1}}} r^{2} (\rd_{v} \phi)^{2} + \int_{\Gmm^{(2)}_{\tau_{1}}} r^{2} (\rd_{u} \phi)^{2}\right)\\
& +C_{A_+,A_-} 2^N \iint_{\calD(\tau_{1}, \tau_{2})} w^N \left( \tau_1^{-(s-\gmm)} \log_+^{-\gamma}(\f{1}{\Omg})+ \tau_1^{-s} \chi_+\log_+^{\gamma}(\f{1}{\Omg}) \right) (\rd_u\phi)^2\\
	&+C_{A_+,A_-} 2^N \iint_{\calD(\tau_{1}, \tau_{2})} w^N \left( \tau_1^{-(s-\gmm)} \log_+^{-\gamma}(\f{1}{\Omg})+ \tau_1^{-s} \chi_+\log_+^{\gamma}(\f{1}{\Omg}) \right) (\rd_v\phi)^2.
\end{split}
\end{equation}
Since $s>2 \geq \gmm$, we can take $\tau_0$ to be sufficiently large after fixing $N>0$ such that for $C_{A_+,A_-}$ as in \eqref{interior:ILED:last}, we have
$$ C_{A_+,A_-} 2^N \tau_0^{-(s-\gmm)} \leq \f{N}2.$$
After subtracting
$$\f N2\iint_{\calD(\tau_{1}, \tau_{2})} w^{N} \log_+^{-\gamma}(\f 1{\Omg}) \bb( (\rd_{u} \phi)^{2} + (\rd_{v} \phi)^{2} \bb)$$
on both sides of \eqref{interior:ILED:last}, we obtain
\begin{equation*}
\begin{split}
&\frac{N}{2} \iint_{\calD(\tau_{1}, \tau_{2})} w^{N} \log_+^{-\gamma}(\f 1{\Omg}) \bb( (\rd_{u} \phi)^{2} + (\rd_{v} \phi)^{2} \bb)\\
\leq &C 2^N\left(\int_{\Gmm^{(1)}_{\tau_{1}}} r^{2} (\rd_{v} \phi)^{2} + \int_{\Gmm^{(2)}_{\tau_{1}}} r^{2} (\rd_{u} \phi)^{2}\right)\\
& +C_{A_+,A_-} 2^N\tau_1^{-s}\iint_{\calD(\tau_{1}, \tau_{2})} w^N \chi_+\log_+^{\gamma}(\f{1}{\Omg})(\rd_u\phi)^2\\
	&+C_{A_+,A_-} 2^N\tau_1^{-s}\iint_{\calD(\tau_{1}, \tau_{2})} w^N \chi_-\log_+^{\gamma}(\f{1}{\Omg})(\rd_v\phi)^2.
\end{split}
\end{equation*}
Since $N$ is now fixed, we can absorb it into the constant $C_{A_+,A_-}$ to obtain the desired conclusion.
\end{proof}
An immediate corollary is that we can improve the estimate in Proposition \ref{prop:intr:energy} by controlling the bulk term with the estimate \eqref{eq:intr:ILED} in Proposition \ref{prop:intr:ILED}:
\begin{proposition} \label{prop:intr:energy.1}  
There exists $\tau_0$ sufficiently large such that for every $\tau_{1}, \tau_{2}$ satisfying $\tau_0 \leq \tau_{1} \leq \tau_{2}$, we have
\begin{equation} \label{eq:intr:energy.1}
\begin{aligned}
&\hskip-2em
	\int_{\Gmm_{\tau_{2}}^{(1)}} r^{2} (\rd_{v} \phi)^{2} + \int_{\Gmm_{\tau_{2}}^{(2)}} r^{2} (\rd_{u} \phi)^{2} 
	+ \int_{\CH(\tau_{1}, \tau_{2})} r^{2} (\rd_{u} \phi)^{2} + \int_{\EH(\tau_1,\tau_2)} r^{2} (\rd_{v} \phi)^{2} \\
	\leq &C_{A_+,A_-} \bb( \int_{\Gmm^{(1)}_{\tau_{1}}} (\rd_{v} \phi)^{2} + \int_{\Gmm^{(2)}_{\tau_{1}}} (\rd_{u} \phi)^{2} \bb)+C_{A_+,A_-}\tau_1^{-s}\iint_{\calD(\tau_{1}, \tau_{2})} \chi_+\log_+^{\gamma}(\f{1}{\Omg})(\rd_u\phi)^2\\
	&+C_{A_+,A_-}\tau_1^{-s}\iint_{\calD(\tau_{1}, \tau_{2})} \chi_-\log_+^{\gamma}(\f{1}{\Omg})(\rd_v\phi)^2.
\end{aligned}
\end{equation}
\end{proposition}

\subsection{Red-shift estimates}\label{sec.red.shift}

In the next proposition, we prove an estimate that is localized near the Cauchy horizon. Since we are solving the wave equation ``backwards'' near the Cauchy horizon, the \underline{blue-shift} effect becomes a \underline{red-shift} effect as we approach $i^{+}$. This estimate in particular gives a good bulk term near the Cauchy horizon which has a better weight than that in \eqref{eq:intr:ILED}. 
\begin{proposition} \label{prop:intr:red-shift}  
There exists $A_{-,0}>0$ sufficiently large such that the following holds if $A_-\geq A_{-,0}$ with a constant $C_{A_-,A_+}>0$ depending on $A_-$ and $A_+$: For $\tau_0$ sufficiently large, for every $\alp \in [0,\alp_0]$ and for $\tau_0 \leq \tau_{1} \leq \tau_{2}$, we have
\begin{equation} \label{eq:intr:red-shift}
\begin{aligned}
& \hskip-2em
	\int_{\Gmm^{(1)}_{\Gmm_{\tau_{2}}}} \chi_{-} \log_+^{\alp} (\frac{1}{\Omg}) (\rd_{v} \phi)^{2}
	+ \alp \iint_{\calD(\tau_{1}, \tau_{2})} \chi_{-} \log_+^{\alp-1} (\frac{1}{\Omg}) (\rd_{v} \phi)^{2} \\
	\leq & C_{A_+,A_-} \bb( \int_{\Gmm^{(1)}_{\tau_{1}}} \big( 1 + \chi_{-} \log_+^{\alp}(\frac{1}{\Omg}) \big) (\rd_{v} \phi)^{2}
	+ \int_{\Gmm^{(2)}_{\tau_{1}}} (\rd_{u} \phi)^{2} \bb)	\\
	&+C_{A_+,A_-}\tau_1^{-s}\iint_{\calD(\tau_{1}, \tau_{2})} \chi_+\log_+^{\gamma}(\f{1}{\Omg})(\rd_u\phi)^2\\
	&+C_{A_+,A_-}\tau_1^{-s}\iint_{\calD(\tau_{1}, \tau_{2})} \chi_-\log_+^{\gamma}(\f{1}{\Omg})(\rd_v\phi)^2.
\end{aligned}
\end{equation}
\end{proposition}

\begin{proof} 
Since $\Omg$ and $\Omg_{RN}$ are comparable (according to Theorem \ref{main.theorem.C0.stability}), it suffices to derive the desired estimate using $\Omg_{RN}$ in the weight. This has the advantage that we have better control over the derivatives of $\Omg_{RN}$ then that of $\Omg$.

When $\alp = 0$, \eqref{eq:intr:red-shift} follows from \eqref{eq:intr:energy}. Hence it suffices to consider the case $\alp > 0$. We begin with
 \begin{align}
	0 =& \iint_{\calD(\tau_{1}, \tau_{2})} \chi_{-} \log_+^{\alp} \frac{1}{\Omg_{RN}} \rd_{v} \phi 
				\bb( \rd_{u} \rd_{v} \phi + \frac{\dvr}{r} \rd_{u} \phi + \frac{\nu}{r} \rd_{v} \phi \bb) \notag \\
	=& \frac{1}{2} \int_{\Gmm^{(1)}_{\tau_{1}}} \chi_{-} \log_+^{\alp} \frac{1}{\Omg_{RN}} (\rd_{v} \phi)^{2}
		- \frac{1}{2} \int_{\Gmm^{(1)}_{\tau_{2}}} \chi_{-} \log_+^{\alp} \frac{1}{\Omg_{RN}} (\rd_{v} \phi)^{2} \label{eq:intr:red-shift:1} \\
	& - \frac{\alp}{4} \iint_{\calD(\tau_{1}, \tau_{2})} \chi_{-} \log_+^{\alp-1} \frac{1}{\Omg_{RN}} \frac{(- \rd_{u} \Omg_{RN}^{2})}{\Omg_{RN}^{2}} (\rd_{v} \phi)^{2} \label{eq:intr:red-shift:2}\\
	&  - \frac{1}{2} \iint_{\calD(\tau_{1}, \tau_{2})} (\rd_u\chi_{-}) \log_+^{\alp} \frac{1}{\Omg_{RN}} (\rd_{v }\phi)^{2} \label{eq:intr:red-shift:3}\\
	 &  + \iint_{\calD(\tau_{1}, \tau_{2})} \f{\chi_{-}}{r} \log_+^{\alp} \frac{1}{\Omg_{RN}} \lambda(\rd_{v} \phi)(\rd_{u} \phi) \label{eq:intr:red-shift:4}\\
	&  + \iint_{\calD(\tau_{1}, \tau_{2})} \f{\chi_{-}}{r} \log_+^{\alp} \frac{1}{\Omg_{RN}} \nu(\rd_{v} \phi)^2. \label{eq:intr:red-shift:5}
\end{align}
For each of these terms, we either show that it is bounded by the right hand side of \eqref{eq:intr:red-shift} or we show that it has a good sign, i.e., the same sign as the boundary integral on $\Gamma^{(1)}_{\tau_2}$ in \eqref{eq:intr:red-shift:1}.

Recall that on Reissner--Nordstr\"om, we have
\begin{equation*}
	- \rd_{u} \Omg_{RN}^{2} = \rd_{u} 4 \left( 1- \frac{2M}{r_{RN}} + \frac{\bfe^{2}}{r_{RN}^{2}} \right)  = - 2 \frac{\Omg_{RN}^{2}}{r_{RN}^{2}} \left( M - \frac{\bfe^{2}}{r_{RN}} \right).
\end{equation*}
The crucial observation here is that
\begin{equation*}
	M - \frac{\bfe^{2}}{r_{-}} < 0
\end{equation*}
and hence by choosing $A_{-,0}$ to be sufficiently large, so that the support of $\chi_{-}$ is close enough to $\CH$, we have $- \rd_{u} \Omg_{RN}^{2} \geq C^{-1} \Omg_{RN}^{2}$ for some $C > 0$ on the support of $\chi_{-}$ and the space-time integral in \eqref{eq:intr:red-shift:2} has the same sign as the boundary integral on $\Gmm^{(1)}_{\tau_{2}}$ in \eqref{eq:intr:red-shift:1}. (Notice that this term also gives the good bulk term on the left hand side of \eqref{eq:intr:red-shift}.)

The term \eqref{eq:intr:red-shift:3} can be controlled using Proposition~\ref{prop:intr:ILED}, as it is safely localized away from $\CH$. For \eqref{eq:intr:red-shift:4}, we first use the Cauchy-Schwarz inequality to write
\begin{equation*}
	\abs{\eqref{eq:intr:red-shift:4}}
	\leq \frac{\varepsilon}{2} \iint_{\calD(\tau_{1}, \tau_{2})} \chi_{-} \log_+^{\alp-1} (\frac{1}{\Omg}) (\rd_{v} \phi)^{2} + \frac{1}{2 \varepsilon} \iint_{\calD(\tau_{1}, \tau_{2})} \chi_{-} \, \frac{\lambda^2 \log_+^{\alp+1} (\frac{1}{\Omg})}{r^{2} } (\rd_{u} \phi)^{2} .
\end{equation*}
Choosing $\varepsilon > 0$ sufficiently small and using the fact that $\Omg_{RN}^2\sim \Omg^2$ (by Theorem \ref{main.theorem.C0.stability}), the first term can be bounded by \eqref{eq:intr:red-shift:2}. For the second term, recall from Proposition \ref{dvr.improved} that 
$$|\lambda|\leq C(\Omg^2_{RN}+v^{-2s}).$$
In particular, this implies that in the support of $\chi_-$, i.e., when $u+v\geq A_{-}-1$, we have
$$|\lambda^2 \log_+^{\alp+1} (\frac{1}{\Omg}) |\leq C\log_+^{-4s+\alp+1}(\f{1}{\Omg})\leq C\log_+^{-\gamma}(\f{1}{\Omg}).$$
Together with the lower bound on $r$, we thus have
\begin{equation*}
\begin{split}
&\iint_{\calD(\tau_{1}, \tau_{2})} \chi_{-} \, \frac{\lambda^2 \log_+^{\alp+1} (\frac{1}{\Omg})}{r^{2} } (\rd_{u} \phi)^{2}\leq C\iint_{\calD(\tau_{1}, \tau_{2})} \log_+^{-\gamma}(\f{1}{\Omg}) (\rd_{u} \phi)^{2}
\end{split}
\end{equation*}
and the right hand side can be bounded using Proposition~\ref{prop:intr:ILED}. Finally, for the term \eqref{eq:intr:red-shift:5}, notice that since $|\nu-\nu_{RN}|\leq C|u|^{-s}$, by choosing $\tau_0$ sufficiently large, we have $\nu<0$ on $\{u+v=A_{-}-1\}\cap\{\tau\geq \tau_0\}$. By \eqref{eqn.Ray}, $\f{\nu}{\Omg^2}$ is monotonically decreasing and we thus have $\nu\leq 0$ on the support of $\chi_-$. Therefore, \eqref{eq:intr:red-shift:5} has the same sign as the boundary integral on $\Gmm^{(1)}_{\tau_{2}}$ in \eqref{eq:intr:red-shift:1}. 

Combining all these estimates and dropping the term \eqref{eq:intr:red-shift:5} which has a good sign, we obtain the desired conclusion. \qedhere
\end{proof}

Our next proposition is an analogue of Proposition~\ref{prop:intr:red-shift}, but instead localized near the event horizon. As in Proposition~\ref{prop:intr:red-shift}, we capture red-shift effect along the event horizon $\calH^{+}$ as we approach $i^{+}$. In particular, we have a good bulk term for $\rd_u\phi$ in this region which has a better weight compared to that in \eqref{eq:intr:ILED}.

\begin{proposition} \label{prop:intr:red-shift:EH}  
There exists $A_{+,0}<0$ sufficiently negative such that the following holds if $A_+\leq A_{+,0}$ with a constant $C_{A_-,A_+}>0$ depending on $A_-$ and $A_+$: For $\tau_0$ sufficiently large, and for every $\tau_{1}, \tau_{2}$ such that $\tau_0 \leq \tau_{1} \leq \tau_{2}$, we have
\begin{equation} \label{eq:intr:red-shift:EH}
\begin{split}
	&\int_{\Gmm^{(2)}_{\tau_{2}}} \chi_{+} \Omg^{-2}(\rd_{u} \phi)^{2}
	+ \iint_{\calD(\tau_{1}, \tau_{2})} \chi_{+} \Omg^{-2} (\rd_{u} \phi)^{2} \\
	\leq &C_{A_+,A_-} \bb( \int_{\Gmm^{(1)}_{\tau_{1}}} (\rd_{v} \phi)^{2} + \int_{\Gmm^{(2)}_{\tau_{1}}} (1 + \chi_{+} \Omg^{-2}) (\rd_{u} \phi)^{2} \bb)+C_{A_+,A_-}\tau_1^{-s}\iint_{\calD(\tau_{1}, \tau_{2})} \chi_+\log_+^{\gamma}(\f{1}{\Omg})(\rd_u\phi)^2\\
	&+C_{A_+,A_-}\tau_1^{-s}\iint_{\calD(\tau_{1}, \tau_{2})} \chi_-\log_+^{\gamma}(\f{1}{\Omg})(\rd_v\phi)^2.
\end{split}
\end{equation}
\end{proposition}

\begin{proof} 
As in the proof of Proposition \ref{prop:intr:red-shift}, we use $\Omg_{RN}$ instead of $\Omg$ in the weight since $\Omg_{RN}$ and $\Omg$ are comparable. We begin with
 \begin{align}
	0 =& \iint_{\calD(\tau_{1}, \tau_{2})} \chi_{+} \Omg_{RN}^{-2} \rd_{u} \phi 
				\bb( \rd_{u} \rd_{v} \phi + \frac{\dvr}{r} \rd_{u} \phi + \frac{\nu}{r} \rd_{v} \phi \bb) \notag \\
	=& 
		-\frac{1}{2} \int_{\Gmm^{(2)}_{\tau_{1}}} \chi_{+} \Omg_{RN}^{-2} (\rd_{u} \phi)^{2}
		+ \frac{1}{2} \int_{\Gmm^{(2)}_{\tau_{2}}} \chi_{+} \Omg_{RN}^{-2} (\rd_{u} \phi)^{2} 
				\label{eq:intr:red-shift:EH:1} \\
	& + \frac{1}{2} \iint_{\calD(\tau_{1}, \tau_{2})} \chi_{+} \frac{\rd_{v} \Omg_{RN}^{2}}{\Omg_{RN}^{4}} (\rd_{u} \phi)^{2} \label{eq:intr:red-shift:EH:2}\\
	&  - \frac{1}{2} \iint_{\calD(\tau_{1}, \tau_{2})} (\rd_v\chi_{+}) \Omg_{RN}^{-2}(\rd_{u}\phi)^{2} \label{eq:intr:red-shift:EH:3}\\
	 &  + \iint_{\calD(\tau_{1}, \tau_{2})} \frac{1}{r\Omg_{RN}^2} \chi_{+} \lambda (\rd_{u} \phi)^2 \label{eq:intr:red-shift:EH:4}\\
	&  + \iint_{\calD(\tau_{1}, \tau_{2})} \frac{1}{r\Omg_{RN}^2} \chi_{+} \nu(\rd_{u} \phi) (\rd_{v} \phi). \label{eq:intr:red-shift:EH:5}
\end{align}
As in the proof of Proposition \ref{prop:intr:red-shift}, for each of these terms, we either control it or show that it has a good sign. First, we see that if $A_{+,0}$ is chosen to be sufficiently negative, then we have
\begin{equation}\label{RS.EH.est}
	\rd_{v} \Omg^{2}_{RN} = - \rd_{v} 4 \left( 1 - \frac{2M}{r_{RN}} + \frac{\bfe^{2}}{r_{RN}^{2}} \right)  = 2 \frac{\Omg_{RN}^{2}}{r_{RN}^{2}} \left( M - \frac{\bfe^{2}}{r_{RN}^{2}} \right) \geq C^{-1} \Omg_{RN}^{2}
\end{equation}
 on the support of $\chi_{+}$. This inequality follows from the observation that $M - \frac{\bfe^{2}}{r_{+}} > 0$. In particular, \eqref{eq:intr:red-shift:EH:2} has the same sign as the boundary integral on $\Gmm^{(2)}_{\tau_{2}}$.

The rest of the proof proceeds similarly to that of Proposition~\ref{prop:intr:red-shift}. The term \eqref{eq:intr:red-shift:EH:3} can be bounded by Proposition~\ref{prop:intr:ILED} since the term is localized away from the event horizon. The term \eqref{eq:intr:red-shift:EH:4} is estimated slightly different from that in Proposition~\ref{prop:intr:red-shift} since $\lambda$ does not have a favorable sign. Nevertheless, since
$$|\lambda|\leq C(\Omg_{RN}^2+v^{-s}),$$
we have
$$|\eqref{eq:intr:red-shift:EH:4}|\leq C\iint_{\calD(\tau_{1}, \tau_{2})} \chi_{+} \left( 1+\f{\tau_1^{-s}}{\Omg^2_{RN}} \right)(\rd_u\phi)^2 .$$
Therefore, by \eqref{RS.EH.est}, we can choose $A_{+,0}$ to be sufficiently negative such that $|\eqref{eq:intr:red-shift:EH:4}|$ can be dominated by \eqref{eq:intr:red-shift:EH:2}, i.e., we have
\begin{equation}\label{good.term.RS.EH}
\eqref{eq:intr:red-shift:EH:2}+\eqref{eq:intr:red-shift:EH:4}\geq C^{-1}\iint_{\calD(\tau_{1}, \tau_{2})} \chi_{+} \Omg_{RN}^{-2} (\rd_{u} \phi)^{2}.
\end{equation}
Finally, to handle \eqref{eq:intr:red-shift:EH:5}, recall the following estimates from Theorem \ref{main.theorem.C0.stability}:
$$|\nu_{RN}|\leq C\Omg_{RN}^2,\quad |\nu-\nu_{RN}|\leq C\Omg_{RN}^2 v^{-s}.$$
Therefore, by the Cauchy-Schwarz inequality, for every $\varepsilon>0$, we have
\begin{equation*}
\begin{split}
|\eqref{eq:intr:red-shift:EH:5}| \leq \frac{\varepsilon}{2} \iint_{\calD(\tau_{1}, \tau_{2})} \chi_{+} \Omg_{RN}^{-2} (\rd_{u} \phi)^{2} + \frac{C}{2 \varepsilon} \iint_{\calD(\tau_{1}, \tau_{2})} \chi_{+} \Omg_{RN}^2 (\rd_{v} \phi)^{2} .
\end{split}
\end{equation*}
Choosing $\varepsilon$ sufficiently small the first term can be controlled by \eqref{good.term.RS.EH}. On the other hand, since $\Omg_{RN}^2\leq C\log_+^{-\gamma}(\f{1}{\Omg})$, the second term can be bounded using Proposition~\ref{prop:intr:ILED}. After noting that $\Omg$ and $\Omg_{RN}$ are comparable, this concludes the proof of the proposition. \qedhere 
\end{proof}

At this point, we can fix $A_-$ and $A_+$ so that $A_->A_{-,0}$ and $A_+<A_{+,0}$ as in Propositions \ref{prop:intr:red-shift} and \ref{prop:intr:red-shift:EH}. We now drop the subscripts in the constants $C_{A_+,A_-}$, i.e., from now on, we use the convention $C$ also depends on $A_-$ and $A_+$. 

\subsection{Putting everything together}\label{sec.instab.together}

We now state a proposition which combines all the bounds that have been proven so far.
\begin{proposition} \label{prop:intr:red-shift:final}
For $\tau_0$ sufficiently large, for every $\alp\in [0,\alp_0]$ and for every $\tau_{1}, \tau_{2}$ such that $\tau_0 \leq \tau_{1} \leq \tau_{2}$, we have
\begin{equation} \label{eq:intr:red-shift:final}
\begin{aligned}
& \hskip-2em
	\int_{\Gmm_{\tau_{2}}^{(1)}} \bb( 1 + \chi_{-} \log_+^{\alp}(\frac{1}{\Omg}) \bb) (\rd_{v} \phi)^{2} 
	+ \int_{\Gmm_{\tau_{2}}^{(2)}} \Omg^{-2} (\rd_{u} \phi)^{2} \\
	 &+ \int_{\CH(\tau_{1}, \tau_{2})}  (\rd_{u} \phi)^{2} + \int_{\EH(\tau_{1}, \tau_{2})} (\rd_{v} \phi)^{2}  \\
	&	+ \iint_{\calD(\tau_{1}, \tau_{2})}\left(  \left( \log_+^{-\gamma}(\f{1}{\Omg}) + \alp \chi_{-} \log_+^{\alp-1}(\frac{1}{\Omg}) \right) (\rd_{v} \phi)^{2}  
		+	\left( \log_+^{-\gamma}(\f{1}{\Omg}) + \chi_{+} \Omg^{-2} \right) (\rd_{u} \phi)^{2}  \right)\\
	 \leq &C \bb( \int_{\Gmm^{(1)}_{\tau_{1}}} \bb( 1 + \chi_{-} \log_+^{\alp}(\frac{1}{\Omg}) \bb) (\rd_{v} \phi)^{2}
				+ \int_{\Gmm^{(2)}_{\tau_{1}}} \Omg^{-2} (\rd_{u} \phi)^{2} \bb)\\
			&	+C\tau_1^{-s}\iint_{\calD(\tau_{1}, \tau_{2})} \chi_-\log_+^{\gamma}(\f{1}{\Omg})(\rd_v\phi)^2.
\end{aligned}
\end{equation}
\end{proposition}
\begin{proof}
Combining the estimates in Propositions \ref{prop:intr:ILED}, \ref{prop:intr:energy.1}, \ref{prop:intr:red-shift} and \ref{prop:intr:red-shift:EH}, we obtain \eqref{eq:intr:red-shift:final} except that on the right hand side we instead have
\begin{equation*}
\begin{split}
\leq &C \bb( \int_{\Gmm^{(1)}_{\tau_{1}}} \bb( 1 + \chi_{-} \log_+^{\alp}(\frac{1}{\Omg}) \bb) (\rd_{v} \phi)^{2}
				+ \int_{\Gmm^{(2)}_{\tau_{1}}} \Omg^{-2} (\rd_{u} \phi)^{2} \bb)\\
			&	+C\tau_1^{-s}\iint_{\calD(\tau_{1}, \tau_{2})} \chi_+\log_+^{\gamma}(\f{1}{\Omg})(\rd_u\phi)^2+C\tau_1^{-s}\iint_{\calD(\tau_{1}, \tau_{2})} \chi_-\log_+^{\gamma}(\f{1}{\Omg})(\rd_v\phi)^2.
			\end{split}
			\end{equation*}
Finally, by choosing $\tau_0$ to be sufficiently large, we can subtract $C\tau_1^{-s}\iint_{\calD(\tau_{1}, \tau_{2})} \chi_+\log_+^{\gamma}(\f{1}{\Omg})(\rd_u\phi)^2$ from both sides (since $\log_{+}^{\gmm}(\frac{1}{\Omg}) \leq C \Omg^{-2}$ on the support of $\chi_{+}$) and obtain the desired conclusion.			
\end{proof}

Iterating Proposition~\ref{prop:intr:red-shift:final}, we obtain a decay statement for $\rd_{v} \phi$ on $\EH$.
\begin{proposition} \label{prop:intr:decay}
Assume that for every $\tau_0$ sufficiently large, there exists $\tau_0'\geq \tau_0$ such that \eqref{contra.2} holds. Then for every $\tau \geq \tau_0'$ we have
\begin{equation} \label{eq:intr:decay}
\begin{aligned}
& \hskip-2em
	\int_{\EH(\tau, \infty)} (\rd_{v} \phi)^{2} 
	& \leq C_{\tau_0', E_{\tau_0'}} \tau^{- \alp_{0}}.
\end{aligned}
\end{equation}

\end{proposition}

\begin{proof} 
We will prove the following statement for every $\tau\geq \tau_0'$ and for $n \leq \alp_{0}$ by an induction on $n$:
\begin{equation} \label{eq:intr:decay:indHyp}
\begin{aligned}
& \hskip-2em
	\tau^{n} \int_{\EH(\tau, \infty)} (\rd_{v} \phi)^{2} 
	+ \sum_{j=0}^{n} (\alp_{0} - j) \tau^{j} \iint_{\calD(\tau, \infty)} \chi_{-} \log_+^{\alp_{0}-j-1} (\frac{1}{\Omg}) (\rd_{v} \phi)^{2} \\
	& + \tau^{n} \iint_{\calD(\tau, \infty)} \chi_{+} \Omg^{-2} (\rd_{u} \phi)^{2} 
	+ \tau^{n} \iint_{\calD(\tau, \infty)} \log_+^{-\gamma}(\f{1}{\Omg}) \bb( (\rd_{v} \phi)^{2} + (\rd_{u} \phi)^{2} \bb)
	\leq \calI_{n},
\end{aligned}
\end{equation}
where $\calI_{n}$ is a positive constant depending on $E_{\tau_0'}$, $n$, $\tau_0'$, $M$, ${\bf e}$, $E$, $s$, $\alp_0$ and is independent of $\tau$.

We begin with the $n = 0$ case. By Proposition~\ref{prop:intr:red-shift:final} with $\alp=\alp_0$, and using the contradiction assumption \eqref{contra.2}, we get
\begin{equation}\label{intr:induction:n=0}
\begin{aligned}
& \hskip-2em
	\int_{\EH(\tau, \infty)} (\rd_{v} \phi)^{2} 
	+ \alp_{0}  \iint_{\calD(\tau, \infty)} \chi_{-} \log_+^{\alp_{0}-1} (\frac{1}{\Omg}) (\rd_{v} \phi)^{2} \\
	& +  \iint_{\calD(\tau, \infty)} \chi_{+} \Omg^{-2} (\rd_{u} \phi)^{2} 
	+  \iint_{\calD(\tau, \infty)} \log_+^{-\gamma}(\f{1}{\Omg}) \bb( (\rd_{v} \phi)^{2} + (\rd_{u} \phi)^{2} \bb)\\
	\leq &C(1+ E_{\tau_0'})+C\tau^{-s}\iint_{\calD(\tau, \infty)} \chi_-\log_+^{\gamma}(\f{1}{\Omg})(\rd_v\phi)^2.
\end{aligned}
\end{equation}
Notice that we have used \eqref{contra.2} as well as the estimates in Theorem \ref{main.theorem.C0.stability} to show that the ``data terms'' on $\Gmm_{\tau_0'}$ are bounded by $C(1+E_{\tau_0'})$. Recall now that $\gamma\leq 2\leq \alp_0-1$ and therefore after choosing $\tau_0$ to be sufficiently large, we can subtract $C\tau^{-s}\iint_{\calD(\tau, \infty)} \chi_-\log_+^{\gamma}(\f{1}{\Omg})(\rd_v\phi)^2$ from both sides of \eqref{intr:induction:n=0} to obtain
\begin{equation}\label{intr:induction:n=0.1}
\begin{aligned}
& \hskip-2em
	\int_{\EH(\tau, \infty)} (\rd_{v} \phi)^{2} 
	+ \alp_{0}  \iint_{\calD(\tau, \infty)} \chi_{-} \log_+^{\alp_{0}-1} (\frac{1}{\Omg}) (\rd_{v} \phi)^{2} \\
	& +  \iint_{\calD(\tau, \infty)} \chi_{+} \Omg^{-2} (\rd_{u} \phi)^{2} 
	+  \iint_{\calD(\tau, \infty)} \log_+^{-\gamma}(\f{1}{\Omg}) \bb( (\rd_{v} \phi)^{2} + (\rd_{u} \phi)^{2} \bb)
	\leq C(1+ E_{\tau_0'}),
\end{aligned}
\end{equation}
which is the desired conclusion for $n=0$.

Assume, for the purpose of induction, that \eqref{eq:intr:decay:indHyp} holds for $n = 0, 1, \ldots, n_{0}-1$, where $n_{0}$ is an integer such that $1 \leq n_{0} \leq \alp_{0}$. Then for every $k \in \bbN\cap\{2^k\geq \tau_0'\}$, by the pigeonhole principle, there exists $\tau_{(k)} \in [2^{k}, 2^{k+1}]$ such that
\begin{equation} \label{eq:intr:decay:pigeonhole}
\begin{aligned}
& \hskip-2em
	\int_{\Gmm^{(1)}_{\tau_{(k)}}} \bb( \big( \log_+^{-\gamma}(\f{1}{\Omg}) + (\alp_{0} - n_{0} + 1) \chi_{-} \log_+^{\alp_{0} - n_{0}} (\frac{1}{\Omg}) \big) (\rd_{v} \phi)^{2} \bb) \\
	& + \int_{\Gmm^{(2)}_{\tau_{(k)}}} \bb( (\log_+^{-\gamma}(\f{1}{\Omg}) + \chi_{+} \Omg^{-2}) (\rd_{u} \phi)^{2}\bb)
	\leq C \calI_{n_{0}-1} \tau_{(k)}^{-n_{0}},
\end{aligned}
\end{equation}
for some $C > 0$. Observe that the first two terms on the right-hand side of \eqref{eq:intr:red-shift:final} for $\alp = \alp_{0} - n_{0}$ and $\tau_{1} = \tau_{(k)}$ is bounded by a constant multiple of the left-hand side of \eqref{eq:intr:decay:pigeonhole}, where the constant may depend on $n_{0}$ but is independent of $\tau_{(k)}$.
By appealing to Proposition~\ref{prop:intr:red-shift:final}, we obtain that for every $\tau\in [\tau_{(k)},\tau_{(k+1)})$:
\begin{equation} \label{intr:induction:generaln}
\begin{aligned}
& \hskip-2em
	 \int_{\EH(\tau_{(k)}, \tau)} (\rd_{v} \phi)^{2} 
	+  (\alp_0-n_0)\iint_{\calD(\tau_{(k)}, \tau)} \chi_{-} \log_+^{\alp_{0}-n_0-1} (\frac{1}{\Omg}) (\rd_{v} \phi)^{2} \\
	& +  \iint_{\calD(\tau_{(k)}, \tau)} \chi_{+} \Omg^{-2} (\rd_{u} \phi)^{2} 
	+  \iint_{\calD(\tau_{(k)}, \tau)} \log_+^{-\gamma}(\f{1}{\Omg}) \bb( (\rd_{v} \phi)^{2} + (\rd_{u} \phi)^{2} \bb)\\
	\leq &C\calI_{n_0-1}\tau_{(k)}^{-n_0}+C\tau_{(k)}^{-s}\iint_{\calD(\tau_{(k)}, \tau)} \chi_-\log_+^{\gamma}(\f{1}{\Omg})(\rd_v\phi)^2.
\end{aligned}
\end{equation}
We now separate the argument into two cases: either $\alp_0-n_0\geq 2$ or $\alp_0-n_0\leq 1$. In the first case, since $\gamma\leq 2\leq \alp_0-n_0$, we can apply the induction hypothesis for $n=n_0-1$, which gives
\begin{equation*}
\begin{split}
\tau_{(k)}^{-s}\iint_{\calD(\tau_{(k)}, \tau)} \chi_-\log_+^{\gamma}(\f{1}{\Omg})(\rd_v\phi)^2\leq &C\tau_{(k)}^{-s}\iint_{\calD(\tau_{(k)}, \tau)} \chi_{-} \log_+^{\alp_{0}-n_0} (\frac{1}{\Omg}) (\rd_{v} \phi)^{2}\leq \calI_{n_0-1}\tau_{(k)}^{-n_0-s+1}.
\end{split}
\end{equation*}
Since $s>2$, combining this with \eqref{intr:induction:generaln} thus gives
\begin{equation} \label{intr:induction:generaln.1}
\begin{aligned}
& \hskip-2em
	 \int_{\EH(\tau_{(k)}, \tau)} (\rd_{v} \phi)^{2} 
	+  (\alp_0-n_0)\iint_{\calD(\tau_{(k)}, \tau)} \chi_{-} \log_+^{\alp_{0}-n_0-1} (\frac{1}{\Omg}) (\rd_{v} \phi)^{2} \\
	& +  \iint_{\calD(\tau_{(k)}, \tau)} \chi_{+} \Omg^{-2} (\rd_{u} \phi)^{2} 
	+  \iint_{\calD(\tau_{(k)}, \tau)} \log_+^{-\gamma}(\f{1}{\Omg}) \bb( (\rd_{v} \phi)^{2} + (\rd_{u} \phi)^{2} \bb)
	\leq C\calI_{n_0-1}\tau_{(k)}^{-n_0}.
\end{aligned}
\end{equation}
In the second case, since $n_0-1\geq \alp_0-2$, we can thus apply the inductive hypothesis for $n=\alp_0-2$ (notice that the assumption of Theorem \ref{final.blow.up.step} in particular ensures that $\alp_0-3\geq 0$). By \eqref{eq:intr:decay:indHyp}, we therefore have 
$$  \tau^{\alp_0-2}\iint_{\calD(\tau, \infty)} \chi_{-} \log (\frac{1}{\Omg}) (\rd_{v} \phi)^{2}+\tau^{\alp_0-3}\iint_{\calD(\tau, \infty)} \chi_{-} \log_+^2 (\frac{1}{\Omg}) (\rd_{v} \phi)^{2}\leq \mathcal I_{\alp_0-2}$$
for all $\tau\geq \tau_0'$. By the Cauchy-Schwarz inequality, since $\gamma\in (1,2]$, this then gives
$$ \iint_{\calD(\tau, \infty)} \chi_{-} \log_+^{\gamma} (\frac{1}{\Omg}) (\rd_{v} \phi)^{2}\leq C\mathcal I_{\alp_0-2}\tau^{-\alp_0+1+\gamma}$$
for all $\tau\geq \tau_0'$. We can use this to control the last term in \eqref{intr:induction:generaln} to get
\begin{equation} \label{intr:induction:generaln.2}
\begin{aligned}
& \hskip-2em
	 \int_{\EH(\tau_{(k)}, \tau)} (\rd_{v} \phi)^{2} 
	+  (\alp_0-n_0)\iint_{\calD(\tau_{(k)}, \tau)} \chi_{-} \log_+^{\alp_{0}-n_0-1} (\frac{1}{\Omg}) (\rd_{v} \phi)^{2} \\
	& +  \iint_{\calD(\tau_{(k)}, \tau)} \chi_{+} \Omg^{-2} (\rd_{u} \phi)^{2} 
	+  \iint_{\calD(\tau_{(k)}, \tau)} \log_+^{-\gamma}(\f{1}{\Omg}) \bb( (\rd_{v} \phi)^{2} + (\rd_{u} \phi)^{2} \bb)\\
	\leq &C\calI_{n_0-1}\tau_{(k)}^{-n_0}+C\mathcal I_{\alp_0-2}\tau_{(k)}^{-\alp_0+1+\gamma-s}\leq C(\calI_{n_0-1}+\calI_{\alp_0-2})\tau_{(k)}^{-n_0},
\end{aligned}
\end{equation}
where in the last estimate we have used $\gamma\leq s-1$ and $n_0\leq \alp_0$.

Therefore, in both cases, by \eqref{intr:induction:generaln.1} and \eqref{intr:induction:generaln.2} (and using the induction hypothesis for the $j\leq n_0-1$ term in the sum below), we conclude that 
\begin{equation} \label{eq:intr:decay:conclusion}
\begin{aligned}
&
	2^{kn_0} \int_{\EH(\tau_{(k)}, 4\tau_{(k)})} (\rd_{v} \phi)^{2} 
	+ \sum_{j=0}^{n_0} (\alp_{0} - j) 2^{kj} \iint_{\calD(\tau_{(k)}, 4\tau_{(k)})} \chi_{-} \log_+^{\alp_{0}-j-1} (\frac{1}{\Omg}) (\rd_{v} \phi)^{2} \\
	& + 2^{kn_0} \iint_{\calD(\tau_{(k)}, 4\tau_{(k)})} \chi_{+} \Omg^{-2} (\rd_{u} \phi)^{2} 
	+ 2^{kn_0} \iint_{\calD(\tau_{(k)}, 4\tau_{(k)})} \log_+^{-\gamma}(\f{1}{\Omg}) \bb( (\rd_{v} \phi)^{2} + (\rd_{u} \phi)^{2} \bb)
	\leq \calI_{n_0}.
\end{aligned}
\end{equation}
Note, in particular, that $[2^{k+1}, 2^{k+2}] \subseteq [\tau_{(k)}, 4 \tau_{(k)}]$. Hence, for any $\tau \geq 2^{\lceil \log_2\tau_0' \rceil+1}$, we can sum \eqref{eq:intr:decay:conclusion} for $k\geq \lfloor \log_2\tau\rfloor -1$ to obtain \eqref{eq:intr:decay:indHyp} for $n = n_{0}$. Finally, for $\tau\in [\tau_0', 2^{\lceil \log_2\tau_0' \rceil+1})$, the desired estimate \eqref{eq:intr:decay:indHyp} for $n=n_0$ follows from \eqref{intr:induction:n=0.1} (recall that the implicit constant is allowed to depend on $\tau_0'$). The concludes the induction and proves \eqref{eq:intr:decay:indHyp}. The conclusion of the proposition follows as an immediate consequence.
\qedhere
\end{proof}

We now conclude the proof of \eqref{blow.up.interior} in  Theorem \ref{final.blow.up.step}:
\begin{proof}[Proof of \eqref{blow.up.interior} in Theorem \ref{final.blow.up.step}]
Using Proposition \ref{prop:intr:decay}, we in particular have 
$$\int_{\mathcal H^+(\tau,2\tau)} \tau^{\alp_0} (\rd_v\phi)^2 \leq C$$
for $\tau\geq \tau_0'$ (for $C$ depending in particular on $\tau_0'$ but independent of $\tau$.
We apply this estimate for a sequence $\tau_k=2^k$ to get
\begin{equation*}
\begin{split}
\int_{\EH\cap\{v\geq \tau_0'\}} \f{v^{\alp_0}}{\log^2(1+v)}(\rd_v\phi)^2\leq &C\sum_{k=0}^\infty \int_{\mathcal H^+(\tau_k,\tau_{k+1})} \f{\tau^{\alp_0}}{(k+1)^2}(\rd_v\phi)^2\leq C\sum_{k=0}^{\infty}\f{1}{(k+1)^2}<\infty.
\end{split}
\end{equation*}
We have thus achieved \eqref{contra.2.goal} and conclude the proof of \eqref{blow.up.interior} in Theorem \ref{final.blow.up.step}.
\qedhere
\end{proof}

\subsection{Blow up of $\f{\lambda}{\Omg^2}$}\label{sec.blow.up.dvr}

\begin{proof}[Proof of \eqref{blow.up.interior.dvr} in Theorem \ref{final.blow.up.step}]
This is proven using \eqref{eqn.Ray}:
\begin{equation*}
\rd_v(\f{\lambda}{\Omg^2})=-\f{r}{\Omg^2}(\rd_v\phi)^2.
\end{equation*}
By Theorem \ref{main.theorem.C0.stability}, $|\rd_v\log\Omg-\rd_v\log\Omg_{RN}|+|\lambda-\lambda_{RN}|(u,v) \leq C v^{-s}$, $|r-r_{RN}| \leq C \max\{v^{-s}, |u|^{-s+1} \}$ for all $u<u_s$ and therefore there exists $A_\lambda\in\mathbb R$ sufficiently large and $u_\lambda<u_s$ such that 
\begin{enumerate}
\item 
$\f{\lambda}{\Omg^{2}}(u,-u+A_\lambda)<0$ for all $u<u_\lambda$ ,
\item 
$r(u,v)>r_0>0$ for every $(u,v)\in \{(u,v): u<u_\lambda, \, v\geq -u+A_\lambda \}$ for some $r_0$ .
\end{enumerate}
Therefore, by \eqref{R.eqn}, we can integrate in the $v$-direction starting from $v=-u+A_\lambda$ to get
\begin{equation*}
\begin{split}
\left|\f{\lambda}{\Omg^2}\right|(u,v)\gtrsim &\, \int_{-u+A_\lambda}^v (\f{r(\rd_v\phi)^2}{\Omg^2})(u,v')\, dv'\\
\gtrsim &\,r_0\int_{-u+A_\lambda}^v e^{2\kappa_-(v'+u)}(\rd_v\phi)^2(u,v')\, dv'
\gtrsim  \,\int_{-u+A_\lambda}^v e^{2\kappa_- v'}(\rd_v\phi)^2(u,v')\, dv',
\end{split}
\end{equation*}
for every $(u,v)\in \{(u,v): u<u_\lambda, \, v+u\geq A_\lambda\}$, with an implicit constant depending on $r_0$ and $u$. 

By Theorem \ref{final.blow.up.step} (and the fact that for every fixed $u$, $\Omg\sim e^{-\kappa_- v}$ for $v$ sufficient large with a constant depending on $u$), the right hand side $\to \infty$ as $v\to \infty$. This implies the desired conclusion.

\end{proof}

\section{Stability and blow up on the entire Cauchy horizon: Proof of Theorem~\ref{thm.nonpert}} \label{sec.nonpert}

In this section, we prove Theorem~\ref{thm.nonpert}. Unlike Theorems~\ref{main.theorem.C0.stability} and \ref{final.blow.up.step}, we need to work in the \emph{nonperturbative} region, i.e., the spacetime is not necessarily close to a Reissner--Nordstr\"om in any quantitative sense. Special features of Einstein--Maxwell--(real)--scalar--field system in spherical symmetry therefore play an important role in the proof. A key step is propagation of $L^{1}$-type bounds on the metric coefficients and $\phi$, which hold initially thanks to Theorem~\ref{main.theorem.C0.stability}, in characteristic rectangles with a lower bound on $r$ (Lemma~\ref{lem.nonpert-L1}). This strengthens the estimates for the scalar field in \cite[Section~13]{D2}.

We first introduce the coordinates $(u, V)$, which will be fixed for the remainder of this section.
Fix an incoming null curve $\uC_{1}$ (maximally extended) whose past endpoint intersects $\EH_{1}$. We define the coordinates $(U, v)$ by the gauge condition \eqref{gauge.1} on $\EH_{1}$ and \eqref{gauge.2} (i.e., $\rd_{U} r = -1$) on the whole $\uC_{1}$, so that $\uC_{1} = \set{(U, v) : v = 1}$. In particular, if we use the same initial incoming curve $\uC_{1}$, then this coordinate system extends that of Theorem~\ref{main.theorem.C0.stability} past the perturbative region $U \leq U_{s}$. As before, we define $(u, V)$ from $(U, v)$ by \eqref{U.def} and \eqref{V.def}, respectively. 

We record a basic observation regarding the coordinates $(u, V)$ on $\CH_{1}$.
\begin{lemma} \label{lem.u-reg}
After extending $V$ continuously to $\calQ^{+}$ with respect to the topology of $\bbR^{1+1}$, the Cauchy horizon $\CH_{1}$ coincides with the curve $\set{V = 1}$. Moreover, $u$ is finite and nondegenerate (i.e., $d u \neq 0$) on $\CH_{1}$ minus (possibly) the future endpoint. In particular, $u_{\CH_{1}} \in (-\infty, \infty]$ in the statement of Theorem~\ref{thm.nonpert} is well-defined.
\end{lemma}
 
\begin{proof}
Fix an outgoing curve $C_{\ast}$ in $\calQ^{+}$ which intersects $\CH_{1}$ minus the future endpoint. Since $\calQ$ is globally hyperbolic, $C_{\ast}$ intersects the initial hypersurface $\EH_{1} \cup \Sgm \cup \EH_{2}$, from which it follows that $C_{\ast}$ intersects $\uC_{1}$ in $\calQ$. Since the $u$ coordinate is constant on $C_{\ast}$, it now suffices to verify that $u$ is finite and nondegenerate on every point in $\uC_{1} \subset \calQ$.

By the condition \eqref{gauge.2}, the Raychaudhuri equation (which ensures that $r$ decreases along $\uC_{1}$ in the incoming direction) and the fact that $r \geq 0$ on $\calQ$, the function $U$ is finite and nondegenerate on every point on $\uC_{1} \cap \calQ$. Since the change of variables \eqref{U.def} is nondegenerate and keeps $u(U)$ finite as long as $U$ is, the same statement holds for $u$ as desired. 
\end{proof}

Consider a characteristic rectangle $\calR = \set{(u, V) : u_{1} \leq u \leq u_{2}, \ V_{1} \leq V < 1}$, where $u_{1}, u_{2}, V_{1}$ are any numbers such that $-\infty < u_{1} < u_{2} < u_{\CH_{1}}$ and $V_{1} > V(1)$ (see Figure~\ref{fig:nonpert}).

\begin{figure}[h]
\begin{center}
\def\svgwidth{180px}
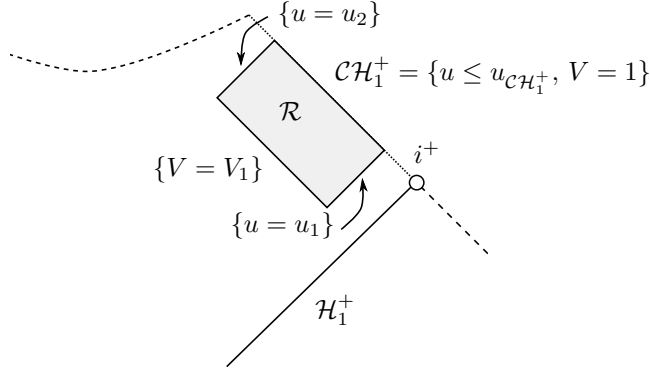 
\caption{Characteristic rectangle $\calR$} \label{fig:nonpert}
\end{center}
\end{figure}

In the following lemmas, we will assume that the following bounds on $r$ holds in $\calR$:
\begin{equation} \label{eq:nonpert-r}
	0 < r_{0} \leq r(u, V) \leq R < \abs{\bfe} \quad \hbox{ for all } (u, V) \in \calR.
\end{equation}
The key restrictive assumption is the \emph{lower bound} $r(u, V) \geq r_{0} > 0$. In fact, the upper bound in \eqref{eq:nonpert-r} turns out to be a simple consequence of the fact that $\calR$ sits in the interior of a subextremal black hole; see the proof of Theorem~\ref{thm.nonpert} below. 

Under the above assumptions, we first show that the spacetime volume of $\calR$ is finite. Its proof requires the use of \eqref{eq:nonpert-r} and the precise structure of the spherically symmetric Einstein--Maxwell--(real)--scalar--field system (in particular, the strict inequality $R < \abs{\bfe}$ in \eqref{eq:nonpert-r} is crucial).

\begin{lemma} \label{lem.nonpert-vol}
Let $\calR = \set{(u, V) : u_{1} \leq u \leq u_{2}, \ V_{1} \leq V < 1} \subset \calQ$ be a characteristic rectangle such that \eqref{eq:nonpert-r} holds. Then the spacetime volume of $\calR$ is finite, i.e.,
\begin{equation}\label{finite.volume}
\int_{u_1}^{u_2}\int_{V_1}^{1} \Omg^2(u,V)\, d V\, du\leq \f {8 R^2}{|1-\f{{\bf e}^2}{R^2}|}<\infty,
\end{equation}
where $R$ as in \eqref{eq:nonpert-r}.
\end{lemma}
\begin{proof}
We apply the \eqref{WW.SS} equations in the $(u, V)$ coordinate. By the $\rd_u\rd_{V} r$ equation, 
\begin{equation}\label{r2.uv}
\f 1{2}\rd_u\rd_{V} r^2=\rd_u(r\rd_{V} r) =r\rd_u\rd_{V} r+\rd_u r\rd_{V} r=-\f {\Omg^2}{4 }(1-\f{{\bf e}^2}{r^2}).
\end{equation}
By \eqref{eq:nonpert-r}, we have $1-\f{{\bf e}^2}{r^2}\leq 1-\f{{\bf e}^2}{R^2}<0$. Therefore, 
$$\int_{u_1}^{u_2}\int_{V_1}^{V_2} \Omg^2(u,V)\, d V\, du\leq \f 2{|1-\f{{\bf e}^2}{R^2}|} \int_{u_1}^{u_2}\int_{V_1}^{V_2} (\rd_u\rd_{V}r^2)(u,V)\, dV\, du\leq \f {8 R^2}{|1-\f{{\bf e}^2}{R^2}|},$$
where the last step simply follows from $r^2(u_1, V_1)+r^2(u_1, V_2)+r^2(u_2, V_1)+r^2(u_2, V_2)\leq 4R^2$.
\end{proof}

Next, we prove the key $L^{1}$ bounds on $\Omg, r$ and $\phi$. The idea is to divide $\calR$ into sub-rectangles to gain a smallness parameter, which is possible thanks to the finiteness of the spacetime volume of $\calR$.

\begin{lemma} \label{lem.nonpert-L1}
Let $\calR = \set{(u, V) : u_{1} \leq u \leq u_{2}, \ V_{1} \leq V < 1} \subset \calQ$ be a characteristic rectangle such that \eqref{eq:nonpert-r} holds. Assume furthermore that 
\begin{equation} \label{eq:nonpert-L1-ini}
\begin{aligned}
	& \int_{u_{1}}^{u_{2}} r \abs{\rd_{u} r}(u, V_{1}) d u
	+ \int_{u_{1}}^{u_{2}} r \abs{\rd_{u} \phi}(u, V_{1}) d u
	+ \int_{u_{1}}^{u_{2}} \abs{\rd_{u} \log \Omg}(u, V_{1}) d u \\
	& + \int_{V_{1}}^{1} r \abs{\rd_{V} r}(u_{1}, V) d V
	+ \int_{V_{1}}^{1} r \abs{\rd_{V} \phi}(u_{1}, V) d V	
	+ \int_{V_{1}}^{1} \abs{\rd_{V} \log \Omg}(u_{1}, V) d V \leq D_{\calR} 
\end{aligned}
\end{equation}
for some $0 < D_{\calR} < \infty$. Then the following estimates hold:
\begin{align}
	\int_{u_{1}}^{u_{2}} \sup_{V \in [V_{1}, 1)} r \abs{\rd_{u} r}(u, V) d u
	+ \int_{V_{1}}^{1} \sup_{u \in [u_{1}, u_{2}]} r \abs{\rd_{V} r}(u, V) d V 
	\leq & C_{r_{0}, R, \bfe, D_{\calR}}, 			\label{eq:nonpert-L1-r} \\
	\int_{u_{1}}^{u_{2}} \sup_{V \in [V_{1}, 1)} r \abs{\rd_{u} \phi}(u, V) d u
	+ \int_{V_{1}}^{1} \sup_{u \in [u_{1}, u_{2}]} r \abs{\rd_{V} \phi}(u, V) d V 
	\leq & C_{r_{0}, R, \bfe, D_{\calR}}, 			\label{eq:nonpert-L1-phi} \\
	\int_{u_{1}}^{u_{2}} \sup_{V \in [V_{1}, 1)} \abs{\rd_{u} \log \Omg}(u, V) d u
	+ \int_{V_{1}}^{1} \sup_{u \in [u_{1}, u_{2}]} \abs{\rd_{V} \log \Omg}(u, V) d V 
	\leq & C_{r_{0}, R, \bfe, D_{\calR}}.			\label{eq:nonpert-L1-omg}
\end{align}
\end{lemma}

\begin{proof}
We proceed in three steps, obtaining bounds for $r$, $\phi$ and $\log \Omg$ in order.

\pfstep{Step~1: $L^1$ estimates for $\rd_u r$ and $\rd_{V}r$}
Our goal is to show that there exist partitions $u_1=u^{(0)}<u^{(1)}<...<u^{(m)}=u_2$ and $V_1= V^{(0)}< V^{(1)}<\dots< V^{(n)} = 1$ for some $m,n\in\mathbb N$ such that\footnote{Let us note that the main point in \eqref{r.diff.part} is the supremums in the estimates. The same bounds without the supremums can be easily obtained using \eqref{eq:nonpert-r}, $\rd_u r<0$ and the fact that $\rd_v r$ changes sign at most once along a constant $u$-curve. (The last fact follows from \eqref{eqn.Ray}.)}
\begin{equation}\label{r.diff.part}
\max_{0\leq i\leq m-1}\int_{u^{(i)}}^{u^{(i+1)}} \sup_{ V\in [ V_1, V_2]}\left|\f{\rd_u r}{r}\right|(u, V)\, du+\max_{0\leq j\leq n-1}\int_{V^{(j)}}^{V^{(j+1)}}\sup_{u\in [u_1,u_2]}\left|\f{\rd_{V} r}{r}\right|(u,V)\,d V\leq \f 12.
\end{equation}
We will find the partition in $V$ and estimate the second term in \eqref{r.diff.part}. The other term can be controlled in a completely analogous manner. Note that \eqref{eq:nonpert-L1-r} follows from \eqref{r.diff.part} by summing up in $i, j$.

Using the bound for $r |\rd_{V} r|(u_1, V)$ in \eqref{eq:nonpert-L1-ini}, the bounds for $r$ in \eqref{eq:nonpert-r} and the estimate \eqref{finite.volume} for the spacetime volume, for every $\delta>0$, we can choose a partition $V_1=V^{(0)}< V^{(1)}<\dots< V^{(n)}=1$ such that 
\begin{equation}\label{smallness}
\max_{0\leq j\leq n-1}\left(\int_{V^{(j)}}^{V^{(j+1)}}\left|r\rd_{V} r\right|(u_1, V)\,d V+\int_{u_1}^{u_2}\int_{V^{(j)}}^{V^{(j+1)}} \Omg^2(u,V)\, d V\, du\right)\leq \delta.
\end{equation}
Integrating \eqref{r2.uv} in $u$, we thus obtain using \eqref{smallness} and \eqref{eq:nonpert-r} that
\begin{equation}
\begin{split}
&\max_{0\leq j\leq n-1}\int_{V^{(j)}}^{V^{(j+1)}}\sup_{u \in [u_1,u_2]} r \left|\rd_{V} r \right|(u, V)\,d V\\
\leq &\delta +\max_{0\leq j\leq n-1}\left(\sup_{\substack{u \in [u_1,u_2] \\ V \in [ V^{(j)}, V^{(j+1)}]}} \f 1{4 }\left|1-\f{{\bf e}^2}{r^2}\right|(u, v)\right)\left(\int_{u_1}^{u_2}\int_{V^{(j)}}^{V^{(j+1)}} \Omg^2(u, V)\, d V\, d u\right)\leq C_{r_0,\bfe} \delta,
\end{split}
\end{equation}
for some constant $C_{r_0,\bfe}>0$ depending on $r_0$ and ${\bf e}$. Using \eqref{eq:nonpert-r} again and choosing $\delta$ sufficiently small depending on $r_0$ and ${\bf e}$, we thus obtain
$$\max_{0\leq j\leq n-1}\int_{V^{(j)}}^{V^{(j+1)}}\sup_{u\in [u_1,u_2]}\left|\f{\rd_{V} r}{r}\right|(u, V)\,d V\leq \f 14.$$
A similar bound for the first term in \eqref{r.diff.part} can be obtained in a completely identical manner. 

\pfstep{Step 2: $L^{1}$ estimates for $\rd_u\phi$ and $\rd_{V}\phi$} Our goal is to show \eqref{eq:nonpert-L1-phi}.
Note that $\int_{V_{1}}^{1} r|\rd_{V}\phi| (u_{1}, V) d V$ and $\int_{u_{1}}^{u_{2}} r|\rd_{u}\phi| (u, V_{1}) d u$ are bounded by \eqref{eq:nonpert-L1-ini}. We will propagate these bounds in the direction where $V$ and $u$ are both increasing.

Consider $\mathcal R_{ij}:=\{(u, V): V^{(j)}\leq V \leq V^{(j+1)}, \, u^{(i)}\leq u\leq u^{(i+1)}\}$ for $0\leq i\leq m-1$, $0\leq j\leq n-1$, where $V^{(j)}$, $u^{(i)}$, etc. are as in the previous step. Note that by definition $\cup_{i=0}^{m-1}\cup_{j=0}^{n-1} \mathcal R_{ij}=\mathcal R$. By the initial $L^{1}$ boundedness of $|\rd_{V}\phi|$ and $|\rd_u\phi|$ mentioned above, it therefore suffices to show that for every $0\leq i\leq m-1$, $0\leq j\leq n-1$,
\begin{equation}\label{phi.Li.bounds.goal}
\begin{aligned}
\int_{u^{(i)}}^{u^{(i+1)}} & \sup_{V \in [V^{(j)}, V^{(j+1)})} r \abs{\rd_{u} \phi}(u, V) \, d u 
+ \int_{V^{(j)}}^{V^{(j+1)}} \sup_{u \in [u^{(i)}, u^{(i+1)}]} r \abs{\rd_{V} \phi}(u, V) \, d V \\
\leq & 2  \left( \int_{u^{(i)}}^{u^{(i+1)}} r |\rd_{u}\phi|(v,V^{(j)})+\int_{V^{(j)}}^{V^{(j+1)}} r |\rd_{V}\phi|(u^{(i)},V)\right),
\end{aligned}
\end{equation}
i.e.,  in every $\mathcal R_{ij}$, the $L^{1} L^{\infty}$ norm of $r |\rd_u\phi|$ and $r |\rd_{V}\phi|$ are at most twice their initial values on the lower left and right sides. Indeed, since there are finitely many $\calR_{ij}$'s, iterating \eqref{phi.Li.bounds.goal} gives the desired estimate \eqref{eq:nonpert-L1-phi}.

In order to prove \eqref{phi.Li.bounds.goal}, we rewrite \eqref{WW.SS} as
$$\rd_u(r\rd_{V} \phi)=-(\rd_{V} r)(\rd_u\phi),\quad \rd_{V}(r\rd_u \phi)=-(\rd_u r)(\rd_{V}\phi).$$
Integrating these equations in the $-\rd_u$ and $\rd_{V}$ directions respectively, we obtain
\begin{align*}
& \int_{u^{(i)}}^{u^{(i+1)}}  \sup_{V \in [V^{(j)}, V^{(j+1)})} r \abs{\rd_{u} \phi}(u, V) \, d u 
+ \int_{V^{(j)}}^{V^{(j+1)}} \sup_{u \in [u^{(i)}, u^{(i+1)}]} r \abs{\rd_{V} \phi}(u, V) \, d V  \\
\leq & \left( \int_{u^{(i)}}^{u^{(i+1)}} r |\rd_{u}\phi|(v,V^{(j)})+\int_{V^{(j)}}^{V^{(j+1)}} r |\rd_{V}\phi|(u^{(i)},V)\right) \\
& + \left( \max_{0\leq i\leq m-1}\int_{u^{(i)}}^{u^{(i+1)}} \sup_{ V\in [ V_1, V_2]}\left|\f{\rd_u r}{r}\right|(u, V)\, du+\max_{0\leq j\leq n-1}\int_{V^{(j)}}^{V^{(j+1)}}\sup_{u\in [u_1,u_2]}\left|\f{\rd_{V} r}{r}\right|(u,V)\,d V 
 \right) \\
 & \phantom{+} \times \left( \int_{u^{(i)}}^{u^{(i+1)}}  \sup_{V \in [V^{(j)}, V^{(j+1)})} r \abs{\rd_{u} \phi}(u, V) \, d u 
+ \int_{V^{(j)}}^{V^{(j+1)}} \sup_{u \in [u^{(i)}, u^{(i+1)}]} r \abs{\rd_{V} \phi}(u, V) \, d V \right) \\
\leq & \left( \int_{u^{(i)}}^{u^{(i+1)}} r |\rd_{u}\phi|(v,V^{(j)})+\int_{V^{(j)}}^{V^{(j+1)}} r |\rd_{V}\phi|(u^{(i)},V)\right) \\
& + \frac{1}{2} \left( \int_{u^{(i)}}^{u^{(i+1)}}  \sup_{V \in [V^{(j)}, V^{(j+1)})} r \abs{\rd_{u} \phi}(u, V) \, d u 
+ \int_{V^{(j)}}^{V^{(j+1)}} \sup_{u \in [u^{(i)}, u^{(i+1)}]} r \abs{\rd_{V} \phi}(u, V) \, d V \right).
\end{align*}
where the last line is achieved using \eqref{r.diff.part}. Rearranging this estimate then gives \eqref{phi.Li.bounds.goal}.

\pfstep{Step~3: $L^1$ estimates for $\rd_u\log \Omg$ and $\rd_{V} \log \Omg$} Our goal is to show that
\begin{equation}\label{dOmg.L1}
\int_{u_1}^{u_2} \sup_{V\in [V_1, 1)}\left|\rd_u \log \Omg \right|(u, V)\, du+\int_{V_{1}}^{1} \sup_{u\in [u_1,u_2]}\left|\rd_{V} \log \Omg \right|(u, V)\,d V<\infty.
\end{equation}
Note that $\int_{V_{1}}^{1} |\rd_{V} \log \Omg| (u_{1}, V) d V$ and $\int_{u_{1}}^{u_{2}} |\rd_{u} \log \Omg| (u, V_{1}) d u$ are bounded by \eqref{eq:nonpert-L1-ini}.
Using the equation for $\rd_u\rd_{V} \log\Omg$ in \eqref{WW.SS}, in order to prove \eqref{dOmg.L1}, it therefore suffices to bound
\begin{equation*}
\begin{split}
\int_{V_1}^{1}\int_{u_1}^{u_2}\left(|\rd_u\phi\rd_{V}\phi|+\f { \Omg^2 {\bf e}^2}{2 r^4}+\f{\Omg^2}{4r^2}+\f{|\rd_u r\rd_{V} r|}{r^2}\right)(u, V)\, du\, d V=:\iint (I+II+III+IV).
\end{split}
\end{equation*}
The term $I$ is bounded using \eqref{eq:nonpert-L1-phi} and H\"older's inequality. The terms $II$ and $III$ are bounded thanks to \eqref{eq:nonpert-r} and \eqref{finite.volume}. Finally, the term $IV$ is bounded using the estimate \eqref{eq:nonpert-L1-r} and H\"older's inequality.
\end{proof}

As a consequence of the $L^{1}$ bounds, $C^{0}$ extendibility on the entire $\CH_{1}$ (excluding the endpoint at which $r = 0$) can be established.
\begin{lemma} \label{lem.nonpert-C0}
Let $\calR = \set{(u, V) : u_{1} \leq u \leq u_{2}, \ V_{1} \leq V < 1} \subset \calQ$ be a characteristic rectangle such that \eqref{eq:nonpert-r} and \eqref{eq:nonpert-L1-ini} hold. Then one can attach the boundary\footnote{Here $\overline{\calR}$ refers to the closure of $\calR$ in the topology induced by the conformal embedding $\calQ \hookrightarrow \bbR^{1+1}$ described in Theorem~\ref{thm:kommemi}.} $\CH_{1} \cap \overline{\calR} = \set{(u, V) : u_{1} \leq u \leq u_{2}, \, V = 1}$ to $\calR$, to which $r$, $\phi$ and $\log \Omg$ extend continuously. 
\end{lemma}
\begin{proof}
We only consider the case of $\log \Omg$; the conclusion for $r$ and $\phi$ follows analogously from the bounds in Lemma~\ref{lem.nonpert-L1}. As in Proposition~\ref{C0.1}, it suffices to show that given any sequence $u_{(i)} \to u$ and $V_{(i)} \to 1$, $\log \Omg(u_{(i)}, V_{(i)})$ is a Cauchy sequence. We have
\begin{align*}
	& \abs{\log \Omg (u_{(i)}, V_{(i)}) - \log \Omg (u_{(j)}, V_{(j)})} \\
	\leq & \abs{\int_{u_{(i)}}^{u_{(j)}} \rd_{u} \log \Omg(u', V_{(i)}) \, d u' }
		+ \abs{\int_{V_{(i)}}^{V_{(j)}} \rd_{V} \log \Omg(u_{(j)}, V') \, d V'} \\
	\leq & \int_{u_{(i)}}^{u_{(j)}} \sup_{V' \in [V_{1}, 1)} \abs{\rd_{u} \log \Omg}(u', V') \, d u' 
		+ \int_{V_{(i)}}^{V_{(j)}} \sup_{u' \in [u_{1}, u_{2}]} \abs{\rd_{V} \log \Omg}(u', V') \, d V'
\end{align*}
where the last line goes to zero as $i, j \to \infty$ thanks to \eqref{eq:nonpert-L1-omg} and the fact that $u_{(j)} - u_{(i)}, V_{(j)} - V_{(i)} \to 0$. \qedhere
\end{proof}

Next, we show that the blow up of $\rd_{V} r$ and $\Omg^{-2} \rd_{V} \phi$ propagate along $\CH_{1}$.
\begin{lemma} \label{lem.nonpert-blowup}
Let $\calR = \set{(u, V) : u_{1} \leq u \leq u_{2}, \ V_{1} \leq V < 1} \subset \calQ$ be a characteristic rectangle such that \eqref{eq:nonpert-r} and \eqref{eq:nonpert-L1-ini} hold. Assume furthermore that
\begin{align*}
	\lim_{V \to 1} \rd_{V} r(u_{1}, V) =& -\infty.
\end{align*}
Then for every $u \in [u_{1}, u_{2}]$, we have
\begin{align*}
	\lim_{V \to 1} \rd_{V} r(u, V) =& -\infty, \\
	\int_{V_{1}}^{1} \frac{(\rd_{V} \phi)^{2}}{\Omg^{2}}(u, V) \, d V =& \infty.
\end{align*}
\end{lemma}
\begin{proof}
By Lemma~\ref{lem.nonpert-C0}, there exists a constant $C > 0$ such that
\begin{equation*}
	C^{-1} \leq \Omg^{2}(u, V) \leq C \quad \hbox{ on } \overline{\mathcal R}.
\end{equation*}
Fix any $u_{1} \leq u \leq u_{2}$. Plugging in this bound to \eqref{r2.uv}, we see that
\begin{equation} \label{eq:nonpert-dvr-lip}
\abs{\rd_{u} (r \rd_{V} r)}(u, V)
\leq \frac{1}{4} \bb| 1 - \frac{\bfe^{2}}{r_{0}^{2}} \bb| \Omg^{2}(u, V) \leq C.
\end{equation}
By the fundamental theorem of calculus, we may write
\begin{equation*}
	- \rd_{V} r(u, V) = \frac{r(u_{1}, V)}{r(u, V)} (- \rd_{V} r) (u_{1}, V) - \frac{1}{r(u, V)} \int_{u_{1}}^{u} \rd_{u} (r \rd_{V} r) (u, V) \, \ud u'.
\end{equation*}
Then using \eqref{eq:nonpert-dvr-lip}, it follows that
\begin{equation*}
	\liminf_{V \to 1} (- \rd_{V} r)(u, V)
	\geq \frac{r_{0}}{R}  \liminf_{V \to 1} (- \rd_{V} r)(u_{1}, V) - C = \infty.
\end{equation*}
By the Raychaudhuri equation \eqref{eqn.Ray} for $\Omg^{-2} \rd_{V} r$, we also have
\begin{equation*}
	\int_{V_{1}}^{1} \frac{(\rd_{V} \phi)^{2}}{\Omg^{2}}(u, V) \, \ud V
	\geq \frac{1}{R} \int_{V_{1}}^{1} r \frac{(\rd_{V} \phi)^{2}}{\Omg^{2}}(u, V) \, \ud V
	= \lim_{V \to 1} \frac{(-\rd_{V} r)}{\Omg^{2}}(u, V) - \frac{(-\rd_{V} r)}{\Omg^{2}}(u, V_{1}) = \infty,
\end{equation*}
which completes the proof. \qedhere
\end{proof}

We are now ready to prove Theorem~\ref{thm.nonpert}.
\begin{proof}[Proof of Theorem~\ref{thm.nonpert}]
Let $u_{s} \in \bbR$ be as in Theorem~\ref{main.theorem.C0.stability} and fix $u_{2} \in (u_{s}, u_{\CH_{1}})$. Let $u_{1} \in (-\infty, u_{s}]$ and\footnote{Here, $V(1)$ is the value of $V$ corresponding to $v=1$.} $V_{1} \in (V(1), 1)$ be parameters to be chosen below. 

We claim that there exist $r_0, R\in \mathbb R$ such that \eqref{eq:nonpert-r} holds in $\mathcal R = \set{(u, V) : u_{1} \leq u \leq u_{2}, \, V_{1} \leq V < 1}$. First, by \eqref{eq:adm-id-adm} and Lemma~\ref{no.trapped} in Appendix~\ref{sec.subext.pf}, we have ${\rd_u r}<0$ on $\Sigma_0\cap J^-(\mathcal I^+_1)$. Therefore, by \eqref{eqn.Ray}, we have\footnote{Here, $\Omg^2$ is to be understood in the $(u,v)$ coordinate system.} $\f{\rd_u r}{\Omg^2}<0$ in $\mathcal R$. In particular, $r(u,V)$ is decreasing in $u$ for every fixed $V$. 

Observe now that the $r$-value on the Reissner--Nordstr\"om Cauchy horizon $r_-=M-\sqrt{M^2-{\bf e}^2}$ satisfies $r_-<|{\bf e}|$ since $|{\bf e}|<M$. 
By Theorem \ref{main.theorem.C0.stability}, $\lim_{u\to -\infty}\lim_{v\to \infty} r(u,v)= r_-$. Therefore, choosing $V_{1}$ sufficiently close to $1$, there exists $R\in (0,|{\bf e}|)$ such that $r(u_1,V)\leq R$ whenever $V \in [V_{1}, 1)$. The monotonicity of $r$ described above then implies that $r(u, V)\leq R$ for every $(u, V)\in \mathcal R$.

To show the lower bound in $r$, we use the fact that $(u_{2}, 1)$ lies on the non-endpoint of $\CH_{1}$ to get $\lim_{V \to 1} r(u_2,V)>0$ (cf.~Theorem~\ref{thm:kommemi}(2)(e)). Therefore, choosing $V_1$ sufficiently close to $1$, there exists $r_0 \in \mathbb R$ such that $r(u_2, V)\geq r_0 >0$ for every $V \in [V_1, 1]$. Using the monotonicity of $r$ again, we therefore obtain $r(u,V)\geq r_0$ for every $(u, V)\in \mathcal R$.

Next we verify the hypothesis \eqref{eq:nonpert-L1-ini} of Lemma~\ref{lem.nonpert-L1}. Since $u_{1} \leq u_{s}$, the terms on the second line of \eqref{eq:nonpert-L1-ini} are finite. On the other hand, finiteness of the terms on the first line of \eqref{eq:nonpert-L1-ini} simply follows from the fact that the segment $\set{(u, V) : u_{1} \leq u \leq u_{2}, \, V = V_{1}}$ is a compact subset of $\calQ$.

In conclusion, we see that $\calR$ satisfies \eqref{eq:nonpert-r} and \eqref{eq:nonpert-L1-ini}. By Lemma~\ref{lem.nonpert-C0}, $C^{0}$ extendibility to $\set{(u, V) : u_{1} \leq u \leq u_{2}, \, V = 1} = \overline{\calR} \cap \CH_{1}$ holds. Moreover, by Lemma~\ref{lem.nonpert-blowup} and Theorem~\ref{final.blow.up.step}, \eqref{eq:nonpert-blowup-phi}--\eqref{eq:nonpert-blowup-dvr} hold for every $u \in [u_{1}, u_{2}]$ provided that \eqref{final.blow.up.step.assumption} holds on $\EH_{1}$. Since $u_{2} \in (u_{s}, u_{\CH_{1}}]$ is arbitrary, Theorem~\ref{thm.nonpert} follows.
\end{proof}

\begin{remark} \label{rem.nonpert-bfsph-pf}
In the case $\calS = \0$, so that $\CH_{1} \cap \CH_{2}$ consists of a bifurcation sphere $p = (u_{\CH_{1}}, 1)$ with $r(p) > 0$, we may repeat the above proof in the coordinate system $(U_{\CH_{2}}, V_{\CH_{1}})$ (as described in Remark~\ref{rem.nonpert-bfsph}) with $\calR = [U_{\CH_{2}, 1}, 1) \times [V_{\CH_{1}, 1}, 1)$, where $U_{\CH_{2}, 1}, V_{\CH_{1}, 1}$ are sufficiently close to $1$. The crucial observation is that the lower bound $r \geq r_{0} > 0$ holds on $\calR$. We leave the details to the reader.
\end{remark}

\section{$C^2$ inextendibility: Proof of Theorem \ref{main.theorem.C2}}\label{sec.main.theorem.C2}

In this section, we prove the $C^2$-future-inextendibility of the maximal globally hyperbolic future development of the admissible Cauchy data satisfying the assumptions of Theorem \ref{main.theorem.C2}. The proof will be based on a contradiction argument. We assume for the sake of argument that $(\mathcal M,g)$ as in Theorem \ref{main.theorem.C2} is future-extendible in $C^2$, i.e., we assume that there exist $\iota$ and $(\widetilde{\mathcal M}, \tilde{g})$ as in Definition \ref{def.FE}. 
Our goal will be to derive a contradiction to the blow up of $\frac{1}{\Omg^{2}}{\rd_{V} \phi}$ derived in Section~\ref{sec.nonpert}.

In order to lighten the notation, let us write $\iota(\mathcal M)$ simply as $\mathcal M$ when there is no danger of confusion. Given $(\mathcal M,g)$, $(\widetilde{\mathcal M}, \tilde{g})$ and $\iota$ as above, we will also denote by $\rd\mathcal M$ the (topological) boundary of $\iota(\mathcal M)$ in $\widetilde{\mathcal M}$.

Let us briefly describe the ideas of the proof of Theorem \ref{main.theorem.C2}. The proof consists of two main steps:
\begin{enumerate}
\item In the first part of the proof, using ideas in \cite{Kommemi, DafRen}, we will reduce the problem to ruling out radial null geodesics $\gamma$ that exit $\mathcal M$ through\footnote{Strictly speaking, it is the projection $\boldsymbol\pi(\gamma)$ that exists through $\CH_1$ or $\CH_2$, i.e., the closure of $\boldsymbol\pi(\gamma)$ in $\mathcal Q^+$ intersects $\CH_1$ or $\CH_2$.} $\CH_1$ or $\CH_2$ and entering $\widetilde{\mathcal M}$ and such that $r>0$ at $\gamma\cap \rd\mathcal M$. (Lemmas \ref{lemma.Lipschitz}-\ref{not.at.S} and Steps~1 and 2 in Lemma~\ref{not.at.CH})
\item Next, suppose a radial null geodesic $\gamma$ exits $\mathcal Q$, say\footnote{The case where $\CH_1$ is replaced by $\CH_2$ is of course completely analogous.}, through $\CH_1$. We note that \eqref{eq:nonpert-blowup-phi} implies that the Ricci curvature of a parallely propagated vector field blows up along $\gamma$. This clearly contradicts that $\gamma$ is a geodesic in $\widetilde{\mathcal M}$. (Step~3 in Lemma~\ref{not.at.CH})
\end{enumerate}

We now begin the first step of the proof with a standard result:
\begin{lemma}\label{lemma.Lipschitz}
The boundary $\rd\mathcal M$ is locally an achronal Lipschitz hypersurface.
\end{lemma}
\begin{proof}
This is standard and can be found for instance in \cite{HawkEll}.
\end{proof}
The next lemma is due to Dafermos--Rendall \cite{DafRen}:
\begin{lemma}\label{Killing.extension}
The standard rotations $(\mathcal O_1, \mathcal O_2, \mathcal O_3)$ extend continuously to $\rd \mathcal M$.
\end{lemma}
\begin{proof}
The fact that any Killing vector field in $(\mathcal M,g)$ extends continuously to the boundary in the $C^2$ extension $(\widetilde{\mathcal M},\tilde{g})$ is proven in \cite{DafRen}. The lemma thus follows as $\mathcal O_1, \mathcal O_2, \mathcal O_3$ are Killing vector fields in $(\mathcal M, g)$.
\end{proof}
In particular, this implies,
\begin{lemma}\label{r.limit}
The area radius function $r$ extends continuously to $\rd \mathcal M$.
\end{lemma}
\begin{proof}
This follows from Lemma \ref{Killing.extension} together with $r^2=\f 12 \sum_{i=1}^3 g(\mathcal O_i,\mathcal O_i)$.
\end{proof}

\begin{lemma}\label{span.O}
Let $p\in \rd\mathcal M$ such that $r(p)\neq 0$. Then the continuous extensions of $(\mathcal O_1, \mathcal O_2, \mathcal O_3)$ at $T_p \widetilde{\mathcal M}$ span a $2$-dimensional spacelike subspace of $T_p \widetilde{\mathcal M}$ with respect to the metric $\tilde{g}$.
\end{lemma}
\begin{proof}
\pfstep{Step~1} We first show that $span\{\mathcal O_1(p), \mathcal O_2(p), \mathcal O_3(p)\}$ contains a $2$-dimensional spacelike subspace. We use the following claim: 

{\bf Claim:} For every point $q\in\mathcal M$, there exist $(i,j)\in \{1,2,3\}^2$ with $i\neq j$ such that $g(\mathcal O_i,\mathcal O_i)(q)\geq \f{2 r^2(q)}{3}$, $g(\mathcal O_j,\mathcal O_j)(q)\geq \f{2 r^2(q)}{3}$ and $\widetilde{\mathcal O}_{ij}=\mathcal O_i-\f{g(\mathcal O_i,\mathcal O_j)}{g(\mathcal O_j,\mathcal O_j)}\mathcal O_j$ satisfies $g(\widetilde{\mathcal O}_{ij},\widetilde{\mathcal O}_{ij})(q)\geq \f{r^2(q)}{2}$.

\emph{Proof of Claim:} Rescaling the estimates by the radius of the $2$-spheres, it suffices to show this on a standard $2$-sphere with $r=1$ embedded in $\mathbb R^3$ given by $\{(x,y,z):x^2+y^2+z^2=1\}$. In Cartesian coordinates, the vector fields are given by $\mathcal O_1=x\rd_y-y\rd_x$, $\mathcal O_2=y\rd_z-z\rd_y$, $\mathcal O_3=z\rd_x-x\rd_z$. Without loss of generality, we can assume that
\begin{equation}\label{ordering}
g(\mathcal O_1,\mathcal O_1)\geq g(\mathcal O_2,\mathcal O_2)\geq g(\mathcal O_3, \mathcal O_3).
\end{equation}
In this case, we choose $i=1$ and $j=2$. Since $\sum_{k=1}^3 g(\mathcal O_k, \mathcal O_k)=2$, this implies $g(\mathcal O_i, \mathcal O_i)\geq \f 23$ and $g(\mathcal O_j, \mathcal O_j)\geq \f 23$. Observe that
$$\widetilde{\mathcal O}_{ij}=x\rd_y-y\rd_x+\f{xz}{y^2+z^2}(y\rd_z-z\rd_y)=-y\rd_x+\f{x y^2}{y^2+z^2}\rd_y+\f{xyz}{y^2+z^2}\rd_z.$$
Using $x^2+y^2+z^2=1$,
$$g(\widetilde{\mathcal O}_{ij},\widetilde{\mathcal O}_{ij})=y^2+\f{x^2 y^4+x^2y^2z^2}{(y^2+z^2)^2}=\f{y^2 ((y^2+z^2)^2+ x^2 y^2+x^2  z^2)}{(y^2+z^2)^2}=\f{y^2 }{y^2+z^2}.$$
It remains to note that if $g(\widetilde{\mathcal O}_{ij},\widetilde{\mathcal O}_{ij})=\f{y^2 }{y^2+z^2}< \f 12$, then 
$$g(\mathcal O_3,\mathcal O_3)=x^2+z^2>x^2+y^2=g(\mathcal O_1,\mathcal O_1),$$
which contradicts \eqref{ordering}. Therefore, we have $g(\widetilde{\mathcal O}_{ij},\widetilde{\mathcal O}_{ij})\geq \f 12$. This concludes the proof of the claim.

We now take a sequence $\{p_k\}_{k=1}^\infty\subset \mathcal M $ such that $p_k\to p$ (which exists since $p\in \rd\mathcal M$). After passing to a subsequence if necessary, there exists $(i,j)\in \{1,2,3\}^2$ such that the conclusion of the Claim above holds for every $p_k$ (with the same choice of $i$ and $j$). By continuity, it holds at point $p$ that $g(\mathcal O_i,\mathcal O_i)(p)\geq \f{2 r^2(p)}{3}$, $g(\mathcal O_j,\mathcal O_j)(p)\geq \f{2 r^2(p)}{3}$ and $g(\widetilde{\mathcal O}_{ij},\widetilde{\mathcal O}_{ij})(p)\geq \f{r^2(p)}{2}$. Since $r(p)\neq 0$ by assumption, this implies that there exist two linearly independent spacelike vectors $\mathcal O_i$ and $\mathcal O_j$ in $T_p\widetilde{\mathcal M}$. This concludes Step 1.

\pfstep{Step~2} We then show that the dimension of $span\{\mathcal O_1(p), \mathcal O_2(p), \mathcal O_3(p)\}$ is $\leq 2$. Let $\{p_i\}_{i=1}^\infty\subset \mathcal M$ be a sequence of points such that $p_i\to p$. (Such a sequence exists since $p\in \rd \mathcal M$.) Since the tangent space on the standard $2$-sphere is $2$-dimensional, for every $i\in \mathbb N$, there exists $(a_{1,i}, a_{2,i}, a_{3,i})\in \mathbb S^2$, (i.e., $(a_{1,i}, a_{2,i}, a_{3,i})\in \mathbb R^3$ with $a_{1,i}^2 + a_{2,i}^2 + a_{3,i}^2=1$) such that $\sum_{j=1}^3 a_{j,i}\mathcal O_j(p_i)=0$. Since $\mathbb S^2$ is compact, there exists a subsequence $i_k$ such that $(a_{1,i_k}, a_{2,i_k}, a_{3,i_k})$ converges to some $(a_{1,\infty}, a_{2,\infty}, a_{3,\infty})$ in $\mathbb S^2$. The continuity of the $\mathcal O_j$ at $p$ then implies that $\sum_{j=1}^3 a_{j,\infty}\mathcal O_j(p)=0$. Therefore, $\mathcal O_1$, $\mathcal O_2$ and $\mathcal O_3$ are linearly dependent at $p$.

Combining Steps 1 and 2 yields the desired result.
\end{proof}

Using the above preliminaries, we show in the next lemma that either there exists a timelike geodesic which exits $\mathcal M$ and enters the extension $\widetilde{\mathcal M}$ through a point $p\in \rd\mathcal M$ with $r(p)=0$ or there exists a \underline{radial} null geodesic which exits $\mathcal M$ and enters the extension $\widetilde{\mathcal M}$ through a boundary point $p\in \rd\mathcal M$ with $r(p)>0$. Thus, to show that the future extension $\widetilde{\mathcal M}$ does not exist, it suffices to rule out such geodesics. This will indeed be carried out in Lemmas \ref{not.at.S} and \ref{not.at.CH}.
\begin{lemma}\label{inextendibility.main.step}
At least one of the following holds:
\begin{enumerate}
\item There exist $p\in \rd\mathcal M$ with $r(p)=0$ and a future-directed timelike geodesic $\gamma:(-\epsilon,\epsilon)\to \widetilde{\mathcal M}$ such that $\gamma((-\epsilon,0))\subset \mathcal M$ and $\gamma(0)=p$.
\item There exist $p\in \rd\mathcal M$ with $r(p)\neq 0$ and a future-directed null geodesic $\gamma:(-\epsilon,\epsilon)\to \widetilde{\mathcal M}$ such that $\gamma((-\epsilon,0))\subset \mathcal M$, $\gamma(0)=p$ and such that $\gamma\restriction_{(-\epsilon,0)}$ is \underline{radial}.
\end{enumerate}
\end{lemma}
\begin{proof}
We choose $p\in \rd\mathcal M$ such that $\rd\mathcal M$ is differentiable at $p$. Such a $p$ exists since by Lemma \ref{lemma.Lipschitz} the boundary $\rd \calM$ is Lipschitz and therefore by Rademacher's theorem it is differentiable almost everywhere\footnote{To be understood in terms of the Lebesgue measure with respect to local coordinates.}. The two cases\footnote{Notice that $r\geq 0$ in $\mathcal M$ and thus the limit on $\rd\mathcal M$ must also be non-negative.} in the statement of the lemma depend on whether $r(p)=0$ or $r(p)>0$.

Suppose $r(p)=0$. We construct $\gamma:(-\epsilon,\epsilon)\to \widetilde{\mathcal M}$ to be any future-directed timelike geodesic such that $\gamma(0)=p$. It suffices to show that for $\epsilon$ sufficiently small, $\gamma((-\epsilon,0))\subset\mathcal M$. Suppose not, then for every $\epsilon$, there exists $s\in (-\epsilon,0)$ such that $\gamma(s)\in \widetilde{\mathcal M}\setminus\mathcal M$. Since $\gamma$ is future-directed, we have $p\in I^+(\gamma(s))$. Since $I^+(\gamma(s))$ is open, there exists an open neighbourhood $\mathcal U$ of $p$ such that $\mathcal U\subset I^+(\gamma(s))$. On the other hand, since $p\in\rd\mathcal M$, $\mathcal U$ contains a point of $\mathcal M$. In other words, $I^+(\gamma(s))$ contains a point of $\mathcal M$, which contradicts the fact that $\widetilde{\mathcal M}$ is a future extension.

It remains to consider the case where $r(p)\neq 0$. Let $X$, $Y$ be two linearly independent past-directed null vectors at $T_p\widetilde{\mathcal M}$ which are normal to $span\{\mathcal O_1, \mathcal O_2, \mathcal O_3\}$ with respect to the metric $\tilde{g}$. (Notice that since $span\{\mathcal O_1, \mathcal O_2, \mathcal O_3\}$ is spacelike and $2$-dimensional according to Lemma \ref{span.O}, such $X$ and $Y$ exist. Moreover, the choices of $X$ and $Y$ are unique up to re-scalings.) We claim that either\footnote{Note that $T_p (\rd\mathcal M)$ is well-defined by the choice of $p$} $X\not\in T_p (\rd\mathcal M)$ or $Y\not\in T_p (\rd\mathcal M)$. Otherwise, $X+Y\in T_p(\rd\mathcal M)$ is a past-directed timelike vector and this contradicts Lemma \ref{lemma.Lipschitz}, which states that $\rd\mathcal M$ is achronal. Without loss of generality, we assume $X\not\in T_p (\rd\mathcal M)$.

We now construct $\gamma$ by solving for the unique geodesic through $p$ which is initially $-X$. Note that $\gamma$ is future-directed.

{\bf Claim: $\gamma((-\epsilon,0))\subset \mathcal M$ for sufficiently small $\epsilon$.}

We assume for the sake of contradiction that there exists a sequence of negative real numbers $s_n\nearrow 0$ such that $\gamma(s_n)\in \widetilde{\mathcal M}\setminus \mathcal M$. First, we see that in fact for every $s_n$ we must have $\gamma(s_n)\in \rd\mathcal M$. Otherwise there exists an open neighbourhood $\mathcal U$ of $\gamma(s_n)$ which is a subset of $\widetilde{\mathcal M}\setminus \mathcal M$ and in particular there exists $q\in \mathcal U\subset \widetilde{\mathcal M}\setminus \mathcal M$ such that $I^+(q)$ contains $p$. Since $I^+(q)$ is open, it also contains a small open neighbourhood $\mathcal V$ of $p$. On the other hand, since $p\in \rd\mathcal M$ and $\mathcal V$ is an open neighbourhood of $p$, $\mathcal V$ must contain a point $v\in \mathcal M$. Now $v\in I^+(q)\cap \mathcal M$ where $q\in \widetilde{\mathcal M}\setminus \mathcal M$. This contradicts the assumption that $\widetilde{\mathcal M}$ is a future extension of $\mathcal M$.

It thus remains to rule out the possibility that $\gamma(s_n)\in \rd\mathcal M$ for all $s_n$. This is indeed impossible since $\dot\gamma(0)=-X$ is not tangential to $\rd\mathcal M$ at $p$ and $\rd\mathcal M$ is differentiable at $p$.

{\bf Claim: $\gamma\restriction_{(-\ep,0)}$ is radial.} 

Assume that it is not, i.e., there exists $s\in(-\epsilon,0)$ such that $g(\dot\gamma(s),\mathcal O_i)=a\neq 0$ for some $i=1,2,3$. On the other hand, since $\gamma$ is a geodesic and $\mathcal O_i$ is Killing in $(\mathcal M, g)$, we have $\dot\gamma\left(g(\dot\gamma(s),\mathcal O_i)\right)=g(\nab_{\dot{\gamma}}\dot{\gamma},\mathcal O_i)+g(\dot{\gamma},\nab_{\dot{\gamma}}\mathcal O_i)=0$, i.e., $g(\dot\gamma(s),\mathcal O_i)$ is a constant along $\gamma$. In particular, since $\dot\gamma$, $\mathcal O_i$ and $\tilde{g}$ are all continuous up to $p$, we have 
$\tilde{g}(-X,\mathcal O_i)(p)=a\neq 0$, contradicting the choice of $X$.
\end{proof}

Equipped with Lemma \ref{inextendibility.main.step}, it now remains to show that both alternatives (1) and (2) cannot hold. The following lemma is immediate from results in \cite{Kommemi}.
\begin{lemma}\label{not.at.S}
The alternative (1) in Lemma \ref{inextendibility.main.step} cannot hold.
\end{lemma}
\begin{proof}
Assume for the sake of contradiction that $p$ and $\gamma$ are as in (1) in Lemma \ref{inextendibility.main.step}. 
By the Raychaudhuri equations, as well as the fact that $r(p) = 0$ while $r \geq 0$ everywhere, there exists a neighbourhood $\calU \subset \widetilde{\calM}$ of $p$ such that $1 - \frac{2m}{r} \leq 0$ (which is equivalent to $\rd_{u} r \rd_{v} r \geq 0$ in a double null coordinate) on $\calU \cap \calM$. By \cite[Section~5.8]{Kommemi}, the following lower bound of the Kretschmann scalar holds on $\calU \cap \calM$: $R_{\mu\nu\alpha\beta}R^{\mu\nu\alpha\beta}\geq \f{c}{r^4}$, where $R_{\mu\nu\alpha\beta}$ is the Riemann curvature tensor of $(\mathcal M,g)$ and $c>0$ is a constant (depending on $p$). This clearly contradicts the fact that $\gamma$ is a geodesic in $\widetilde{\mathcal M}$ passing through $p$, at which $r(p) = 0$.
\end{proof}
Finally, we rule out the possibility that $r(p)>0$ for $p$ as in Lemma \ref{inextendibility.main.step}:
\begin{lemma}\label{not.at.CH}
The alternative (2) in Lemma \ref{inextendibility.main.step} cannot hold.
\end{lemma}
\begin{proof}
Assume for the sake of contradiction that $p$ and $\gamma$ are as in (2) in Lemma \ref{inextendibility.main.step}. One sees that $\gamma$ must exit $\mathcal M$ through\footnote{This means that the closure of the projection of $\gamma\restriction_{\mathcal M}$ in $\mathcal Q^+$ intersects those components of the boundary.} $\mathcal C\mathcal H^+_1$ or $\mathcal C\mathcal H^+_2$ or $\mathcal I^+_1$ or $\mathcal I^+_2$ or $i^+_1$ or $i^+_2$ (recall Theorem \ref{thm:kommemi}). 

\pfstep{Step 1: $\gamma$ cannot exit through $\mathcal I^+_1$ or $\mathcal I^+_2$} This simply follows from the fact that $r$ has a finite limit on $\rd \mathcal M$ (Lemma~\ref{r.limit}), together with the definition that $r\to \infty$ along radial null curves towards $\mathcal I^+_1$ or $\mathcal I^+_2$.

\pfstep{Step 2: $\gamma$ cannot exit through $i^+_1$ or $i^+_2$} If $\gamma$ exits through $i^+_1$ or $i^+_2$, then $\gamma\cap \mathcal M\subset \EH_1$ or $\gamma\cap \mathcal M\subset \EH_2$. However, $\EH_1$ and $\EH_2$ are both future affine complete null geodesics (cf. Remark~\ref{rmk.completeness}), hence it is impossible for $\gamma$ to leave $\mathcal M$.

\pfstep{Step 3: $\gamma$ cannot exit through $\mathcal C\mathcal H^+_1$ or $\mathcal C\mathcal H^+_2$} 
Suppose the contrary. Assume without loss of generality that $\gamma$ exits through $\mathcal C\mathcal H^+_1$. 

We use the same coordinate system $(u, V)$ as in Theorem~\ref{thm.nonpert}, in which ${\boldsymbol \pi}(\gamma_{(-\epsilon,0)})\in \mathcal Q$ is given by a constant $u$ curve - let us assume that it is a subset of $\{(u,V):u=u_r\}$ for some $u_r\in (-\infty, u_{\CH_{1}}]$. Since $(\widetilde{\mathcal M},\tilde{g})$ is $C^2$ and $\gamma$ is a geodesic, the component of the Ricci curvature $\lim_{s\to 0-} Ric(\dot\gamma(s),\dot\gamma(s))$ is bounded. Noting that $\f{1}{\Omg^2}\rd_V$ is geodesic, we have $\dot\gamma= \f{c}{\Omg^2}\rd_V$ for some constant $c>0$ and therefore by \eqref{EMSFS},
\begin{equation}\label{geod.bdd}
\f{1}{\Omg^4}(\rd_V\phi)^2(u_r,V)=\f 12 Ric \left( \f{1}{\Omg^2}\rd_V,\f{1}{\Omg^2}\rd_V \right)(u_r, V)\mbox{ is bounded as } V \to 1.
\end{equation}
When $u_{r} < u_{\CH_{1}}$, this statement in combination with the boundedness of $\log \Omg$ from Lemma~\ref{lem.nonpert-C0} immediately contradicts the blow up \eqref{eq:nonpert-blowup-phi} in Theorem~\ref{thm.nonpert}. When $u_{r} = u_{\CH_{1}}$, which is only possible when $\calS = \emptyset$ and $p$ is the bifurcation sphere with $r(p) > 0$, the same argument works by Remark~\ref{rem.nonpert-bfsph-pf}. \qedhere
\end{proof}

Putting together the above lemmas, we conclude the proof of Theorem \ref{main.theorem.C2}:
\begin{proof}[Proof of Theorem \ref{main.theorem.C2}]
Lemmas \ref{inextendibility.main.step}, \ref{not.at.S} and \ref{not.at.CH} obviously lead to a contradiction. Therefore, no $C^2$-future-extensions as in Definition \ref{def.FE} can exist.
\end{proof}

\appendix

\section{Subextremality of the event horizons}\label{sec.subext.pf}
The goal of this section of the appendix is to prove Proposition~\ref{prop.subextremality}. In what follows, we denote by $(\calM, g, \phi, F)$ a maximal globally hyperbolic future development of an $\omg_{0}$-admissible initial data (where $\omg_{0} > 2$). First, we will need two lemmas, which will also be useful in later parts of the paper. 

\begin{lemma}[The exterior regions are free of trapped (or anti-trapped) surfaces]\footnote{A $2$-sphere given by constant $(u,v)$ is called \emph{trapped} (resp. \emph{anti-trapped}) if $\rd_v r$ and $\rd_u r$ are both negative (resp. positive).}\label{no.trapped}
Let $(u, v)$ be a double null coordinate system on $\calQ = \calM / SO(3)$ normalized according to Lemma~\ref{lem:cauchy-to-char}. If $u\leq u_{\EH_1}$, then
\begin{equation}\label{no.trapped.in.1}
\rd_v r(u,v)\geq 0,\quad \rd_u r(u,v) < 0.
\end{equation}
If $v\leq v_{\EH_2}$, then
\begin{equation}\label{no.trapped.in.2}
\rd_u r(u,v)\geq 0,\quad \rd_v r(u,v) < 0.
\end{equation}
\end{lemma}
\begin{proof}
It clearly suffices to check \eqref{no.trapped.in.1} as \eqref{no.trapped.in.2} is similar. Suppose $u\leq u_{\EH_1}$. We first establish $\rd_v r(u, v)\geq 0$. Assume this is not the case, i.e., there exists $(u_c,v_c) \in \calQ$ with $u_c\leq u_{\EH}$ such that $\rd_v r(u_c, v_c)<0$. Then by continuity, there exists $(u_{c}', v_{c}') \in \calQ$ such that $u_c'<u_c\leq u_{\mathcal H^+}$ and $\rd_v r(u'_c, v'_c)<0$. By the Raychaudhuri equation \eqref{eqn.Ray}, we then have $\rd_v r(u'_c, v)<0$ for every $v\geq v'_c$, which contradicts the fact that $\sup r(u'_c, \cdot) = \infty$ (see Definition~\ref{def.EH}). Hence, we have established
\begin{equation}\label{no.trapped.dvr}
\rd_v r (u,v) \geq 0\quad\mbox{for }u\leq u_{\EH}.
\end{equation}

Next, by the admissibility condition \eqref{eq:adm-id-adm} (which by Lemma~\ref{lem:cauchy-to-char} implies $\rd_u r\restriction_{\Sigma_0}<0$ for $\rho \geq \rho_1$ and $\rd_v r\restriction_{\Sigma_0}<0$ for $\rho \leq \rho_2$) and the Raychaudhuri equation \eqref{eqn.Ray}, for every $(u,v)\in \mathcal Q$, either $\rd_u r(u,v)<0$ or $\rd_v r (u,v)<0$. Hence, \eqref{no.trapped.in.1} follows from this observation and \eqref{no.trapped.dvr}.
\end{proof}

\begin{lemma}[Limiting values of $r$ and $\varpi$ on $\EH_1$ and $\EH_2$]
Let $(u, v)$ be a double null coordinate system on $\calQ$ normalized according to Lemma~\ref{lem:cauchy-to-char}.
Define
$$r_{\EH_1}=\sup_{\EH_1} r,\quad r_{\EH_2}=\sup_{\EH_2} r,\quad \varpi_{\EH_1}=\sup_{\EH_1} \varpi,\quad \varpi_{\EH_2}=\sup_{\EH_2} \varpi.$$
Then 
\begin{equation}\label{r.varpi.limits}
r_{\EH_1}=\lim_{v\to\infty} r(u_{\EH_1}, v),\quad r_{\EH_2}=\lim_{u\to \infty} r(u,v_{\EH_2}),\quad \varpi_{\EH_1}=\lim_{v\to\infty} \varpi(u_{\EH_1}, v),\quad \varpi_{\EH_2}=\lim_{u\to \infty} \varpi(u,v_{\EH_2}).
\end{equation}
Moreover,
\begin{equation}\label{r.varpi.poly}
r_{\EH_1}=\varpi_{\EH_1} +\sqrt{\varpi_{\EH_1}^2-\e^2},\quad r_{\EH_2}=\varpi_{\EH_2} +\sqrt{\varpi_{\EH_2}^2-\e^2}.
\end{equation}
\end{lemma}
\begin{proof}
In the proof, we only consider the case of $\EH_{1}$; the other case can be handled similarly.
We recall the (nontrivial!) fact that $r_{\EH_1}<\infty$. This follows from the Penrose inequality in the spherically symmetric setting, see \cite[Lemma 3]{DafTrapped}.

\pfstep{Step~1: Proof of \eqref{r.varpi.limits}} To prove \eqref{r.varpi.limits} for $r$, note that by Lemma~\ref{no.trapped}, $r$ is non-decreasing (towards the future) along the event horizons. To prove \eqref{r.varpi.limits} for $\varpi$, note that by Lemma~\ref{no.trapped} and the equations for $\varpi$ in \eqref{eq:EMSF-r-phi-m}, $\varpi$ is non-decreasing (towards the future) along the event horizons.

\pfstep{Step~2: Proof of $r_{\EH_{1}} = \varpi_{\EH_{1}} \pm \sqrt{\varpi_{\EH_{1}}^{2} - \e^{2}}$}
Recall that $1-\mu = - 4 \Omg^{-2} \rd_{u} r \rd_{v} r$. Thus by Lemma~\ref{no.trapped}, we have $1-\mu \geq 0$ in the exterior region $\set{u \leq u_{\EH_{1}}}$. On $\EH_{1}$, we claim that
\begin{equation} \label{eq:1-mu-vanish}
	\liminf_{v \to \infty} (1 - \mu) (u_{\EH_{1}}, v) = \frac{r_{\EH_{1}}^{2} - 2 \varpi_{\EH_{1}} r_{\EH_{1}} + \e^{2}}{r_{\EH_{1}}^{2}} = 0.
\end{equation}
Since $0 < r \restriction_{\Sgm_{0} \cap \EH_{1}} \leq r_{\EH_{1}} < \infty$, the desired formula for $r_{\EH_{1}}$ would then follow by the quadratic formula.

For the purpose of contradiction, suppose that \eqref{eq:1-mu-vanish} is false. Hence there exists $c > 0$ and $v_{0}$ such that 
\begin{equation} \label{eq:1-mu-pos}
	(1-\mu) (u_{\EH_{1}}, v) \geq c \quad \hbox{ for } v \geq v_{0}.
\end{equation}
The idea is to show that \eqref{eq:1-mu-pos}, when combined with the fact that $r_{\EH_{1}} < \infty$, implies that $r$ is also bounded on nearby outgoing null curves, which contradicts Definition~\ref{def.EH}. 

To make this idea precise, we use a bootstrap argument. Let $\veps > 0$ and $A \geq 2$ be small and large parameters (respectively) to be fixed later. Replacing $v_{0}$ by a larger number if necessary, we may ensure that $r_{\EH_{1}} - r(u_{\EH_{1}}, v_{0})< \veps$ and $\varpi_{\EH_{1}} - \varpi(u_{\EH_{1}}, v_{0}) < \veps$. Furthermore, we choose $u_{0} < u_{\EH_{1}}$ so that $r(u_{0}, v_{0}) - r(u_{\EH_{1}}, v_{0}) < \veps$ and $\varpi(u_{0}, v_{0}) - \varpi(u_{\EH_{1}}, v_{0}) < \veps$. Given $v_{1} > v_{0}$, we make the bootstrap assumption
\begin{equation} \label{eq:1-mu-btstrp}
	r(u_{0}, v) - r(u_{\EH_{1}}, v) < 4 \veps, \quad
	\varpi(u_{0}, v) - \varpi(u_{\EH_{1}}, v) < 2 A \veps \quad 
	\hbox{ for } v_{0} \leq v \leq v_{1},
\end{equation}
which initially holds for some $v_{1} > v_{0}$ by continuity.
For a large enough $A$ and a sufficiently small $\veps$, we claim that \eqref{eq:1-mu-btstrp} can be improved to
\begin{equation} \label{eq:1-mu-goal}
	r(u_{0}, v) - r(u_{\EH_{1}}, v) < 2 \veps, \quad
	\varpi(u_{0}, v) - \varpi(u_{\EH_{1}}, v) < A \veps \quad 
	\hbox{ for } v_{0} \leq v \leq v_{1}.
\end{equation}
From the claim, it would follow (via a simple continuity argument) that $r(u_{0}, v) - r(u_{\EH_{1}}, v) < 2\veps$ for every $v \geq v_{1}$, which is impossible due to Definition~\ref{def.EH} and the fact that $r_{\EH_{1}} < \infty$.

In what follows, we denote by $C$ any constant that may depend on $c$, $\phi$ or geometric quantities of the spacetime (e.g., $\inf_{\calQ \cap \set{u \leq u_{\EH_{1}}}} r$, $\sup_{\calQ \cap \set{u \leq u_{\EH_{1}}}} r$, $\sup_{\calQ \cap \set{u \leq u_{\EH_{1}}}} \abs{\varpi}$, $\e$ etc.), but \underline{independent} of $\veps$ and $A$. By \eqref{eq:1-mu-pos} and the bootstrap assumption \eqref{eq:1-mu-btstrp}, we have
\begin{equation*}
	(1-\mu)(u, v) \geq \frac{c}{2} \quad \hbox{ for } (u, v) \in [u_{0}, u_{\EH_{1}}] \times [v_{0}, v_{1}]
\end{equation*}
for a small enough $\veps$ (depending only on $c$ and $A$). By the first equation in \eqref{eq:EMSF-r-phi-m}, it follows that
\begin{equation*}
	\abs{\rd_{v} \log(-\rd_{u} r)} = \left\vert \frac{2 (\varpi - \frac{\e^{2}}{r})}{r^{2} (1-\mu)} \rd_{v} r\right\vert
	\leq C \rd_{v} r,
\end{equation*}
in $[u_{0}, u_{\EH_{1}}] \times [v_{0}, v_{1}]$. Integrating in $v$, we see that $e^{- C \veps} \leq \frac{- \rd_{u} r(u, v)}{- \rd_{u} r(u, v_{0})} \leq e^{C \veps}$. Integrating in $u$ and recalling our choice of $(u_{0}, v_{0})$, we arrive at
\begin{equation*}
	r(u_{0}, v) - r(u_{\EH_{1}}, v) \leq \veps e^{C \veps} \quad \hbox{ for } v \in [v_{0}, v_{1}].
\end{equation*}
Hence, for sufficiently small $\veps$, the first inequality of \eqref{eq:1-mu-goal} follows.

To improve the estimate for $\varpi$ (i.e., the second inequality of \eqref{eq:1-mu-goal}), we need to use the other equations in \eqref{eq:EMSF-r-phi-m}. Let $(u, v) \in [u_{0}, u_{\EH_{1}}] \times [v_{0}, v_{1}]$. By the first and the second equations, as well as the preceding estimates, we have
\begin{equation*}
	\left\vert \rd_{v} \left( r \frac{\rd_{u} \phi}{- \rd_{u} r} \right) \right\vert
	= \left\vert - \rd_{v} \log (- \rd_{u} r ) r \frac{\rd_{u} \phi}{- \rd_{u} r} + \rd_{v} \phi \right\vert
	\leq C \rd_{v} r \left\vert r \frac{\rd_{u} \phi}{- \rd_{u} r} \right\vert + \abs{\rd_{v} \phi}.
\end{equation*}
By the third equation in \eqref{eq:EMSF-r-phi-m} and Cauchy--Schwarz, we have
 \begin{align*}
\int_{v_{0}}^{v} \abs{\rd_{v} \phi} (u, v')\, \ud v'
\leq & \left( \int_{v_{0}}^{v} \frac{1}{2} r^{2} (1-\mu) \frac{(\rd_{v} \phi)^{2}}{\rd_{v} r} (u, v')\, \ud v' \right)^{1/2}
\left(\int_{v_{0}}^{v} \frac{2}{r^{2} (1-\mu)} \rd_{v} r (u, v')\, \ud v' \right)^{1/2} \\
\leq & C (\varpi (u, v) - \varpi(u, v_{0}))^{1/2} (r(u, v) - r(u, v_{0}))^{1/2} \leq C \veps.
\end{align*}
Hence, by Gr\"onwall's inequality, it follows that
\begin{equation*}
	\left\vert r \frac{\rd_{u} \phi}{-\rd_{u} r}\right\vert	(u, v) \leq e^{C \veps} \left( \left\vert r \frac{\rd_{u} \phi}{-\rd_{u} r}\right\vert (u, v_{0}) + C \veps \right) \leq C e^{C \veps} \quad \hbox{ for } (u, v) \in [u_{0}, u_{\EH_{1}}] \times [v_{0}, v_{1}],
\end{equation*}
where the last inequality holds since $[u_{0}, u_{\EH_{1}}]\times \{v_{0}\}$ is a fixed compact set.
Finally, using the fourth equation in \eqref{eq:EMSF-r-phi-m}, as well as the preceding estimates, we have
\begin{equation*}
	\abs{\rd_{u} \varpi} = \left\vert \frac{1}{2} (1-\mu) r^{2} \frac{(\rd_{u} \phi)^{2}}{(-\rd_{u} r)^{2}} \rd_{u} r  \right\vert \leq C (-\rd_{u} r)
\end{equation*}
in $[u_{0}, u_{\EH_{1}}] \times [v_{0}, v_{1}]$. Integrating in $u$ and using the bootstrap assumption, we obtain
\begin{equation*}
	\varpi (u_{0}, v) - \varpi (u_{\EH_{1}}, v) \leq C \veps e^{C \veps}.
\end{equation*}
Taking $A$ sufficiently large (depending on $C$) then $\veps$ small enough (depending on $A$ and $C$), we obtain the desired improvement for $\varpi$, which completes the proof of \eqref{eq:1-mu-goal}.

\pfstep{Step~3: Proof of \eqref{r.varpi.poly}}
Since $r_{\EH_{1}} > 0$, we necessarily have $\varpi_{\EH_{1}} > 0$ by the previous step. When $\varpi_{\EH_{1}} = \abs{\e}$, which corresponds to the extremal case, there is nothing to prove. Therefore, we may assume that $\abs{\e} < \varpi_{\EH_{1}}$ and moreover focus on ruling out the scenario $r_{\EH_{1}} = \varpi_{\EH_{1}} - \sqrt{\varpi_{\EH_{1}}^{2} - \e^{2}}$.

By the monotonicity properties of $\varpi$ in the exterior region $u \leq u_{\EH_{1}}$ (using Lemma~\ref{no.trapped}), we can find $v_{far}$ such that $\varpi(u, v) > \abs{\e}$ in the characteristic rectangle $\calR = \set{u \leq u_{\EH_{1}}, \ v \geq v_{far}}$. Recall from Step~2 that 
\begin{equation*}
1-\mu = \frac{r^{2} - 2 \varpi r + \e^{2}}{r^{2}} \geq 0 \hbox{ where } u \leq u_{\EH_{1}}.
\end{equation*}
Since $r_{\EH_{1}} = \varpi_{\EH_{1}} - \sqrt{\varpi_{\EH_{1}}^{2} - \e^{2}}$, it follows (by a continuity argument) that $r(u, v) \leq \varpi(u, v) - \sqrt{\varpi^{2}(u, v) - \e^{2}}$ in $\calR$. In particular, $r \leq \sup_{\Sgm_{0} \cap \set{u \leq u_{\EH_{1}}}} \varpi < \infty$ in $\calR$, which is impossible in view of Definition~\ref{def.EH}. \qedhere

\end{proof}

We are now ready to prove Proposition~\ref{prop.subextremality}:
\begin{proof}[Proof of Proposition~\ref{prop.subextremality}]
Define double null coordinates according to Lemma~\ref{lem:cauchy-to-char}. Define $\underline{u}(v)$ and $\underline{v}(u)$ such that $(u,\underline{v}(u)),\,(\underline{u}(v),v)\in \Sigma_0$. We prove the proposition only for $\EH_1$ as the proof for $\EH_2$ is similar. For the rest of this proof, $\EH_1$ will be denoted as $\EH$. Moreover, let $u_{\EH}=u_{\EH_1}$ be defined as in Definition~\ref{def.EH}. Without loss of generality, we assume that $\e > 0$.

\pfstep{Step~0: Limit must be subextremal or extremal} Let us first note that the limit along $\EH$ must be either subextremal ($\lim_{v\to \infty}\varpi(u_{\EH},v)>{\bf e}$) or extremal ($\lim_{v\to \infty}\varpi(u_{\EH},v)={\bf e}$). This follows from \eqref{r.varpi.poly} and the fact that $r_{\EH}$ is real.

{\bf Contradiction assumption:} Suppose for the sake of contradiction that the limit along $\EH$ is extremal, i.e., $\lim_{v\to \infty}\varpi(u_{\mathcal H^+},v)={\bf e}$.

\pfstep{Step~1: Definition of $u_*$, $u_{**}$ and $v_*$} Let 
$$v_*=\inf \{ v\in \mathbb R: \rd_v r (\underline{u}(v'),v')\geq 0 \mbox{ for all }v'\geq v\}.$$ 
Such a $v_*$ exists since the set above is non-empty and bounded below, by the asymptotic flatness condition \eqref{eq:adm-id-af} and Lemma~\ref{lem:cauchy-to-char}. Moreover, 
\begin{enumerate}
\item it holds that $v_*\leq \underline{v}(u_{\mathcal H^+})$ by Lemma~\ref{no.trapped}; and
\item it also holds that there exists an $\ep_c>0$ such that $\rd_u r (u, \underline{v}(u))<0$ for all $u \in (-\infty,\ub(v_*)+\ep_c)$ by the admissibility condition \eqref{eq:adm-id-adm}.
\end{enumerate}
Notice in particular that by the continuity of $\rd_v r$, we have $\rd_v r(\ub(v_*),v_*)=0$.

Next, choose $u_*>\ub(v_*)$ such that $\rd_v r(u_*, \underline{v}(u_*))<0$. By the definition of $v_*$, $u_*> \ub(v_*)$ can be chosen arbitrarily close to $\ub(v_*)$ such that 
\begin{enumerate}
\item the region $\calD_{*} =\{(u,v): u\in [\ub(v_*), u_*],\, v\in [\underline{v}(u), v_{*}]\}$ is contained in the maximal globally hyperbolic future development of the initial data; and
\item $\rd_u r(u, \underline{v}(u))<0$ for all $u\leq u_*$ (for which we use point (2) in the properties of $v_*$ above).
\end{enumerate}
By the Raychaudhuri equation \eqref{eqn.Ray}, it follows that
\begin{equation}\label{dur<0.u**}
\rd_u r (u, v) < 0 \quad \hbox{ for all } (u, v) \in \calD_{\ast}.
\end{equation}
Now on the constant $v=v_*$ curve, define $u_{**}$ by
$$u_{**}=\sup \{ u\geq \ub(v_*): \rd_v r (u',v_*)\geq 0 \mbox{ for all }u'\in [\ub(v_*), u]\}.$$ 
Such an $u_{**}\geq \ub(v_*)$ exists since the set is nonempty (by the fact $\rd_v r (\ub(v_*),v_*)=0$) and bounded above (since $\rd_v r (u_*, v_*)<0$ by the choice of $u_*$ and the Raychaudhuri equation \eqref{eqn.Ray}). 

The points defined above are depicted in the diagram below:

\begin{figure}[h] 
\begin{center}
\def\svgwidth{150px}
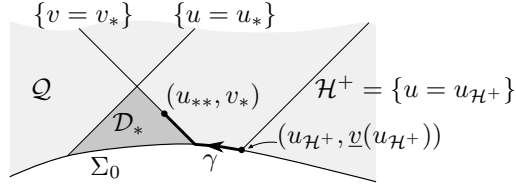 
\caption{Depiction of $u_{*}$, $u_{**}$, $v_{*}$, $\calD_{*}$ and $\gmm$} \label{fig:subextremality}
\end{center}
\end{figure}

\pfstep{Step~2: Monotonicity of $\varpi$ and $r$}
We now connect $\EH$ to the point $(u_{**}, v_*)$ via a piecewise smooth curve $\gamma$ with increasing $u$ and decreasing $v$ as depicted in Figure~\ref{fig:subextremality}. More precisely, we connect $\gmm(0) = (u_{\EH}, \underline{v}(u_{\EH}))$ to $\gmm(\frac{1}{2}) = (\ub(v_*), v_*)$ following $\Sgm_{0}$, then connect $\gmm(\frac{1}{2})$ to $\gmm(1) = (u_{**}, v_{*})$ following the null curve $\set{v = v_{*}}$. 

While it is possible that $\gmm$ is degenerate (i.e., constant in the parameter) in some parts, we nevertheless have $\dot{\gmm}^{u} \geq 0$ and $\dot{\gmm}^{v} \leq 0$. Moreover, by the construction in Step~1, we have $\rd_{u} r < 0$ and $\rd_{v} r \geq 0$ on the image of $\gmm$. Therefore, $\frac{\ud}{\ud s} r \circ \gmm(s) = \dot{\gmm}^{u} \rd_{u} r + \dot{\gmm}^{v} \rd_{v} r \leq 0$ and $\frac{\ud}{\ud s} \varpi \circ \gmm(s) = \dot{\gmm}^{u} \rd_{u} \varpi + \dot{\gmm}^{v} \rd_{v} \varpi \leq 0$ (where we used \eqref{eq:EMSF-r-phi-m}). It follows that
\begin{align} 
\varpi(u_{**}, v_*)\leq & \varpi(u_{\EH}, \underline{v}(u_{\EH}))\leq  \lim_{v\to \infty} \varpi(u_{\EH}, v)= \e,  \label{ext.varpi.ineq} \\
r(u_{**}, v_*)\leq & r(u_{\EH}, \underline{v}(u_{\EH})) \leq \lim_{v\to \infty} r(u_{\EH}, v) = \e. \label{ext.r.ineq}
\end{align}

\pfstep{Step~3: Conclusion of the proof}
By the first equation in \eqref{eq:EMSF-r-phi-m} and the definition of $\mu$, we have
\begin{equation}\label{ext.r.wave.ineq}
\rd_{u} \rd_{v} r (u_{**}, v_*)= -\frac{\Omg^2 (\varpi - \frac{\e^{2}}{r})}{2 r^{2}}(u_{**}, v_*)\geq 0,
\end{equation}
where the last inequality follows from \eqref{ext.varpi.ineq} and \eqref{ext.r.ineq}.
On the other hand, it is impossible to have $\rd_{u} \rd_v r (u_{**}, v_*)> 0$. Indeed, this contradicts the choice of $u_{**}$ as one would then have $\rd_v r (u, v_*)\geq 0$ for $u$ in a neighborhood of $u_{**}$.

Therefore, in order to conclude the contradiction argument, it suffices to rule out $\rd_{u} \rd_v r (u_{**}, v_*)= 0$. In view of the above argument, this holds only if equalities hold in \eqref{ext.varpi.ineq} and \eqref{ext.r.ineq}, i.e., $\varpi(u_{**}, v_*)=r(u_{**}, v_*)=\bf e$. Now, by \eqref{dur<0.u**}, $\rd_u r(u_{**}, v_*)<0$. Importantly, this is a strict inequality. Thus, there exist $c>0$ and $\ep_r>0$ such that $r(u,v_*)-\bf e\leq -c(u-u_{**})$ for $u\in [u_{**}, u_{**}+\ep_r]$. Using the fact that $\Omg^2$, $\varpi$ and $r$ are $C^1$ and also the monotonicity \eqref{ext.varpi.ineq}, one deduces that there exist $c'>0$ and $\ep_r'>0$ such that
$$\rd_{u} \rd_v r (u, v_*)\geq c'(u-u_{**})$$
for $u\in [u_{**}, u_{**}+\ep'_r]$. This then contradicts the choice of $u_{**}$ as $\rd_v r (u, v_*)\geq 0$ for $u$ in a neighborhood of $u_{**}$. This concludes the contradiction argument.
\end{proof}

\section{Gauge condition on the event horizon}\label{sec:appendix.gauge}
In the statement of Theorem~\ref{main.theorem.C0.stability}, we imposed a gauge condition \eqref{gauge.1} on $C_{-\infty}$ which is most convenient for proving the estimates. On the other hand, it is less convenient to apply the theorem with this gauge condition (and in fact it is a priori not entirely clear that such a gauge exists). In this section of the appendix, we compare the gauge condition \eqref{gauge.1} with the gauge condition $\f{\rd_v r}{1-\mu}=1$, which is more often used in the literature (for instance in \cite{DRPL}). More precisely, we have the following result:

\begin{proposition} \label{prop:gauge}
Consider a characteristic initial data set on $C_{-\infty} \cup \underline{C}_{1}$ posed in a $(U,v)$ coordinate system defined on $[0, U_{0}] \times [1, \infty)$ such that $C_{-\infty}= \set{(U, v) : U = 0, \, v \geq 1}$ and $\underline{C}_{1} = \set{(U, v) : 0 \leq U < U_{0}, \, v = 1}$. Assume that the following hold on $C_{-\infty}$:
\begin{enumerate}[(a)]
\item The gauge condition $\f{\rd_v r}{1-\mu}=1$ (equivalently, $\f{\rd_U r}{\Omg^2_{\mathcal H}}=-\f 14$) is imposed;
\item $r$ is non-decreasing, $\lim_{v \to \infty} r = M +\sqrt{M^2-{\bf e}^2}$ and $\lim_{v \to 1+} r = r_{0}$ for some constants $0 < \e^{2} < M$ and $r_{0} > 0$;
\item $\lim_{v \to \infty} \varpi = M$, where $\varpi$ is as in \eqref{varpi.def};
\item The decay estimate \eqref{phi.EH.decay} holds for some constants $E > 0$ and $s > 1$;
\item $\rd_{U} r < 0$.
\end{enumerate}
Then there exists a change of coordinates $\widetilde{v}(v)$ with $\widetilde{v}:[1,+\infty) \to [1,+\infty)$ such that with respect to the $(U,\widetilde{v})$ coordinate system, we have $\widetilde{v}(1) = 1$ and
\begin{equation} \label{eq:gauge-v}
\widetilde{\Omg}^2_{\mathcal H}:= -2 g(\rd_{U}, \rd_{\widetilde{v}})=\f{4e^{-2\kappa_+ r_+}(r_+-r_-)^{1+\f{\kappa_+}{\kappa_-}}}{r_+^2}e^{2\kappa_+ \widetilde{v}} \quad \hbox{ on } C_{-\infty}.
\end{equation}
Moreover, the decay estimate \eqref{phi.EH.decay} holds with respect to $\widetilde{v}$ on $C_{-\infty}$ with $C_{E, s, M, \bfe, r_{0}, \rd_{U} r(0, 1)} \, \widetilde{v}^{-s}$ on the right-hand side.
\end{proposition}
\begin{proof}
It suffices to solve for a function $\widetilde{v}(v)$ such that $\widetilde{v}(1) = 1$, $C^{-1} \leq \f{\ud \widetilde{v}}{\ud v} \leq C$, where the constant $C > 0$ may depend on $E$, $s$, $M$, $\e$, $r_{0}$ and $\rd_{U} r(0, 1)$, and
\begin{equation}\label{condition.for.vtilde}
\f{e^{-2\kappa_+ r_+}(r_+-r_-)^{1+\f{\kappa_+}{\kappa_-}}}{r_+^2} e^{2\kappa_+ \widetilde{v}(v)} = - \rd_{U} r(0, v) \left(\f{\ud \widetilde{v}}{\ud v} (v)\right)^{-1}.
\end{equation}
Since we only work on $C_{-\infty}$, on which $U = 0$, we suppress the $U$-coordinate and simply write $\rd_{U} r (v) = \rd_{U} r(0, v)$ etc. Moreover, we omit the dependence of constants on $E$, $s$, $M$, $\e$, $r_{0}$ and $\rd_{U} r(0, 1)$.

\pfstep{Step~1: Estimating the decay of $\rd_v r$} 
By (b) and (c), it follows that $\lim_{v\to +\infty} (1-\mu) = 0$. Hence, using (a), 
\begin{equation}\label{dvr.to.0.in.rmk}
\lim_{v\to+\infty} \rd_v r = 0.
\end{equation}
We now want to get a quantitative decay estimate for $\rd_v r$. Using the Raychaudhuri equations in \eqref{eqn.Ray}, $-\f{\Omg^2}{4\rd_{U} r} = \f{\rd_{v} r}{1-\mu} = 1$ (cf.~\eqref{Hawking.def}) and \eqref{eq:EMSF-r-phi-m}, we obtain
\begin{equation} \label{ode.dvr.in.rmk}
\rd^2_v r - \f{2(\varpi-\f{{\bf e}^2}{r})}{r^2} (\rd_v r) = -r(\rd_v\phi)^2.
\end{equation}
By the method of integrating factors, we have
\begin{equation} \label{eqn.dvr.in.rmk}
	\rd_{v} r(v) = \lim_{v' \to \infty} e^{-\int_{v}^{v'} \f{2(\varpi-\f{{\bf e}^2}{r})}{r^2} (v'') \, \ud v''} \rd_{v} r(v')
	+ \int_{v}^{\infty} e^{-\int_{v}^{v'} \f{2(\varpi-\f{{\bf e}^2}{r})}{r^2} (v'') \, \ud v''} r (\rd_{v} \phi)^{2}(v') \, \ud v'.
\end{equation}
By \eqref{dvr.to.0.in.rmk} and $\lim_{v \to \infty} \f{2(\varpi-\f{{\bf e}^2}{r})}{r^2} (v) = 2 \kpp_{+}$, where the latter follows from the conditions (b) and (c), it follows that the first term on the right-hand side vanishes. 

To further analyze \eqref{eqn.dvr.in.rmk}, we need some information regarding $r$ and $\varpi$. For $r$, we note the preliminary lower and upper bounds $r_{0} \leq r \leq M + \sqrt{M^{2} - \e^{2}}$. For $\varpi$, by \eqref{eq:EMSF-r-phi-m}, (c) and (d), we have
\begin{equation} \label{m.converge.in.rmk}
	M - \varpi(v) = \frac{1}{2} \int_{v}^{\infty} r^{2} (\rd_{v} \phi)^{2}(v') \, \ud v' \leq C v^{-2s + 1}.
\end{equation}

Consider the interval $J = \set{v \geq 1 : \f{2(\varpi-\f{{\bf e}^2}{r})}{r^2} (v') > \kpp_{+} \hbox{ for all } v' \geq v}$, which is nonempty since $\f{2(\varpi-\f{{\bf e}^2}{r})}{r^2} \to 2 \kpp_{+}$ as $v \to \infty$. On $J$, \eqref{eqn.dvr.in.rmk} implies
\begin{equation*}
	\rd_{v} r(v)
	\leq \int_{v}^{\infty} e^{-\kpp_{+}(v' - v)} r (\rd_{v} \phi)^{2}(v') \, \ud v'.
\end{equation*}
Then by (d) and an argument similar to Lemma~\ref{lemma.1}, we have
\begin{equation} \label{dvr.decay.J.in.rmk}
	0 \leq \rd_{v} r(v) \leq C v^{-2s} \quad \hbox{ on } J,
\end{equation}
and as a consequence
\begin{equation} \label{r.converge.in.rmk}
\abs{M + \sqrt{M^{2} - \e^{2}} - r(v)} \leq C v^{-2s +1} \hbox{ on } J,
\end{equation}
where $C$ is independent of $J$.

Bootstrapping the \eqref{m.converge.in.rmk} and \eqref{r.converge.in.rmk}, we may find $v_{0}$ that depends only on $E$, $s$, $M$ and $\e$ such that $[v_{0}, \infty) \subseteq J$. Estimating $\rd_{v} r$ on the interval $[1, v_{0}]$ by simply applying Gr\"onwall's inequality to \eqref{ode.dvr.in.rmk} (here we need to use the lower bound $r_{0}$ for $r$), we obtain the following generalization of \eqref{dvr.decay.J.in.rmk} for all $v \geq 1$:
\begin{equation} \label{dvr.decay.in.rmk}
	0 \leq \rd_{v} r(v) \leq C v^{-2s}.
\end{equation}

\pfstep{Step~2: Estimates for $\rd_U r$} Starting with the first equation in \eqref{eq:EMSF-r-phi-m}, then using (a) and (e), we obtain
\begin{equation}\label{dvr.in.rmk}
\rd_v (\log (- \rd_U r)) = \f{2(\varpi-\f{{\bf e}^2}{r})}{r^2}.
\end{equation}
Differentiating in $v$ and using \eqref{eq:EMSF-r-phi-m} then yields
\begin{equation}\label{dv2r.in.rmk}
\rd_v^2 (\log (- \rd_U r)) = (\rd_v\phi)^2 -\f {2(\rd_v r)}{r^3} (2\varpi - \f{3{\bf e}^2}{r}).
\end{equation}
Using (d) and \eqref{dvr.decay.in.rmk}, and then integrating \eqref{dv2r.in.rmk}, we obtain
$$\left| \rd_v (\log (- \rd_U r))(v) - \lim_{v'\to +\infty} \rd_v (\log (- \rd_U r))(v')\right| \leq C v^{-2s}.$$
By \eqref{dvr.in.rmk} and the conditions (b) and (c), we obtain that $\lim_{v'\to +\infty} \rd_v (\log (- \rd_U r))(v') = 2\kappa_+$ so that
\begin{equation}\label{dvdur.in.rmk}
\left| \rd_v (\log (- \rd_U r))(v) - 2\kappa_+\right| \leq C v^{-2s}.
\end{equation}
As a consequence, there exist a constant $c_{0} > 0$ and $g(v)$ such that 
\begin{equation} \label{dur.eq.in.rmk}
\frac{-\rd_{U} r (v)}{-\rd_{U} r(1)}= c_{0} e^{-2 \kappa_+ v} (1+g(v)), 
\end{equation}
where the following bounds hold for all $v \geq 1$:
\begin{equation}\label{dur.bd.in.rmk}
\abs{g(v)} \leq C v^{-2s+1}.
\end{equation}

\pfstep{Step~3: Constructing the desired $\widetilde{v}$} We now return to the construction of $\widetilde{v}(v)$. Using \eqref{dur.eq.in.rmk} we may rewrite \eqref{condition.for.vtilde} as
\begin{equation} \label{eq.for.vtilde}
	e^{2 \kpp_{+} \widetilde{v}} \f{\ud \widetilde{v}}{\ud v} 
	= c_{1} e^{2 \kpp_{+} v}(1 + g(v)),
\end{equation}
for some constant $c_1>0$ depending on $c_0$, $(-\rd_{U} r)(1)$, $r_{\pm}$ and $\kappa_{\pm}$. The ODE \eqref{eq.for.vtilde} is explicitly solvable with $\tilde{v}(1) = 1$ by separation of variables as follows:
\begin{equation}
	\widetilde{v}(v) = v + \frac{1}{2 \kpp_{+}} \log c_{1} + \frac{1}{2 \kpp_{+}} \log \left(1 + \frac{1 - c_{1}}{c_{1}} e^{2 \kpp_{+}(1 - v)} + 2 \kpp_{+} \int_{1}^{v} e^{2 \kpp_{+}(v' - v)} g(v') \, \ud v' \right).
\end{equation}
From this formula, it is clear that $\widetilde{v}$ exists for all $v \geq 1$. Moreover, note that the magnitude of the $v'$-integral inside the logarithm may be bounded by $C v^{-2s +1}$ using an argument similar to Lemma~\ref{lemma.1}. It follows that
\begin{equation*}
	\abs{\widetilde{v}(v) - v - (2 \kpp_{+})^{-1} \log c_{1}} \leq C v^{-2s + 1}, 
\end{equation*}
which, in combination with \eqref{eq.for.vtilde}, yields the desired bound for $\frac{\ud \widetilde{v}}{\ud v}$. \qedhere
\end{proof}

\section{Blow up of Christoffel symbols} \label{sec:christoffel}
In Theorems~\ref{final.blow.up.step} and \ref{thm.nonpert} (see, in particular, \eqref{blow.up.interior} and \eqref{eq:nonpert-blowup-phi}), we have seen that the scalar field fails to be in $W^{1,2}_{loc}$ near every point of $\CH_{1}$ in the $C^{0}$ extension constructed in Theorems~\ref{main.theorem.C0.stability} and \ref{thm.nonpert}, respectively, provided that \eqref{final.blow.up.step.assumption} holds on $\EH_{1}$. In this section of the appendix, we investigate in more detail the geometry near $\CH_{1}$ under the same assumptions, and show that certain Christoffel symbols blow up in any $L^{p}_{loc}$ for $p > 1$. More precisely, our goal is to prove the following result:
\begin{proposition} \label{prop:christoffel}
Let $(u, V)$ be the null coordinates constructed in Theorem~\ref{thm.nonpert}, with respect to which the solution extends continuously to $\CH_{1} \setminus \set{p_{\CH_{1}}}$. If the lower bound \eqref{final.blow.up.step.assumption} holds on $\EH_{1}$, then for every $u \in (-\infty, u_{\CH_{1}})$ and $p > 1$, we have 
\begin{equation} \label{eq:christoffel}
	\int_{0}^{1} \abs{(\sgm_{\bbS^{2}}^{-1})^{AB}\Gmm^{u}_{AB}}^{p} \Omg^{2}(u,V) \, d V = \infty.
\end{equation}
Moreover, we have $\sup_{V\in [0,1)} \abs{(\sgm_{\bbS^{2}}^{-1})^{AB}\Gmm^{u}_{AB}}(V) = \infty$.
\end{proposition}
Recall that, in any null coordinates, we have $\Gmm^{u}_{AB} = 2 \Omg^{-2} r \rd_{V} r (\sgm_{\bbS^{2}})_{AB}$. In view of continuous extendibility of $r$, the blow up of the supremum already follows from \eqref{blow.up.interior.dvr} and \eqref{eq:nonpert-blowup-dvr}. In what follows, we exploit further the Raychaudhuri equation, \eqref{blow.up.interior} and bounds proved in Section~\ref{sec.nonpert} to upgrade this blow-up statement.
\begin{proof}
Fix $p\in (1,\infty)$. We begin with a few reductions. First, by the above formula for $\Gmm^{u}_{AB}$, it suffices to show that
\begin{equation} \label{eq:christoffel-goal}
	\int_{0}^{1} \left(-\frac{\rd_{V} r}{\Omg^{2}}\right)^{p} \Omg^{2} (u, V) \, d V = \infty \quad \hbox{ for every } u \in (-\infty, u_{\CH_{1}}).
\end{equation}
In view of the bounds $C^{-1} \leq \Omg^{2}(u, V) \leq C$ and $\abs{\rd_{u} (r \rd_{V} r)}(u, V) \leq C$, which hold in any rectangle $\calR = \set{(u, V) : u_{0} \leq u \leq u_{1}, \, 0 \leq V < 1}$ for some $0 < C < \infty$ (see \eqref{eq:nonpert-dvr-lip} in the proof of Lemma~\ref{lem.nonpert-blowup}), it suffices to prove \eqref{eq:christoffel-goal} under the additional assumption that $u$ lies in the perturbative region, i.e., $u < u_{s}$ where $u_{s}$ is as in Theorem~\ref{main.theorem.C0.stability}. Furthermore, note that both 
$\f{\rd_{V} r}{\Omg^{2}}$ and $\Omg^2 \, dV$ on the LHS of \eqref{eq:christoffel-goal} are invariant under any change of the null coordinate $V$. Hence, using the $(u, v)$ coordinates in Theorem~\ref{main.theorem.C0.stability}, we see that it suffices to prove
\begin{equation} \label{eq:christoffel-goal-v}
	\int_{1}^{\infty} \left(-\frac{\rd_{v} r}{\Omg^{2}}\right)^{p} \Omg^{2} (u, v) \, d v = \infty \quad \hbox{ for every } u \in (-\infty, u_{s}), 
\end{equation}
where $\Omg^{2}$ is now defined with respect to the $(u, v)$ coordinates.

In the remainder of the proof, we fix a value of $u < u_{s}$ and omit the dependence of constants on $u$. Let $\eta > 0$ be a parameter to be fixed later (and we allow the constants $C$ below to depend on $\eta$). By the integrated blow up statement \eqref{blow.up.interior}, we may find a sequence $v_{k} \in \bbN \cap [2, \infty)$ such that $v_{k}$ increases to $\infty$ and
\begin{equation}\label{whatever}
	\int_{v_{k}-1}^{v_{k}} \log_{+}^{\alp_{0}} (\frac{1}{\Omg}) (\rd_{v} \phi)^{2} (u, v) \, d v \geq e^{- \eta v_{k}}.
\end{equation}
Next, by the Raychaudhuri equation \eqref{eqn.Ray}, for any $v \geq v_{k}$ we have
\begin{equation*}
	- \frac{\rd_{v} r}{\Omg^{2}}(u, v) \geq \int_{v_{k}-1}^{v_{k}} \frac{r}{\Omg^{2}} (\rd_{v} \phi)^{2} (u, v') \, d v' - \frac{\rd_{v} r}{\Omg^{2}}(u, v_{k}-1).
\end{equation*}
By \eqref{blow.up.interior.dvr}, we may assume that the last term on the right-hand side is nonnegative (and thus can be dropped) by taking $v_{1}$ sufficiently large. To estimate the first term, we use \eqref{whatever} and also recall from Theorem~\ref{main.theorem.C0.stability} (in particular, the bounds for $r - r_{RN}$ and $\log \Omg - \log \Omg_{RN}$) that $C^{-1} \leq r \leq C$ and $C^{-1} e^{-  \kpp_{-} v} \leq \Omg \leq C e^{-  \kpp_{-} v}$. Hence,
\begin{equation*}
- \frac{\rd_{v} r}{\Omg^{2}}(u, v) \geq C^{-1} e^{(2 \kpp_{-} - \eta) v_{k}} v_{k}^{-\alp_{0}} \quad \hbox{ for every } v \geq v_{k}.
\end{equation*}
It follows that
\begin{align*}
\int_{1}^{\infty} \left(-\frac{\rd_{v} r}{\Omg^{2}}\right)^{p} \Omg^{2} (u, v) \, d v
\geq \sum_{k =1}^{\infty} \int_{v_{k}}^{v_{k}+1} \left(-\frac{\rd_{v} r}{\Omg^{2}}\right)^{p} \Omg^{2} (u, v) \, d v
\geq C^{-1} \sum_{k =1}^{\infty} e^{2 (p-1) \kpp_{-} v_{k} - p \eta v_{k}} v_{k}^{-p \alp_{0}}.
\end{align*}
Since $p > 1$, choosing $\eta > 0$ sufficiently small ensures that each summand on the right-hand side is uniformly bounded from below, which in turn makes the sum infinite as desired. \qedhere
\end{proof}

\bibliographystyle{hplain}
\bibliography{SCCInterior}

\begin{thebibliography}{10}

\bibitem{AAG}
Y.~Angelopoulos, S.~Aretakis, and D.~Gajic.
\newblock {Late-time asymptotics for the wave equation on spherically
  symmetric, stationary spacetimes}.
\newblock {\em preprint}, 2016, arXiv:1612.01566.

\bibitem{AAG2}
Y.~Angelopoulos, S.~Aretakis, and D.~Gajic.
\newblock A vector field approach to almost-sharp decay for the wave equation
  on spherically symmetric, stationary spacetimes.
\newblock {\em preprint}, 2016, arXiv:1612.01565.

\bibitem{Barack}
L.~Barack.
\newblock Late time decay of scalar, electromagnetic, and gravitational
  perturbations outside rotating black holes.
\newblock {\em Phys. Rev. D (3)}, 61(2):024026, 27, 2000.

\bibitem{CH}
S.~Chandrasekhar and J.~B. Hartle.
\newblock On crossing the {C}auchy horizon of a {R}eissner-{N}\"ordstrom
  black-hole.
\newblock {\em Proc. Royal Society of London A}, 384(1787):301--315, 1982.

\bibitem{CBG}
Y.~Choquet-Bruhat and R.~Geroch.
\newblock Global aspects of the {C}auchy problem in general relativity.
\newblock {\em Comm. Math. Phys.}, 14:329--335, 1969.

\bibitem{Christodoulou:1991yfa}
D.~Christodoulou.
\newblock {The formation of black holes and singularities in spherically
  symmetric gravitational collapse}.
\newblock {\em Commun. Pure Appl. Math.}, 44(3):339--373, 1991.

\bibitem{Chr.instab}
D.~Christodoulou.
\newblock The instability of naked singularities in the gravitational collapse
  of a scalar field.
\newblock {\em Annals of Mathematics}, 149(1):183--217, 1999.

\bibitem{Chr:CQG}
D.~Christodoulou.
\newblock On the global initial value problem and the issue of singularities.
\newblock {\em Classical Quantum Gravity}, 16(12A):A23--A35, 1999.

\bibitem{Chr}
D.~Christodoulou.
\newblock {\em The formation of black holes in general relativity}.
\newblock EMS Monographs in Mathematics. European Mathematical Society (EMS),
  Z\"urich, 2009, arXiv:0805.3880.

\bibitem{CGNS3}
J.L. Costa, P.M. Gir{\~a}o, J.~Nat\'ario, and J.D. Silva.
\newblock {On the global uniqueness for the {E}instein-{M}axwell-scalar field
  system with a cosmological constant. {P}art 3: {M}ass inflation and
  extendibility of the solutions}.
\newblock {\em preprint}, 2014, arXiv:1406.7261.

\bibitem{CGNS2}
J.L. Costa, P.M. Gir{\~a}o, J.~Nat\'ario, and J.D. Silva.
\newblock On the global uniqueness for the {E}instein-{M}axwell-scalar field
  system with a cosmological constant. {P}art 2. {S}tructure of the solutions
  and stability of the {C}auchy horizon.
\newblock {\em Comm. Math. Phys.}, 339(3):903--947, 2015, arXiv:1406.7253.

\bibitem{D1}
M.~Dafermos.
\newblock {Stability and instability of the {C}auchy horizon for the
  spherically symmetric {E}instein-{M}axwell-scalar field equations}.
\newblock {\em Annals of Math.}, 158(3):875--928, 2003.

\bibitem{DafTrapped}
M.~Dafermos.
\newblock Spherically symmetric spacetimes with a trapped surface.
\newblock {\em Classical Quantum Gravity}, 22(11):2221--2232, 2005.

\bibitem{D2}
M.~Dafermos.
\newblock {The interior of charged black holes and the problem of uniqueness in
  general relativity}.
\newblock {\em Comm. Pure Appl. Math.}, 58(4):445--504, 2005.

\bibitem{D3}
M.~Dafermos.
\newblock {Black holes without spacelike singularities}.
\newblock {\em Comm. Math. Phys.}, 332:729--757, 2014.

\bibitem{DL}
M.~Dafermos and J.~Luk.
\newblock {The interior of dynamical vacuum black holes {I}: The
  $C^0$-stability of the {K}err {C}auchy horizon}.
\newblock {\em preprint}, 2017, arXiv:1710.01722.

\bibitem{DafRen}
M.~Dafermos and A.~D. Rendall.
\newblock Inextendibility of expanding cosmological models with symmetry.
\newblock {\em Classical Quantum Gravity}, 22(23):L143--L147, 2005.

\bibitem{DRPL}
M.~Dafermos and I.~Rodnianski.
\newblock {A proof of Price's law for the collapse of a self-gravitating scalar
  field}.
\newblock {\em Invent. Math.}, 162(2):381--457, 2005.

\bibitem{DRNM}
M.~Dafermos and I.~Rodnianski.
\newblock {A new physical-space approach to decay for the wave equation with
  applications to black hole spacetimes}.
\newblock {\em XVIth {I}nternational {C}ongress on {M}athematical {P}hysics},
  pages 411--433, 2009.

\bibitem{DRS}
M.~Dafermos and I.~Rodnianski.
\newblock {The red-shift effect and radiation decay on black hole spacetimes}.
\newblock {\em {Comm. Pure Appl. Math.}}, 62(7):859--919, 2009,
  arXiv:gr-qc/0512119.

\bibitem{DafShl}
M.~Dafermos and Y.~Shlapentokh-Rothman.
\newblock {Time-Translation Invariance of Scattering Maps and Blue-Shift
  Instabilities on {K}err Black Hole Spacetimes}.
\newblock {\em preprint}, 2015, arXiv:1512.08260.

\bibitem{DS}
R.~Donninger and W.~Schlag.
\newblock Decay estimates for the one-dimensional wave equation with an inverse
  power potential.
\newblock {\em Int. Math. Res. Not. IMRN}, (22):4276--4300, 2010.

\bibitem{DSS1}
R.~Donninger, W.~Schlag, and A.~Soffer.
\newblock A proof of {P}rice's law on {S}chwarzschild black hole manifolds for
  all angular momenta.
\newblock {\em Adv. Math.}, 226(1):484--540, 2011.

\bibitem{DSS2}
R.~Donninger, W.~Schlag, and A.~Soffer.
\newblock On pointwise decay of linear waves on a {S}chwarzschild black hole
  background.
\newblock {\em Comm. Math. Phys.}, 309(1):51--86, 2012.

\bibitem{Fra}
A.~Franzen.
\newblock {Boundedness of massless scalar waves on {R}eissner-{N}ordstr\"om
  interior backgrounds}.
\newblock {\em preprint}, 2014, arXiv:1407.7093.

\bibitem{Gleeson}
E.~Gleeson.
\newblock Linear instability of the {R}eissner---{N}ordstr\"om {C}auchy
  horizon.
\newblock {\em MSc Thesis, Cambridge University}, 2016.

\bibitem{GSNS}
Y.~Gursel, V.~Sandberg, I.~Novikov, and A.~Starobinsky.
\newblock Evolution of scalar perturbations near the {C}auchy horizon of a
  charged black hole.
\newblock {\em Phys. Rev. D}, 19:413--420, 1979.

\bibitem{HawkEll}
S.~W. Hawking and G.~F.~R. Ellis.
\newblock {\em The Large Scale Structure of Space-Time (Cambridge Monographs on
  Mathematical Physics)}.
\newblock Cambridge University Press, 1975.

\bibitem{Hintz}
P.~Hintz.
\newblock {Boundedness and decay of scalar waves at the {C}auchy horizon of the
  {K}err spacetime}.
\newblock {\em preprint}, 2015, arXiv:1512.08003.

\bibitem{Hiscock}
W.~A. Hiscock.
\newblock Evolution of the interior of a charged black hole.
\newblock {\em Phys. Rev. Lett.}, 83A:110--112, 1981.

\bibitem{HodPiran1}
S.~Hod and T.~Piran.
\newblock Late-time evolution of charged gravitational collapse and decay of
  charged scalar hair. {I}.
\newblock {\em Phys. Rev. D (3)}, 58(2):024017, 6, 1998.

\bibitem{HodPiran2}
S.~{Hod} and T.~{Piran}.
\newblock {Late-time evolution of charged gravitational collapse and decay of
  charged scalar hair. {II}}.
\newblock {\em Phys. Rev. D}, 58(2):024018, July 1998, gr-qc/9801001.

\bibitem{HodPiran3}
S.~Hod and T.~Piran.
\newblock Late-time evolution of charged gravitational collapse and decay of
  charged scalar hair. {III}. {N}onlinear analysis.
\newblock {\em Phys. Rev. D (3)}, 58(2):024019, 6, 1998.

\bibitem{KomThe}
J.~Kommemi.
\newblock {The global structure of spherically symmetric charged scalar field
  spacetimes}.
\newblock {\em {PhD thesis}}, 2012.

\bibitem{Kommemi}
J.~Kommemi.
\newblock The global structure of spherically symmetric charged scalar field
  spacetimes.
\newblock {\em Comm. Math. Phys.}, 323(1):35--106, 2013.

\bibitem{LO.exterior}
J.~Luk and S.-J. Oh.
\newblock Strong cosmic censorship in spherical symmetry for two-ended
  asymptotically flat initial data {II}. the exterior of the black hole region.
\newblock {\em Preprint}, arXiv:1702.05716.

\bibitem{LO1}
J.~Luk and S.-J. Oh.
\newblock Quantitative decay rates for dispersive solutions to the
  {E}instein-scalar field system in spherical symmetry.
\newblock {\em Anal. PDE}, 8(7):1603--1674, 2015.

\bibitem{LO.instab}
J.~Luk and S.-J. Oh.
\newblock {Proof of linear instability of the Reissner-Nordstr{\"o}m Cauchy
  horizon under scalar perturbations}.
\newblock {\em Duke Math. J.}, 166(3):437--493, 2017, arXiv:1501.04598.

\bibitem{LS}
J.~Luk and J.~Sbierski.
\newblock Instability results for the wave equation in the interior of {K}err
  black holes.
\newblock {\em J. Funct. Anal.}, 271(7):1948--1995, 2016, arXiv:1512.08259.

\bibitem{McN}
J.~McNamara.
\newblock {Instability of black hole inner horizons}.
\newblock {\em Proc. Roy. Soc. Lon. A}, 358:499--517, 1978.

\bibitem{MTT}
J.~Metcalfe, D.~Tataru, and M.~Tohaneanu.
\newblock {Price's law on nonstationary spacetimes}.
\newblock {\em Adv. Math.}, 230(3):995--1028, 2012, arXiv:1104.5437.

\bibitem{Ori}
A.~Ori.
\newblock {Inner Structure of a Charged Black Hole: An Exact Mass-Inflation
  Solution}.
\newblock {\em {Phy. Rev. Lett.}}, 67(7):789--792, 1991.

\bibitem{OriSpin}
A.~Ori.
\newblock Evolution of linear gravitational and electromagnetic perturbations
  inside a {K}err black hole.
\newblock {\em Phys. Rev. D (3)}, 61(2):024001, 20, 2000.

\bibitem{Penrose.SCC}
R.~Penrose.
\newblock Gravitational collapse.
\newblock In C.~Dewitt-Morette, editor, {\em Gravitational Radiation and
  Gravitational Collapse, Vol. 64 of IAU Symposium}, pages 82--91. 1974.

\bibitem{PI1}
E.~Poisson and W.~Israel.
\newblock Inner-horizon instability and mass inflation in black holes.
\newblock {\em Phys. Rev. Lett.}, 63:1663--1666, 1989.

\bibitem{PI2}
E.~Poisson and W.~Israel.
\newblock Internal structure of black holes.
\newblock {\em Phys. Rev. D}, 41:1796--1809, 1990.

\bibitem{Ringstrom.book}
H.~Ringstr\"om.
\newblock {\em The {C}auchy problem in general relativity}.
\newblock ESI Lectures in Mathematics and Physics. European Mathematical
  Society (EMS), Z\"urich, 2009.

\bibitem{Sbi.2}
J.~Sbierski.
\newblock {On the initial value problem in general relativity and wave
  propagation in black-hole spacetimes}.
\newblock {\em {PhD thesis}}, 2014.

\bibitem{Sbie.C0}
J.~Sbierski.
\newblock {The $C^0$-inextendibility of the {S}chwarzschild spacetime and the
  spacelike diameter in {L}orentzian geometry}.
\newblock {\em preprint}, 2015, arXiv:1507.00601.

\bibitem{SP}
M.~Simpson and R.~Penrose.
\newblock Internal instability in a {R}eissner-{N}ordstr\"om black hole.
\newblock {\em Internat. J. Theoret. Phys.}, 7:183--197, 1973.

\bibitem{Ta}
D.~Tataru.
\newblock Local decay of waves on asymptotically flat stationary space-times.
\newblock {\em Amer. J. Math.}, 135(2):361--401, 2013, arXiv:0910.5290.

\end{thebibliography}
\end{document}